\documentclass[bezier,11pt]{article}
\RequirePackage{natbib}

\usepackage{amsmath}
\usepackage{amssymb}
\usepackage{latexsym}
\usepackage{ntheorem}
\usepackage{ifthen}
\usepackage{xcolor}
\usepackage{newcent}
\usepackage[applemac]{inputenc}
\usepackage[margin=3cm]{geometry}

\theoremstyle{plain}

\theoremheaderfont{\bf}\theorembodyfont{\sl}

\newtheorem{theorem}{Theorem}[section]
\newtheorem{proposition}[theorem]{Proposition}
\newtheorem{lemma}[theorem]{Lemma}

\theoremheaderfont{\it}\theorembodyfont{\rm}

\newtheorem{remark}[theorem]{Remark}

\theoremheaderfont{\bf}\theorembodyfont{\rm}

\newtheorem{definition}[theorem]{Definition}

\newtheorem{assumption}[theorem]{Assumption}

\theoremstyle{nonumberplain}

\theoremheaderfont{\bf}\theorembodyfont{\sl}

\newenvironment{proof}[1][]
{\ifthenelse{\equal{#1}{}}{\smallskip\noindent\textsl{Proof. }}{\smallskip
\noindent\textsl{Proof #1. }}}{\hfill$\Box$}

\def\Hc{{\cal H}}
\def\Fc{{\cal F}}
\def\Gc{{\cal G}}
\def\Dc{{\cal D}}

\def \a{\alpha}

\def \l{\lambda}

\def \f{\varphi}
\def \o{\omega}
\def \O{\Omega}

\begin{document}
\thispagestyle{empty}

\title{Non-concave optimal investment and no-arbitrage: a measure theoretical approach}
\author{Laurence Carassus   \\\footnotesize LMR, Universit\'e Reims Champagne-Ardenne, France
                            \\\footnotesize laurence.carassus@univ-reims.fr
          \and Romain~Blanchard \\\footnotesize LMR, Universit\'e Reims Champagne-Ardenne, France
          \\\footnotesize romain.blanchard@etudiant.univ-reims.fr \and
 Mikl\'os R\'asonyi   \\\footnotesize MTA Alfr\'ed R\'enyi Institute of Mathematics, Hungary
                            \\\footnotesize rasonyi@renyi.mta.hu         }

\maketitle

\begin{abstract}
We consider non-concave and non-smooth random utility functions with domain of definition equal to the non-negative half-line.
 We use a dynamic programming framework together with measurable selection arguments to establish both the no-arbitrage condition
characterization and the existence of an optimal portfolio in a (generically incomplete) discrete-time financial market model with finite time
horizon.
\end{abstract}

\section{Introduction}

We consider investors trading in a multi-asset and discrete-time financial market.
We revisit two classical problems: the characterization of no arbitrage and the maximisation of the  expected  utility
of the terminal wealth of an investor.

We consider a general random, possibly non-concave and non-smooth utility function $U$, defined on the non-negative half-line
(that can be ``$S$-shaped'' but our
results apply to a broader class of utility functions e.g. to piecewise concave ones) and  we provide sufficient conditions  which
guarantee the existence of an optimal strategy. Similar optimization problems constitute an area of intensive study in
recent years, see e.g. \cite{bensoussan} , \cite{hz}, \cite{jz}, \cite{cd}.

We are working in the setting of \cite{CRR} and remove certain restrictive hypothesis of \cite{CRR}.
Furthermore,  we use methods that are different from the ones in \cite{RS05}, \cite{RS06}, \cite{CR14} and \cite{CRR},
where similar multistep problems were treated. In contrast to the existing literature,
we propose to consider a probability space which is not necessarily complete.


We extend the paper of \cite{CRR} in several directions. First, we  propose an alternative integrability condition
(see Assumption \ref{PhiUX} and Proposition \ref{propufini}) to the rather restrictive one of \cite{CRR} stipulating
that $E^{-} U(\cdot,0)<\infty$. The property $U(0)=-\infty$ holds for a number of important (non-random and concave)
utility functions (logarithm, $-x^{\alpha}$ for $\alpha<0$). It is a rather natural requirement since it expresses the fear
of investor for defaulting ($i.e$ reaching 0). We also introduce  a new (weaker) version of the asymptotic elasticity assumption
(see Assumption \ref{ae}). In particular, Assumption \ref{ae} holds true for concave functions (see Remark \ref{concave}) and
therefore our result extends the one obtained in \cite{RS06} to random utility function and incomplete probability spaces.
Next, we do not require that the value function is finite for all initial wealth as it was postulated in \cite{CRR}; instead we only assumed the less restrictive and more tractable Assumption \ref{uisfinite}.
Finally, instead of using some Carath\'eodory utility function $U$  as in \cite{CRR} ($i.e$ function measurable in $\o$ and continuous in $x$),
we consider function which is  measurable in $\o$ and upper semicontinuous (usc in the rest of the paper) in $x$.  As $U$ is also non-decreasing,
we point out  that this implies that $U$ is jointly measurable in $(\o,x)$. Note that  in the case of complete sigma-algebra -$U$ is then a normal
integrand (see Definition 14.27 in \cite{rw} or Section 3 of Chapter 5 in \cite{Molchanov} as well as Corollary 14.34 in \cite{rw}). This will play an important role in the dynamic programming part to obtain
certain measurability properties.
Allowing non-continuous $U$ is unusual in the financial mathematics literature (though it is common in optimization). We highlight that this
generalisation has a potential to model investor's behaviour which can change suddenly after reaching a desired wealth level. Such a change can
be expressed by a jump of $U$ at the given level.

To solve our optimisation problem, we use dynamic programming as in \cite{RS05}, \cite{RS06}, \cite{CR14} and \cite{CRR} but here we propose a
different approach which provides simpler  proofs. As in \cite{Nutz}, we consider first a one period case with strategy in $\mathbb{R}^{d}$. Then
we use dynamic programming and measurable selection arguments, namely the Aumann Theorem (see, for example,  Corollary 1 in \citet{bv}) to
solve the multi-period problem.  Our modelisation of $(\Omega, \Fc,\mathfrak{F},P)$ is more general than in \cite{Nutz}  as there is only one probability
measure and we don't have to postulate Borel space or analytic sets. We also use the same methodology to reprove classical results on no-arbitrage
characterization (see \cite{RS05} and \cite{JS98}) in our context of possibly incomplete sigma-algebras.

We do not handle the case where the utility is defined on the whole real line (with a similar set of assumptions) as this would have overburdened
the paper. This is left for further research.

The paper is organized as follows: in section \ref{se2} we introduce our setup; section
\ref{secna} contains the main results on no-arbitrage; section \ref{se3} presents the main theorem on terminal wealth
expected  utility  maximisation; section \ref{seone}
establishes the existence of an optimal strategy for the one period case;
we prove our main theorem on utility  maximisation in section \ref{secmulti}.

Finally, section \ref{seannexe} collects some technical results and proofs as well as elements about random sets measurability.

\section{Set-up}\label{se2}

Fix a time horizon $T\in \mathbb{N}$ and let $(\Omega_{t})_{1\leq t\leq T}$ be a sequence of spaces and $(\Gc_t)_{1\leq t\leq T}$ be a sequence of sigma-algebra where $\Gc_t$ is a sigma-algebra on $\Omega_{t}$ for all $t=1, \ldots, T$.   For $t=1, \ldots, T$, we denote by $\Omega^t$ the $t$-fold Cartesian product
$$\Omega^t=\Omega_{1} \times \ldots \times \Omega_{t}.$$
An element of $\Omega^t$ will be denoted by $\o^t=(\o_1,\ldots,\o_t)$ for $(\o_1,\ldots,\o_t) \in \Omega_{1} \times \ldots \times \Omega_{t}$. We also denote by ${\Fc}_t$ the product sigma-algebra on $\Omega^t$
$${\Fc}_t=\Gc_{1} \otimes \ldots \otimes \Gc_{t}.$$
For the sake of simplicity we consider that the state $t=0$ is deterministic and set $\Omega^{0}:=\{\o_{0}\}$ and
$\mathcal{F}_{0}=\mathcal{G}_{0}=\{\emptyset, \Omega^{0}\}$. To avoid heavy notations we will omit the dependency in $\o_{0}$ in the rest of the paper.
We denote by $\mathfrak{F}$ the filtration $(\Fc_{t})_{0\leq t\leq T}$.

Let $P_1$ be a probability measure on $\Fc_1$ and  $q_{t+1}$ be a stochastic kernel on $\Gc_{t+1} \times \Omega^{t}$ for $t=1,\ldots, T-1$.  Namely we assume that for all $\o^t \in \Omega^t$, $B \in \Gc_{t+1} \to q_{t+1}(B|\o^t)$ is a probability measure on $\Gc_{t+1}$ and
for all $B \in \Gc_{t+1}$, $\o^t \in \Omega^t \to q_{t+1}(B|\o^t)$ is $\Fc_t$-measurable.
Here we DO NOT assume that $\Gc_1$ contains the null sets of $P_1$ and that $\Gc_{t+1}$ contains the null sets of $q_{t+1}(.|\o^t)$ for
all  $\o^t \in \Omega^t$.
Then we define for  $A \in \Fc_t$ the probability $P_t$  by Fubini's Theorem for stochastic kernel (see Lemma \ref{fubini0}).
\begin{eqnarray}
\label{decompoP}
P_t(A)=\int_{\Omega_{1}}\int_{\Omega_{2}}\cdots\int_{\Omega_{t}}1_{A}(\omega_{1}, \ldots, \omega_{t})q_{t}(d\omega_{t}|\omega^{t-1})\cdots q_{2}(d\omega_{2}|\omega^{1})P_{1}(d\omega_{1}).
\end{eqnarray}
Finally $(\Omega, \Fc, \mathfrak{F},P):=(\Omega^{T}, \Fc_{T}, \mathfrak{F},P_T)$ will be our basic measurable space. The expectation under $P_t$ will  be denoted by $E_{P_t}$ ; when $t=T$, we simply write $E$.
\begin{remark}
If we choose for $\Omega$ some Polish space, then any probability measure $P$ can be decomposed in the form of
\eqref{decompoP} (see the measure decomposition theorem in \cite{dm}  III.70-7).
\end{remark}
From now on the positive (resp. negative) part of some number or
random variable $X$ is denoted by $X^+$ (resp. $X^-$). We will also write $f^{\pm}(X)$ for $\left(f(X)\right)^{\pm}$ for any random
variable $X$ and (possibly random) function $f$.\\
In the rest of the paper we will use generalised integral: for some $f_{t}: \Omega^{t} \to \mathbb{R} \cup \{\pm \infty\}$, $\mathcal{F}_{t}$-measurable, such that  $\int_{\Omega^{t}} f^{+}_{t}(\o^{t})P_{t}(d\o^{t})<\infty$ or  $\int_{\Omega^{t}} f^{-}_{t}(\o^{t})P_{t}(d\o^{t})<\infty$, we define
$$\int_{\Omega^{t}} f_{t}(\o^{t})P_{t}(d\o^{t}):= \int_{\Omega^{t}} f^{+}_{t}(\o^{t})P_{t}(d\o^{t})-\int_{\Omega^{t}} f^{-}_{t}(\o^{t})P_{t}(d\o^{t}),$$
where the equality holds in $\mathbb{R} \cup \{\pm \infty\}$.
We refer to Lemma \ref{fubini0}, Definition \ref{DefInt} and Proposition \ref{fubinirem} of the Appendix for more details and properties. In particular, if $f_{t}$ is non-negative or if $f_{t}$ is such that  $\int_{\Omega^{t}} f^{+}_{t}(\o^{t})P_{t}(d\o^{t})<\infty$ (this will be the two cases  of interest in the paper) we can apply Fubini's Theorem  \footnote{From now, we call Fubini's theorem the Fubini theorem for stochastic kernel (see eg  Lemma \ref{fubini0}, Proposition \ref{fubinirem}).} and we have
$$\int_{\Omega^{t}} f_{t}(\o^{t})P_{t}(d\o^{t})=\int_{\Omega_{1}}\int_{\Omega_{2}}\cdots\int_{\Omega_{t}}f_{t}(\omega_{1}, \ldots, \omega_{t})q_{t}(d\omega_{t}|\omega^{t-1})\cdots q_{2}(d\omega_{2}|\omega^{1})P_{1}(d\omega_{1}),$$
where the equality holds in $[0,\infty]$ if $f_{t}$ is non-negative  and in $[-\infty,\infty)$ if $\int_{\Omega^{t}} f^{+}_{t}(\o^{t})P_{t}(d\o^{t})<\infty$.\\
Finally, we give some notations about completion of the probability space $(\Omega^t, \Fc_t,P_t)$ for some $t \in \{1,\ldots,T\}$. We will denote by ${\cal N}_{P_{t}}$ the set of $P_t$ negligible sets of $\Omega^t$ $i.e$ ${\cal N}_{P_{t}}=\{N \subset \Omega^t, \; \exists M \in \Fc_t, \; N \subset M \mbox{ and } P_t(M)=0\}$. Let $\overline{\Fc}_{t}=\{A \cup N, A \in \mathcal{F}_{t}, N \in {\cal N}_{P_{t}}\}$ and
$\overline{P}_{t}(A\cup N)={P}_{t}(A)$ for $A\cup N \in \overline{\Fc}_{t}$. Then it is well known that $\overline{P}_{t}$ is a measure on $\overline{\Fc}_{t}$ which coincides with ${P}_{t}$ on ${\Fc}_{t}$, that $(\Omega^t, \overline \Fc_t,\overline P_t)$ is a complete probability space and that $\overline P_t$ restricted to ${\cal N}_{P_{t}}$ is equal to zero.\\

For $t=0,\ldots, T-1$, let $\Xi_t$ be  the set of $\Fc_t$-measurable random variables mapping $\Omega^t$ to $\mathbb{R}^d$.\\

The following  lemma makes the link between conditional expectation and kernel. To do that, we introduce $\Fc^T_t$, the filtration on $\Omega^{T}$ associated to $\Fc_t$,   defined by
$$\Fc^T_t=\Gc_{1} \otimes \ldots \otimes \Gc_{t} \otimes \{\emptyset,\Omega_{t+1} \} \ldots \otimes \{\emptyset,\Omega_{T} \}.$$
Let $\Xi^T_t$ be the set of $\Fc^T_t$-measurable random variables from $\Omega^T$ to $\mathbb{R}^d$.
Let $X_t$ $:$ $\Omega^T$ $ \to$ $\Omega_t$, $X_t(\o_1,\ldots,\o_T)=\o_t$ be the  coordinate mapping corresponding to $t$.
Then $\Fc^T_t=\sigma(X_1,\ldots,X_t)$. So $h \in \Xi^T_t$ if and only if
there exists some $g  \in \Xi_t$ such that
$h=g(X_1,\ldots,X_t)$. This implies that $h(\o^T)=g(\o^t)$. For ease of notation we will identify $h$ and $g$ and also $\mathcal{F}_{t}$, ${\mathcal{F}}^T_{t}$, $\Xi_t$ and $\Xi^T_t$.
\begin{lemma}
\label{lienespcond}
Let $0\leq s \leq t \leq T$. Let $h \in \Xi_t$ such that $\int_{\O^{t}} h^{+}dP_{t}<\infty$ then
\begin{eqnarray*}
E(h |\Fc_s) & = & \f(X_1,\ldots,X_s) \, P_s \ a.s.\\
\f(\o_{1},\ldots,\o_s) & = & \int_{\Omega_{s+1} \times \ldots \times \Omega_{t}} h(\o_1,\ldots,\o_s,\o_{s+1},\ldots\o_t)q_{t}(\o_{t}|\o^{t-1}) \ldots q_{s+1}(\o_{s+1}|\o^s).
\end{eqnarray*}
\end{lemma}
\begin{proof}
For the sake of completeness, the proof is reported  in Section \ref{proofofres} of the Appendix.
\end{proof}\\

Let $\{S_t,\ 0\leq t\leq T\}$ be a $d$-dimensional $\Fc_t$-adapted process
representing the  price of $d$ risky securities in the
financial market in consideration.
There exists also a riskless
asset for which we assume a constant price  equal to $1$, for the sake of simplicity. Without this assumption, all the developments
below could be carried out using discounted prices.
The notation $\Delta
S_t:=S_t-S_{t-1}$ will often be used.
If $x,y\in\mathbb{R}^d$ then
the concatenation $xy$ stands for their scalar product. The symbol $|\cdot|$ denotes the Euclidean norm
on $\mathbb{R}^d$ (or on $\mathbb{R})$.

Trading
strategies are represented by $d$-dimensional predictable
processes $(\phi_t)_{1\leq t\leq T}$, where $\phi_t^i$ denotes the
investor's holdings in asset $i$ at time $t$; predictability means
that $\phi_t\in\Xi_{t-1}$. The family of all predictable trading
strategies is denoted by $\Phi$.

We assume that trading is self-financing. As the riskless asset's price is constant $1$, the value at time $t$ of a portfolio $\phi$ starting from
initial capital $x\in\mathbb{R}$ is given by
$$
V^{x,\phi}_t=x+\sum_{i=1}^t  \phi_i \Delta S_i.
$$

\section{No-arbitrage condition}
\label{secna}
The following absence of arbitrage condition or NA condition is standard, it is
equivalent to the existence of a risk-neutral measure in discrete-time markets with finite horizon, see e.g. \cite{dmw}.

\medskip

(NA) {\em If $V^{0,\phi}_T\geq 0$ $P$-a.s. for some $\phi \in\Phi$ then
$V^{0,\phi}_T=0$ $P$-a.s.}

\medskip

\begin{remark}
It is proved in Proposition 1.1 of \citet{RS06} that (NA) is equivalent to the no-arbitrage assumption which stipulates that no investor should be allowed to make a profit out of nothing and without risk, even with a budget constraint:
for all $x_{0}\geq0$  if $\phi\in\Phi$ is such that with $V_T^{x_{0},\phi} \geq x_{0}$  a.s., then $V_T^{x_{0},\phi} =x_{0}$ a.s.
\end{remark}
We now provide classical tools and results about the (NA) condition and its ``concrete"  local characterization, see Proposition \ref{AOAmulti}, that we will use in the rest of the paper. We start with the  set $D^{t+1}$ (see Definition \ref{defrs}) where $D^{t+1}(\o^{t})$ is the smallest affine subspace of $\mathbb{R}^{d}$ containing the support of the distribution of $\Delta S_{t+1}(\o^{t},.)$ under $q_{t+1}(.|\o^{t})$. If $D^{t+1}(\o^{t}) = \mathbb{R}^d$ then,
intuitively, there are no redundant assets. Otherwise, for $\phi_{t+1} \in \Xi_{t}$, one may
always replace $\phi_{t+1}(\o^{t},\cdot)$ by its orthogonal projection
$\phi^{\perp}_{t+1}(\o^{t},\cdot)$ on $D^{t+1}(\o^{t})$ without changing the portfolio value since
$\phi_{t+1}(\o^{t}) \Delta S_{t+1}(\o^{t},\cdot)=\phi^{\perp}_{t+1} (\o^{t}) \Delta S_{t+1}(\o^{t},\cdot)$, $q_{t+1}(\cdot|\o^t)$ a.s., see Remark \ref{proj} and Lemma \ref{CfNutz}  below as well as Remark 9.1 of \cite{fs}.

\begin{definition}
\label{defrs}
Let $(\Omega, \mathcal{F})$ be a measurable space and $(T, \mathcal{T})$ a topological space.
A random set $R$ is a set valued function that assigns to each $\o \in \Omega$ a subset  $R(\omega)$ of $T$. We write $R: \Omega \twoheadrightarrow T$. We say that $R$ is measurable if for any open set  $O \in T$  $\{\omega \in \O,\; R(\omega) \cap O \neq \emptyset\} \in \mathcal{F}$.
\end{definition}
\begin{definition}
\label{DefD}
Let  $0 \leq t \leq T$ be fixed.
We define the random set (see Definition \ref{defrs}) $\widetilde{D}^{t+1} : \Omega^{t} \twoheadrightarrow \mathbb{R}^{d}$  by
\begin{align}
\label{defd1}
\widetilde{D}^{t+1}(\o^{t}):=\bigcap \left\{ A \subset \mathbb{R}^{d},\; \mbox{closed}, \; q_{t+1}\left(\Delta S_{t+1}(\o^{t},.) \in A|\o^{t})=1\right)\right\}.
\end{align}
For $\o^t \in {\O}^t$, $\widetilde{D}^{t+1}(\o^t) \subset \mathbb{R}^d$ is the support of the distribution of $\Delta S_{t+1}(\o^t, \cdot)$ under $q_{t+1}(\cdot|\o^t)$. We also define the random set  $D^{t+1} : \Omega^{t} \twoheadrightarrow \mathbb{R}^{d}$ by
\begin{align}
\label{defd2}
D^{t+1}(\o^{t}):= \mbox{Aff} \left( \widetilde{D}^{t+1}(\o^{t}) \right),
\end{align}
where $\mbox{Aff}$ denotes the affine hull of a set.
\end{definition}
The following lemma establishes some important properties of  $\widetilde{D}^{t+1}$ and $D^{t+1}$ and in particular $Graph(D^{t+1}) \in \mathcal{F}_{t} \otimes \mathcal{B}(\mathbb{R}^{d})$. This result will be central in the proof of most of our results.
\begin{lemma}
\label{Dmeasurability}
Let  $0 \leq t \leq T$ be fixed. Let $\widetilde{D}^{t+1} : \Omega^{t} \twoheadrightarrow \mathbb{R}^{d}$ and $D^{t+1} : \Omega^{t} \twoheadrightarrow \mathbb{R}^{d}$ be the random sets defined in \eqref{defd1} and \eqref{defd2} of Definition \ref{DefD}.
Then $\widetilde{D}^{t+1}$ and $D^{t+1}$ are  both non-empty, closed-valued and  $\mathcal{F}_{t}$-measurable random sets (see Definition \ref{defrs}). In  particular, $Graph(D^{t+1}) \in \mathcal{F}_{t} \otimes \mathcal{B}(\mathbb{R}^{d})$.
\end{lemma}
\begin{proof}
The proof is reported in  Section \ref{proofofres} of the Appendix.
\end{proof}\\

In Lemma \ref{localVectorSpace}, which is used in the proof of Lemma \ref{localNA} for projection purposes,  we obtain a well-know result~:  for $\o^{t} \in \Omega^{t}$ fixed and  under a local version of (NA), $D^{t+1}(\o^t)$ is a vector subspace of $\mathbb{R}^{d}$ (see for instance Theorem 1.48 of  \cite{fs}). Then in  Lemma \ref{localNA} we prove that under the (NA) assumption,  for $P_{t}$ almost all $\o^{t}$, $D^{t+1}(\o^{t})$ is a vector subspace of $\mathbb{R}^{d}$. We also provide a local version of the (NA) condition (see \eqref{EqlocalNA2}). Note that Lemma \ref{localNA} is a direct consequence of  Proposition 3.3 in \cite{RS05} combined with Lemma \ref{lienespcond}  (see Remark \ref{RSOA}).  We propose   alternative proofs of  Lemmata \ref{localVectorSpace} and \ref{localNA} which are coherent with our framework and our methodology.
\begin{lemma}
\label{localVectorSpace}
Let  $0 \leq t \leq T$ and $\o^{t} \in \Omega^{t}$ be fixed. Assume that  for all $h \in D^{t+1}(\o^{t}) \backslash\{0\}$ $$q_{t+1}(h\Delta S_{t+1}(\o^t,\cdot)\geq 0|\o^t)<1.$$
Then $0 \in D^{t+1}(\o^{t})$ and the set $D^{t+1}(\o^{t})$ is actually a vector subspace of $\mathbb{R}^d$.
\end{lemma}
\begin{proof}
The proof is reported in  Section \ref{proofofres} of the Appendix.
\end{proof}\\

\begin{lemma}
\label{localNA}
Assume that the (NA) condition holds true. Then for all $0\leq t\leq T-1$, there exists a full measure set $\Omega_{NA1}^{t}$ such that for all $\omega^{t} \in \Omega_{NA1}^{t}$, $0 \in D^{t+1}(\o^t)$, $i.e$ $D^{t+1}(\o^t)$ is a vector space of  $\mathbb{R}^d$.
Moreover, for all $\omega^{t} \in \Omega_{NA1}^{t}$ and all $h \in \mathbb{R}^{d}$ we get that
\begin{align}
\label{EqlocalNA}
q_{t+1}(h\Delta S_{t+1}(\o^t,\cdot)\geq 0|\o^t)=1 &\Rightarrow q_{t+1}(h\Delta S_{t+1}(\o^t,\cdot)= 0|\o^t)=1.
\end{align}
In particular, if $\omega^{t} \in \Omega_{NA1}^{t}$ and $h \in D^{t+1}(\o^{t})$ we obtain that
\begin{align}
q_{t+1}(h\Delta S_{t+1}(\o^t,\cdot)\geq 0|\o^t)=1 &\Rightarrow h=0.
\label{EqlocalNA2}
 \end{align}
\end{lemma}
\begin{proof}
Let  $0\leq t \leq T$ be fixed. Recall that $\overline{\mathcal{F}}_{t}$ is the $P_{t}$-completion of $\mathcal{F}_{t}$ and that $\overline{P}_{t}$ is the (unique) extension of $P_{t}$ to  $\overline{\mathcal{F}}_{t}$.
We introduce  the following random set $\Pi^{t}$
\begin{align*}
\Pi^{t}:=\left\{\omega^{t} \in \Omega^{t}, \; \exists h \in D^{t+1}(\o^{t}), h \neq 0, \; q_{t+1}(h\Delta S_{t+1}(\o^t,\cdot)\geq 0|\o^t)=1\right\}.
\end{align*}
Assume for a moment that $\Pi^{t} \in \overline{\Fc}_{t}$ and that $\overline{P}_{t}(\Pi^{t})=0$ (this will be proven below).
Let $\omega^{t} \in \Omega^{t} \setminus \Pi^{t}$. The fact that $0 \in D^{t+1}(\o^{t})$ is a direct consequence of the definition of $\Pi^t$ and of Lemma \ref{localVectorSpace}.
We now prove \eqref{EqlocalNA}. Let $h \in \mathbb{R}^{d}$ be fixed such that $q_{t+1}(h \Delta S_{t+1}(\o^{t},\cdot)\geq 0|\o^{t})=1$.
We prove that $q_{t+1}(h \Delta S_{t+1}(\o^{t},\cdot) =0|\o^{t})=1$. If $h=0$ this is straightforward. If $h \in D^{t+1}(\o^{t})\setminus \{0\}$,  $\omega^{t} \in \Pi^{t}$ which is impossible.
Now if $h \notin D^{t+1}(\o^{t})$ and $h \neq 0$, let $h'$ be the orthogonal projection of $h$ on $D^{t+1}(\o^{t})$ (recall that since $\o^{t} \notin \Pi^{t}$  $D^{t+1}(\o^{t})$ is a vector subspace).
We first show that $q_{t+1}(h'\Delta S_{t+1}(\o^t,\cdot)\geq 0|\o^t)=1$.
Indeed, if it were not the case the set $B:=\{\omega_{t+1} \in \Omega_{t+1},\; h'\Delta S_{t+1}(\o^t,\o_{t+1})<0\}$ would verify $q_{t+1}(B|\o^{t})>0$. Set
\begin{align}
\label{lt}
L^{t+1}(\omega^{t}):=\left(D^{t+1}(\omega^{t})\right)^{\bot}.
\end{align}
As $(h-h') \in L^{t+1}(\o^{t})$ (recall that $D^{t+1}(\o^{t})$ is a vector subspace),  by Lemma \ref{CfNutz} the set $A:=\{\omega_{t+1} \in \Omega_{t+1},\; (h-h')\Delta S_{t+1}(\o^t,\o_{t+1})=0\}$ verify $q_{t+1}(A|\o^{t})=1$. We would therefore obtain that $q_{t+1}(A \cap B|\o^{t})>0$ which implies that $q_{t+1}(h \Delta S_{t+1}(\o^{t},.)\geq 0|\o^{t})<1$, a contradiction. Thus $q_{t+1}(h'\Delta S_{t+1}(\o^t,\cdot)\geq 0|\o^t)=1$.
If $h' \neq 0$ as $h' \in D^{t+1}(\o^{t})$, $\o^{t} \in \Pi^{t}$ which is again a contradiction. Thus  $h'=0$ and
as $A \cap \{h' \Delta S_{t+1}(\o^{t},\cdot) =0\} \subset \{h \Delta S_{t+1}(\o^{t},\cdot) =0\}$, $q_{t+1}(h \Delta S_{t+1}(\o^{t},\cdot) =0|\o^{t})=1$.\\
As $\Omega^{t} \setminus \Pi^{t} \in \overline{\mathcal{F}_{t}}$ there exists $\Omega^{t}_{NA1} \in \mathcal{F}_{t}$ and $N^{t} \in \mathcal{N}_{P_{t}}$ (the collection of negligible set of $(\O^t, P_t)$) such that $\Omega^{t}  \setminus  \Pi^{t}=\Omega^{t}_{NA1} \cup N^{t}$ and $P_{t}(\Omega^{t}_{NA1})=\overline{P}_{t}(\Omega^{t} \backslash\Pi^{t})=1$. Since $\Omega^{t}_{NA1} \subset \Omega^{t}  \setminus \Pi^{t}$, it follows that for all $\o^{t} \in \Omega^{t}_{NA1}$, $0 \in D^{t+1}(\o^{t})$ and for all $h \in \mathbb{R}^{d}$,
 \eqref{EqlocalNA} holds true.\\
We prove \eqref{EqlocalNA2}.  Assume now that $\omega^{t} \in \Omega_{NA1}^{t}$ and $h \in D^{t+1}(\o^{t})$ are such that
$q_{t+1}(h\Delta S_{t+1}(\o^t,\cdot)\geq 0|\o^t)=1$. Using  \eqref{EqlocalNA} and Lemma \ref{CfNutz} we get that $ h \in L^{t+1}(\omega^{t})$. So $h \in D^{t+1}(\o^{t})\cap L^{t+1}(\omega^{t})=\{0\}$ and \eqref{EqlocalNA2} holds true.

It remains to prove that $\Pi^{t} \in \overline{\Fc}_{t}$ and $\overline{P}_{t}(\Pi^{t})=0$. To do that
we introduce  the following random set $H^{t}: \Omega^{t} \twoheadrightarrow \mathbb{R}^{d}$
\begin{align*}
 H^{t}(\omega^{t}):=\left\{ h \in  D^{t+1}(\o^{t}),\; h \neq 0,\; q_{t+1}(h\Delta S_{t+1}(\o^t,\cdot)\geq 0|\o^t)=1\right\}.
\end{align*}
Then
\begin{eqnarray*}
\Pi^{t} & = & \left\{\omega^{t} \in \Omega^{t},\; H^t(\omega^{t})\neq \emptyset \right\}= \mbox{proj}_{|\Omega^{t}}Graph(H^{t})
 \end{eqnarray*}
since $Graph(H^{t})=\{(\o^t, h) \in \Omega^{t}\times\mathbb{R}^{d},\; h \in H^t(\omega^{t})\}$.

We prove now that $Graph(H^{t}) \in \Fc_{t} \otimes \mathcal{B}(\mathbb{R}^{d})$. Indeed, we can rewrite that
\begin{align*}
Graph(H^{t})= Graph(D^{t+1})  \bigcap \left\{(\o^{t},h) \in \Omega^{t}\times\mathbb{R}^{d},\; q_{t+1}(h\Delta S_{t+1}(\o^t,\cdot)\geq 0|\o^t)=1\right\} \bigcap \left(\Omega^{t} \times \mathbb{R}^{d} \backslash \{0\} \right).
\end{align*}

As from Lemma  \ref{LemmaA1}, $\left\{(\o^{t},h) \in \Omega^{t}\times\mathbb{R}^{d},\; q_{t+1}(h\Delta S_{t+1}(\o^t,\cdot)\geq 0|\o^t)=1\right\} \in \Fc_{t} \otimes \mathcal{B}(\mathbb{R}^{d})$
and from Lemma \ref{Dmeasurability}, $Graph(D^{t+1}) \in \Fc_{t} \otimes \mathcal{B}(\mathbb{R}^{d})$, we obtain that  $Graph(H^{t}) \in \Fc_{t} \otimes \mathcal{B}(\mathbb{R}^{d})$.  
The Projection Theorem (see for example Theorem 3.23 in \cite{CV77}) applies and 
$ \Pi^{t}=\{H^t\neq \emptyset\}=\mbox{proj}_{|\Omega^{t}}Graph(H^{t}) \in \overline{\Fc}_{t}$. From the Aumann Theorem (see Corollary 1 in \citet{bv}) there exists a $\overline{\Fc}_{t}$-measurable selector $\overline{h}_{t+1}: \Pi^{t} \to \mathbb{R}^{d}$ such that $\overline{h}_{t+1}(\o^t) \in H^{t}(\o^t)$ for every $\o^t \in \Pi^t$. We now  extend $\overline{h}_{t+1}$ on $\Omega^{t}$ by setting $\overline{h}_{t+1}(\o^{t})=0$ for $\o^{t} \in \Omega^{t}  \backslash \Pi^{t}$. It is clear that $\overline{h}_{t+1}$ remains ${\overline{\Fc}_{t}}$-measurable.
Applying Lemma \ref{completemes}, there exists $h_{t+1}: \Omega^{t} \to \mathbb{R}^{d}$ which is $\mathcal{F}_{t}$-measurable and satisfies  $h_{t+1}=\overline{h}_{t+1}$ $P_{t}$-almost surely. Then if  we set
\begin{align*}
\varphi(\omega^{t}) &=q_{t+1}(h_{t+1}(\o^{t}) \Delta S_{t+1}(\o^{t},.) \geq 0 | \o^{t}),\\
\overline{\varphi}(\omega^{t})&=q_{t+1}(\overline{h}_{t+1}(\o^{t}) \Delta S_{t+1}(\o^{t},.) \geq 0 | \o^{t}),
\end{align*}
we get from Proposition \ref{LemmaA1}  that $\varphi$ is $\mathcal{F}_{t}$-measurable and  from Proposition \ref{fubiniext} $iii)$  that $\overline{\varphi}$ is $\overline{\mathcal{F}}_{t}$-measurable. Furthermore as $\{\o^{t} \in \O^t,\; \varphi(\o^{t}) \neq \overline{\varphi}(\o^{t})\} \subset \{\o^{t}\in \O^t,\; h_t(\o^{t}) \neq \overline{h}_{t+1}(\o^{t})\} $, $\varphi=\overline{\varphi}$ $P_{t}$-almost surely. This implies that  $\int_{\Omega^{t}} \overline{\varphi} d \overline{P}_{t}=\int_{\Omega^{t}} {\varphi} d P_{t}.$
Now we define the predictable process $(\phi_t)_{1\leq t\leq T}$ by $\phi_{t+1}={h}_{t+1}$ and $\phi_i=0$ for $i\neq t+1$. Then
\begin{eqnarray*}
P(V^{0,\phi}_{T}\geq 0)&= & P(h_{t+1} \Delta S_{t+1} \geq 0)=P_{t+1}(h_{t+1} \Delta S_{t+1} \geq 0)\\
&= &  \int_{\Omega^{t}} {\varphi}(\omega^{t})  P_{ t}(d\omega^{t})= \int_{\Omega^{t}} \overline{\varphi}(\omega^{t})  \overline{P}_{ t}(d\omega^{t}) \\
& = & \int_{\Pi^{t}} q_{t+1}\left(\overline{h}_{t}(\omega^{t})\Delta S_{t+1}(\o^t,\cdot)\geq 0|\o^t\right) \overline{P}_{ t}(d\omega^{t}) + \\ & &  \int_{\Omega^{t}\setminus{\Pi^{t}}} q_{t+1}\left(0  \times \Delta S_{t+1}(\o^t,\cdot)\geq 0|\o^t\right) \overline{P}_{ t}(d\omega^{t})\\
& = & \overline{P}_t(\Pi^{t}) +\overline{P}_t(\Omega^{t}\setminus{\Pi^{t}})=1,
\end{eqnarray*} where we have used that if $\o^t \in \Pi^{t}$, $\overline{h}_{t+1}(\omega^{t}) \in H^t(\o^t)$ and  otherwise $\overline{h}_{t+1}(\omega^{t})=0$. With the same arguments we obtain that
 \begin{align*}
 P(V^{0,\phi}_{T}> 0)&=P_t(h_{t+1} \Delta S_{t+1}>0)\\
& =\int_{\Pi^{t}} q_{t+1}\left(\overline{h}_{t+1}(\omega^{t})\Delta S_{t+1}(\o^t,\cdot)> 0|\o^t\right) \overline{P}_{ t}(d\omega^{t}) +\int_{\Omega^{t}\setminus{\Pi^{t}}} q_{t+1}\left(0> 0|\o^t\right) \overline{P}_{ t}(d\omega^{t})\\
& =\int_{\Pi^{t}} q_{t+1}\left(\overline{h}_{t+1}(\omega^{t})\Delta S_{t+1}(\o^t,\cdot)> 0|\o^t\right) \overline{P}_{ t}(d\omega^{t}).
\end{align*}
Let $\omega^{t} \in \Pi^{t}$ then $q_{t+1}\left(\overline{h}_{t+1}(\omega^{t})\Delta S_{t+1}(\o^t,\cdot)>0 |\o^t\right)>0$. Indeed, if it is not the case then \\
$q_{t+1}\left(\overline{h}_{t+1}(\omega^{t})\Delta S_{t+1}(\o^t,\cdot)\leq 0 |\o^t\right)=1$. As $\omega^{t} \in \Pi^{t}$, $\overline{h}_{t+1}(\omega^{t}) \in D^{t+1}(\o^t)$ and  $q_{t+1}\left(\overline{h}_{t+1}(\omega^{t})\Delta S_{t+1}(\o^t,\cdot)\geq 0|\o^t\right)=1$,  Lemma \ref{CfNutz} applies and $\overline{h}_{t+1}(\omega^{t}) \in L^{t+1}(\o^t)$. Thus we get that $\overline{h}_{t+1}(\omega^{t}) \in L^{t+1}(\o^t) \cap D^{t+1}(\o^t)=\{0\}$,  a contradiction.
So if $\overline{P}_t(\Pi^t)>0$ we obtain that $P(V^{0,\phi}_{T}> 0)>0$. This contradicts the (NA) condition and we obtain $\overline{P}_t(\Pi^t)=0$, the required result.
\end{proof}\\
Similarly as in \cite{RS05} and \cite{JS98}, we prove  a ``quantitative'' characterization of (NA).
\begin{proposition}
\label{AOAmulti}
Assume that the $(NA)$ condition holds true and let $0\leq t \leq T$. Then there exists $\Omega_{NA}^{t} \in \mathcal{F}_{t}$ with  $P_{t}(\Omega_{NA}^{t})=1$ and $\Omega_{NA}^{t} \subset \Omega_{NA1}^{t}$ (see Lemma \ref{localNA} for the definition of $\Omega_{NA1}^{t}$) such that for all $\o^t \in {\Omega}_{NA}^t$, there exists
$\alpha_t(\o^t)  \in (0,1]$ such that for all $h \in D^{t+1}(\o^t)$
\begin{eqnarray}
\label{valaki}
q_{t+1}\big(h\Delta S_{t+1}(\o^t,\cdot)\leq-\alpha_t(\o^t)|h||\o^t\big)\geq \alpha_{t}(\o^t).
\end{eqnarray}
Furthermore $\o^{t}  \to \alpha_{t}(\o^{t})$ is $\mathcal{F}_{t}$-measurable.
\end{proposition}
\begin{proof}
Let $\o^{t} \in \Omega^{t}_{NA1}$ be fixed ($ {\O}_{NA1}^t$ is defined in Lemma \ref{localNA}). \\
\emph{Step 1 : Proof of \eqref{valaki}.} Introduce the following set for $n\geq 1$
\begin{align}
\label{defAn}
A_{n}(\o^{t}):=\left\{ h \in  D^{t+1}(\o^{t}),\; |h|=1,\; q_{t+1}\left(h\Delta S_{t+1}(\o^t,\cdot)\leq-\frac{1}{n}|\o^t\right)< \frac{1}{n} \right\}.
\end{align}
Let $\overline{n}_{0}(\o^{t}):=\inf\{n \geq 1, A_{n}(\o^{t})=\emptyset\}$ with the convention that $\inf \emptyset=+\infty$.
Note that if $D^{t+1}(\o^{t}) =\{0\}$, then $\overline{n}_{0}(\o^{t})=1 <\infty$.  We assume now that $D^{t+1}(\o^{t}) \neq \{0\}$ and we prove by
contradiction that  $\overline{n}_{0}(\o^{t})<\infty$.
Assume that $\overline{n}_{0}(\o^{t})=\infty$ $i.e$ for all $n \geq 1$, $A_{n}(\o^{t}) \neq\emptyset$. We thus get $h_{n}(\o^{t}) \in D^{t+1}(\o^t)$ with $|h_{n}(\o^{t})|=1$ and such that
$$q_{t+1}\left(h_n(\o^{t})\Delta S_{t+1}(\o^t,\cdot)\leq-\frac{1}{n} \, |\, \o^t\right) < \frac{1}{n}.$$
By passing to a sub-sequence we can assume that $h_{n}(\o^{t})$ tends to  some $h^{*}(\o^{t})\in D^{t+1}(\o^t)$ (recall that the set $D^{t+1}(\o^t)$ is closed by definition) with $|h^{*}(\o^{t})|=1$.
Introduce
\begin{eqnarray*}
B(\o^{t}) & := &\{\omega_{t+1} \in \Omega_{t+1},\; h^{*}(\o^{t})\Delta S_{t+1}(\o^t,\omega_{t+1})<0 \}\\
B_{n}(\o^{t})& := & \{\omega_{t+1} \in \Omega_{t+1},\; h_{n}(\o^{t})\Delta S_{t+1}(\o^t,\omega_{t+1}) \leq -1/n \}.
\end{eqnarray*}
Then $B(\o^{t}) \subset \liminf_{n} B_{n}(\o^{t})$.
Furthermore as $1_{\liminf_{n} B_{n}(\o^{t})}=\liminf_{n}1_{B_{n}(\o^{t})}$, Fatou's Lemma implies that
\begin{align*}
q_{t+1}\left(h^*(\o^{t})\Delta S_{t+1}(\o^t,\cdot) < 0|\o^t\right) &
  \leq   \int_{\Omega_{t+1}}
1_{\liminf_{n} B_{n}(\o^{t})}(\o_{t+1}) q_{t+1}(\o_{t+1}|\o^t)\\
 & \leq  \liminf_{n} \int_{\Omega_{t+1}}
1_{B_n(\o^{t})}(\o_{t+1}) q_{t+1}(\o_{t+1}|\o^t)=0.
\end{align*}
This implies that $q_{t+1}\left(h^*(\o^{t})\Delta S_{t+1}(\o^t,\cdot) \geq0|\o^t\right)=1$, and thus from \eqref{EqlocalNA2} in Lemma \ref{localNA} we get that $h^{*}(\o^{t})=0$
which contradicts $|h^{*}(\o^{t})|=1$. Thus $\overline{n}_{0}(\o^{t})<\infty$ and we can set for $\o^{t} \in  \Omega^{t}_{NA1}$
$$\overline{\alpha}_{t}(\o^{t})= \frac{1}{\overline{n}_{0}(\o^{t})}.$$
It is clear that $\overline{\alpha}_{t} \in (0,1]$. Then for all $\o^t \in \Omega^{t}_{NA1}$, for all $h \in D^{t+1}(\o^{t})$ with $|h|=1$, by definition of $A_{\overline{n}_{0}(\o^{t})}(\o^{t})$ we obtain
\begin{eqnarray}
\label{valakibar}
q_{t+1}\left(h\Delta S_{t+1}(\o^t,\cdot)\leq-\overline{\alpha}_t(\o^t)|\o^t\right)\geq \overline{\alpha}_t(\o^t).
\end{eqnarray}
\emph{Step 2 : measurability issue.} \\
We now construct a function ${\alpha}_{t}$ which is $\Fc_t$-measurable and satisfies \eqref{valaki} as well. To do that we use the Aumann Theorem again as in the proof of Lemma \ref{localNA} but this time applied to the random set $A_{n}: \Omega^{t} \twoheadrightarrow \mathbb{R}^{d}$ where $A_{n}(\o^t)$ is defined in \eqref{defAn} if $\omega^{t} \in \Omega^{t}_{NA1}$ and $A_{n}(\o^{t})=\emptyset$ otherwise.

We prove that $graph(A_{n})  \in \mathcal{F}_{t}\otimes \mathcal{B}(\mathbb{R}^{d})$. From Lemma \ref{LemmaA1}, the function $(\o^{t},h) \to q_{t+1} \left(h\Delta S_{t+1}(\o^t,\cdot)\leq-\frac{1}{n}|\o^t\right)$ is  $\mathcal{F}_{t} \otimes\mathcal{B}(\mathbb{R}^{d})$-measurable. From Lemma \ref{Dmeasurability}, $Graph(D^{t+1}) \in \mathcal{F}_{t}\otimes \mathcal{B}(\mathbb{R}^{d})$ and the result follows from
\begin{align*}
Graph(A_{n}) & = Graph(D^{t+1}) \bigcap \left(\Omega_{NA1}^{t}  \times \{h \in \mathbb{R}^d,\; |h|=1\}\right) \\
 & \bigcap \left\{(\o^{t},h)\in \O^t \times  \mathbb{R}^d,\; q_{t+1} \left(h\Delta S_{t+1}(\o^t,\cdot)\leq-\frac{1}{n}|\o^t\right)< \frac{1}{n}\right\}.
\end{align*}
Using the Projection Theorem (see for example Theorem 3.23 in \cite{CV77}), we get that $\{\o^{t}\in \O^t, \;A_{n}(\o^{t}) \neq \emptyset\} \in \overline{\mathcal{F}}_{t}$.
We now extend $\overline{n}_{0}$ to $\O^t$ by setting $\overline{n}_{0}(\o^t)=1$ if $\o^t \notin \Omega^{t}_{NA1}$.
Then
$\{\overline{n}_{0} \geq 1\}=\Omega^{t} \in \Fc_t \subset \overline{\mathcal{F}}_{t}$ and for $k > 1$
\begin{align*}
\{\overline{n}_{0} \geq k\} &= \Omega^{t}_{NA1} \cap \bigcap_{1 \leq n \leq k-1} \{A_{n} \neq \emptyset\} \in \overline{\mathcal{F}}_{t},
\end{align*}
this implies that $\overline{n}_{0}$ and thus $\overline{\alpha}_{t}$ is $\overline{\mathcal{F}}_{t}$-measurable.
Using Lemma
\ref{completemes}, we get some $\mathcal{F}_{t}$-measurable function $\alpha_{t}$  such that $\alpha_{t}=\overline{\alpha}_{t}$ $P_{t}$ almost surely, $i.e$  there exists $M^{t} \in \mathcal{F}_{t}$ such that $P_{t}(M^{t})=0$ and $\{\alpha_{t} \neq \overline{\alpha}_{t}\} \subset M^{t}$.
We set $\Omega^{t}_{NA}:=\Omega^{t}_{NA1}  \bigcap \left(\O^t \setminus M_{t}\right)$. Then $P_{t}(\Omega^{t}_{NA})=1$ and as $\a_t$
is $\mathcal{F}_{t}$-measurable it remains to check that \eqref{valaki}  holds true.

For $\omega^{t} \in \Omega^{t}_{NA}$,
$\alpha_{t}(\o^{t})=\overline{\alpha}_{t}(\o^{t})$ (recall that $\omega^{t} \in \O^t \setminus M_{t}$) and since $\omega^{t} \in \Omega^{t}_{NA1}$, \eqref{valakibar} holds true and consequently \eqref{valaki} as well. It is also clear that $\alpha_{t}(\o^{t}) \in (0,1]$ and the proof is completed.
\end{proof}\\
\begin{remark}
\label{t=0}
In Definition \ref{DefD}, Lemmata \ref{Dmeasurability}, \ref{localVectorSpace}, \ref{localNA} and Proposition \ref{AOAmulti} we have included the case $t=0$. Note however that since $\Omega^{0}=\{\o_{0}\}$, the various statements and their respective proofs could be considerably simplified.
\end{remark}
\begin{remark}
The characterization of (NA) given by \eqref{valaki}
works only for $h \in D^{t+1}(\o^t)$. This is the reason why we will have to project
the strategy $\phi_{t+1} \in \Xi_{t}$  onto $D^{t+1}(\o^t)$ in our proofs.
\end{remark}
\begin{remark}
\label{RSOA}
In order to obtain Proposition \ref{AOAmulti} we could have applied directly  Proposition 3.3. of \cite{RS05} (note their proof doesn't use measurable selection arguments and provides directly the $\Fc_{t}$ measurability of $\alpha_{t}$) and
used  Lemma \ref{lienespcond}.
\end{remark}

\section{Utility problem and main result}\label{se3}
We now describe the investor's risk preferences by a possibly non-concave, random utility function.

\begin{definition}
\label{utilite}
	A \emph{random utility}  is any function $U:~\Omega \times \mathbb{R} \rightarrow \mathbb{R}\cup \{\pm \infty\}$  satisfying
		the following  conditions
	\begin{itemize}
		\item
		for every $x\in \mathbb{R}$, the function $U \left(\cdot,x\right):~\Omega\rightarrow\mathbb{R}\cup \{\pm \infty\}$ is $\Fc$-measurable,
		\item
		for all $\omega\in \Omega$, the function $U \left(\omega,\cdot\right):~\mathbb{R} \rightarrow\mathbb{R}\cup \{\pm \infty\}$ is non-decreasing and  usc on $\mathbb{R}$,
		\item
		$U(\cdot,x)=-\infty$, for all $x <0$.
	\end{itemize}
	 \end{definition}
We introduce the following notations.
\begin{definition}
\label{defGeneral}
For all $x\geq 0$, we denote by $\Phi(x)$  the set of all strategies $\phi \in \Phi$ such that $P_{T}(V_{T}^{x,\phi}(\cdot)\geq 0)=1$ and by  $\Phi(U,x)$ the set of all strategies $\phi \in \Phi(x)$ such that
$EU(\cdot,V_{T}^{x,\phi})$ exists in a generalised sense, $i.e.$ either $EU^{+}(\cdot,V_{T}^{x,\phi}(\cdot))<\infty$ or $EU^{-}(\cdot,V_{T}^{x,\phi}(\cdot))<\infty$.
\end{definition}
\begin{remark}
\label{nat}
Under (NA), if $\phi \in \Phi(x)$ then we have that $P_{t}(V_{t}^{x,\phi}(\cdot)\geq 0)=1$  for all $1 \leq t \leq T$ see  Lemma \ref{AOAAPP}.
\end{remark}
We now formulate the problem which is  our main concern in the sequel.
\begin{definition}
\label{DefU}
	Let $x \geq 0 $. The \emph{non-concave portfolio problem} on a finite horizon $T$ with initial wealth $x$  is
	\begin{eqnarray}\label{eq:OP}
		u(x) & := & \sup_{\phi \in \Phi(U,x)} E U(\cdot,V^{x,\phi}_{T}(\cdot)).
	\end{eqnarray}
\end{definition}
\begin{remark}
\label{moregeneral}
Assume that  there exists some $P$-full measure set $\widetilde{\Omega} \in \mathcal{F}$ such that for all $\o \in \widetilde{\Omega}$, $x \to U(\o,x)$ is   non-decreasing and  usc  on $[0,+\infty)$, $i.e.$ $x \to U(\o,x)$ is usc on $(0,\infty)$ and  for any $(x_{n})_{n \geq 1} \subset [0,+\infty)$ converging to $0$, $U(\o,0) \geq \limsup_{n} U(\o,x_{n})$.  We set $\overline{U}:\Omega \times \mathbb{R} \rightarrow \mathbb{R}\cup \{\pm \infty\}$
$$\overline{U}(\o,x):= U(\o,x) 1_{\widetilde{\Omega}\times [0,+\infty)}(\o,x)+ (-\infty) 1_{\Omega\times (-\infty,0)}(\o,x).$$
Then $\overline{U}$ satisfies Definition \ref{utilite}, see Lemma \ref{usc} for the second item. Moreover, the value function does not change
\begin{align*}
u(x)=& \sup_{\phi \in \Phi(U,x)} E \overline{U}(\cdot,V^{x,\phi}_{T}(\cdot)),
\end{align*}
and if there exists some $\phi^{*} \in \Phi(U,x)$ such that $u(x)=E \overline{U}(\cdot,V^{x,\phi^{*}}_{T}(\cdot))$, then $\phi^{*}$ is an optimal solution for \eqref{eq:OP}.\\
\end{remark}
\begin{remark}
\label{choixU}
Let $U$ be a utility function defined only on $(0,\infty)$ and verifying for every $x\in (0,\infty)$, $U \left(\cdot,x\right):~\Omega\rightarrow\mathbb{R}\cup \{\pm \infty\}$ is $\Fc$-measurable and for all $\omega \in \Omega$, $U \left(\omega,\cdot\right):~(0,\infty) \rightarrow\mathbb{R}\cup \{\pm \infty\}$ is non-decreasing and  usc on $(0,\infty)$.  We  may extend $U$ on $\mathbb{R}$ by setting,  for all $\o \in \Omega$, $\overline{U}(\o,0)=\lim_{x \to 0} U(\o,x)$ and for $x<0$, $\overline{U}(\o,x)=-\infty$. Then, as before, $\overline{U}$ verifies  Definition \ref{utilite} and  the value function has not changed. Note that we could have considered a closed interval $F=[a,\infty)$ of $\mathbb{R}$ instead of $[0,\infty)$, we could have adapted our notion of upper semicontinuity  and all the sequel would apply. \end{remark}
We now present conditions on $U$ which allows to assert that if $\phi \in \Phi(x)$ then $E U(\cdot,V^{x,\phi}_{T}(\cdot))$ is well-defined and that there exists some optimal solution for \eqref{eq:OP}.
\begin{assumption}
\label{uisfinite} For all $\phi \in \Phi(U,1)$, $E U^{+} \left(\cdot,V^{1,\phi}_{T}(\cdot)\right) <\infty$.
\end{assumption}
\begin{assumption}
\label{PhiUX} $\Phi(U,1)=\Phi(1)$.
\end{assumption}
\begin{remark}
\label{rmzero}
Assumptions \ref{uisfinite} and  \ref{PhiUX} are connected but play a different role.
Assumption \ref{PhiUX} guarantees that  $E U \left(\cdot,V^{1,\phi}_{T}(\cdot)\right)$ is well-defined for all $\Phi \in \Phi(1)$ and
allows us to relax  Assumption 2.7 of \cite{CRR}  on the behavior of $U$ around $0$, namely that $EU^{-}(\cdot,0)<\infty$. Then Assumption \ref{uisfinite} (together with Assumption \ref{ae}) is used to show that $u(x)<\infty$ for all $x>0$.  Note that Assumption \ref{uisfinite} is much more easy to verify that the classical assumption that $u(x)<\infty$ (for all or some $x>0$), which is usually made in the theory of maximisation of the terminal wealth utility.

In Proposition \ref{propufini}, we will show that under Assumptions  \ref{uisfinite}, \ref{PhiUX} and \ref{ae},  $E U^{+} \left(\cdot,V^{x,\phi}_{T}(\cdot)\right) <\infty$ for all $x \geq 0$ and $\phi \in \Phi(x)$. Thus $\Phi(U,x)=\Phi(x)$.
Note that if there exists some $\Phi \in \Phi(U,x)$ such that $E U^{+} \left(\cdot,V^{x,\phi}_{T}(\cdot)\right)=\infty$
and $E U^{-} \left(\cdot,V^{x,\phi}_{T}(\cdot)\right) <\infty$ then $u(x)=\infty$ and the problem is ill-posed.

We propose  some examples where Assumptions \ref{uisfinite} or \ref{PhiUX} hold true. Example $ii)$ illustrates the distinction between Assumptions \ref{uisfinite} and  \ref{PhiUX} and justifies we do not merge both assumptions and postulate that  $E U^{+} \left(\cdot,V^{1,\phi}_{T}(\cdot)\right) <\infty$, for all $\phi \in \Phi(1)$.
\begin{itemize}
\item[i)] If $U$ is bounded above then both Assumptions are trivially true. We get directly that $\Phi(U,x)= \Phi(x)$ for all $x \geq 0$.
\item[ii)]   Assume that $E U^{-}(\cdot,0) <\infty$ holds true.  Let $ x\geq 0$ and $\phi \in \Phi(x)$ be fixed. Using that $U^-$ is non-decreasing for all $\o \in \O$  we get that
 $$E U^- (\cdot,V_{T}^{x,\phi}(\cdot)) \leq E U^{-}(\cdot,0)  <  +\infty,$$
Thus $E U(\cdot,V^{x,\phi}_{T}(\cdot))$ is well-defined, $\Phi(U,x)=\Phi(x)$ and  Assumption \ref{PhiUX} holds true.
\item[iii)]  Assume that there exists some $\hat{x} \geq 1$ such that $U(\cdot,\hat{x}-1) \geq 0$ $P$-almost surely and $$\widehat{u}(\hat{x}):= \sup_{\phi \in \Phi(\hat{x})} E U(\cdot,V^{\hat{x},\phi}_{T}(\cdot))<\infty,$$
where we set for $\phi \in \Phi(\hat{x}) \backslash{\Phi(U,\hat{x})}$, $E U(\cdot,V^{\hat{x},\phi}_{T}(\cdot))=-\infty$.
Let  $\phi \in \Phi(1)$ be fixed. Then using that $U$ is non-decreasing for all $\o \in \O$,  we have that $P$-almost surely
$$U(\cdot,V_{T}^{1,\phi}(\cdot)+ \hat{x}-1) \geq  U(\cdot,\hat{x}-1) \geq 0.$$
Therefore $U(\cdot,V_{T}^{1,\phi}(\cdot)+ \hat{x}-1)=U^{+}(\cdot,V_{T}^{1,\phi}(\cdot)+ \hat{x}-1)$ $P$-almost surely. Now using that $U^{+}$ is non-decreasing for all $\o \in \O$ we get that for all $\phi \in \Phi(1)$
$$E U^{+}(\cdot,V_{T}^{1,\phi}(\cdot)) \leq E U^{+}(\cdot,V_{T}^{1,\phi}(\cdot)+\hat{x}-1) = E U(\cdot,V_{T}^{1,\phi}(\cdot)+\hat{x}-1) \leq \widehat{u}(\hat{x})<+\infty$$
and Assumptions \ref{uisfinite} and \ref{PhiUX} are satisfied. Instead of stipulating that $\widehat{u}(\hat{x})<\infty$ it is enough to assume that $E U(\cdot,V_{T}^{\hat{x},\phi}(\cdot)) < \infty$ for all $\phi \in \Phi(\hat{x})$.
\item[iv)] We will prove in Theorem \ref{main2} that under the (NA) condition and Assumption \ref{ae}, Assumptions \ref{uisfinite} and \ref{PhiUX} hold true if  $E U^{+}(\cdot,1)< +\infty$  and  if for all $0 \leq t \leq T$ $|\Delta S_{t}|,\,\frac{1}{\alpha_{t}} \in \mathcal{W}_{t} $ (see \eqref{Wcal} for the definition of $\mathcal{W}_{t}$).
\end{itemize}
\end{remark}

\begin{assumption}\label{ae}
We assume that there
exist some  constants  $\overline{\gamma}\geq 0$, $K>0$,  as well as a random variable $ C$ satisfying $C(\o) \geq 0$ for all $\omega \in \Omega$ and $E( C) <\infty$  such that for all $\omega \in \Omega$, $\lambda\geq 1$ and $x \in \mathbb{R}$, we have
\begin{eqnarray}
\label{ae+}
U(\omega,\lambda x) & \leq & K\lambda^{\overline{\gamma}}\left(U\left(\omega,x+\frac{1}{2}\right)+  C(\omega)\right).
\end{eqnarray}
\end{assumption}
\begin{remark}
First note that the constant $\frac{1}{2}$ in \eqref{ae+} has been chosen arbitrarily to simplify the presentation. This can be done without loss of generality. Indeed, assume there exists some constant $\overline{x} \geq 0$ such that  for all $\omega \in \Omega$, $\lambda\geq 1$ and $x \in \mathbb{R}$
\begin{eqnarray}
\label{aegen}
U(\omega,\lambda x) & \leq & K\lambda^{\overline{\gamma}}\left(U(\omega,x+\overline{x})+  C(\omega)\right).
\end{eqnarray}
Using the monotonicity of $U$, we can always assume $\overline{x}>0$. Set for all $\o \in \O$ and $x \in \mathbb{R}$, $\overline{U}(\o,x)=U(\o,2\overline{x}x)$. Then for all $\omega \in \Omega$, $\lambda\geq 1$ and $x \in \mathbb{R}$, we have that
$$\overline{U}(\o,\lambda x)=U(\o,2\lambda \overline{x}x) \leq K\lambda^{\overline{\gamma}}\left(U(\omega,2\overline{x}x+\overline{x})+  C(\omega)\right)= K\lambda^{\overline{\gamma}}\left(\overline{U}\left(\omega,x+\frac{1}{2}\right)+  C(\omega)\right),$$
and $\overline{U}$ satisfies \eqref{ae+}. It is clear that if $\phi^{*}$ is an optimal solution for the problem \\
$\overline{u}(x):=\sup_{\phi \in \Phi(\overline{U},\frac{x}{2\overline{x}})} E \overline{U}(\cdot,V^{\frac{x}{2\overline{x}},\phi}_{T}(\cdot))$ then $2\overline{x}\phi^{*}$ is an optimal solution for \eqref{eq:OP}. Note as well that, since $K>0$ and $C \geq 0$, it is immediate to see that  for all $\omega \in \Omega$, $\lambda\geq 1$ and $x \in \mathbb{R}$ \begin{eqnarray}
\label{aePos}
U^{+}(\omega,\lambda x) & \leq & K\lambda^{\overline{\gamma}}\left(U^{+}\left(\omega,x+\frac{1}{2}\right)+  C(\omega)\right).
\end{eqnarray}
\end{remark}
\begin{remark}
We now provide some insight  on Assumption \ref{ae}. As the inequality \eqref{ae+} is used to control the behaviour of $U^{+}(\cdot,x)$ for large values of $x$, the usual assumption in the non-concave case (see Assumption 2.10 in \cite{CRR}) is that there exists some $\hat{x} \geq 0$ such that $E U^{+}(\cdot,\hat{x}) <\infty$ as well as a random variable $ C_{1}$ satisfying $E( C_{1}) <\infty$ and $ C_{1}(\o) \geq 0$ for all $\o$  \footnote{In the cited paper $C_{1} \geq 0$ a.s but this is not an issue, see Remark \ref{remae} below} such that for all $x \geq \hat{x}$, $\lambda\geq 1$ and  $\o \in \O$
\begin{align}
\label{aetemp}
U(\omega,\lambda x)  \leq  \lambda^{\overline{\gamma}}\left(U(\omega,x)+  C_{1}(\omega)\right).
 \end{align}
 We prove now  that if \eqref{aetemp} holds true then \eqref{aegen} is verified with $\overline{x}=\hat{x}$, $K=1$ and $C=C_{1}$. Indeed, assume that \eqref{aetemp} is  verified.
For $x \geq 0$, using the monotonicity of $U$, we have for all $\o \in \O$ and $\lambda\geq 1$ that
 $$U(\o,\lambda x) \leq U(\o,\lambda (x+\hat{x})) \leq  \lambda^{\overline{\gamma}}\left(U(\omega,x+\hat{x})+  C_{1}(\omega)\right).$$
 And for $x<0$  this is true as well since $U(\o,x)=-\infty$.\\
Therefore  \eqref{aegen} is a weaker assumption than  \eqref{aetemp}.
Note as well that if we  assume that \eqref{aetemp}  holds true for all $x >0$, then  if $0<x<1$ and $\o \in \O$ we have $$U(\o,1) \leq \left(\frac{1}{x}\right)^{\overline{\gamma}} \left(U(\o, x)+C_{1}(\o)\right),$$
and $U(\o,0):=\lim_{x \to 0, \; x>0} U(\o,x) \geq -C_{1}(\o)$. This excludes for instance the case where $U$ is the logarithm. Furthermore, this also implies that $E U^{-}(\cdot,0) \leq EC_1<\infty$ and we are back to Assumption 2.7 of \cite{CRR}  \\
Alternatively, recalling the way the concave case is handled (see Lemma 2 in \cite{RS05}), we could have introduced that  there exists a random variable $ C_{2}$ satisfying $E( C_{2}) <\infty$ and $ C_{2} \geq 0$ such that for all $x \in \mathbb{R}$, $\o \in \O$
\begin{align}
\label{aetemp2}
U^{+}(\omega,\lambda x)  \leq  \lambda^{\overline{\gamma}}\left(U^{+}(\omega,x)+  C_{2}(\omega)\right).
\end{align}
We have not done so as it is difficult to prove that this inequality is preserved through  the dynamic programming procedure when considering non-concave functions unless we assume that $E U^{-}(\cdot,0)<\infty$ as in \cite{CRR}.
\end{remark}
\begin{remark}
\label{remae}
If there exists some set $\Omega_{AE} \in \mathcal{F}$ with $P(\Omega_{AE})=1$ such that \eqref{ae+}  holds true only for $\o \in \Omega_{AE}$, then setting as in Remark \ref{moregeneral}, $\overline{U}(\o,x):=U(\o,x)1_{\Omega_{AE}\times \mathbb{R}}(\o,x),$
$\overline{U}$ satisfies \eqref{ae+} and the value function in \eqref{eq:OP} does not change.
We also assume without loss of generality that ${C}(\o) \geq 0$ for all $\o$ in \eqref{ae+}. Indeed, if $ C \geq 0$ $P$-a.s, we could consider $\widetilde{C}:= C\mathbb{I}_{\overline C\geq0}$. Then Assumption \ref{ae} would hold true with $\widetilde{C}$ instead of $ C$.\end{remark}

\begin{remark}
In the case where \eqref{aetemp} holds true, we refer to remark 2.5 of \cite{CR14} and remark 2.10 of \cite{CRR} for the interpretation of $\overline{\gamma}$~: for $ C_{1}=0$, it can be seen as a generalization of
the ``asymptotic
elasticity'' of $U$ at $+ \infty$ (see \cite{KS99}). So  \eqref{aetemp}  requires that  the (generalized) asymptotic
elasticity at $+ \infty$ is finite. In this case and if $U$ is differentiable there is a nice economic interpretation of  the ``asymptotic elasticity'' as the ratio of ``marginal utility'': $U'(x)$ and the ``average utility'': $\frac{U(x)}{x}$, see  again Section 6 of \cite{KS99} for  further discussions.
The case $C_1>0$ allows bounded utilities. In \cite{CRR} it is proved that  unlike in the concave case, the fact that $U$ is bounded from
above (and therefore satisfies  \eqref{aegen}) does not implies that the asymptotic elasticity is bounded.   \\
We propose now an example of an unbounded utility function satisfying \eqref{aegen} and such that \\ $\limsup_{x \to \infty} \frac{x U'(x)}{U(x)}=+\infty$. This shows (as the counterexample of \cite{CRR}), that Assumption \ref{ae} is less strong that the usual ``asymptotic elasticity". Let $U: \mathbb{R}  \to \mathbb{R}$ be defined by
\begin{align*}
U(x)= -\infty 1_{(-\infty,0)}(x)+ \sum_{p \geq 0} p1_{[p,p+1-\frac{1}{2^{p+1}})}(x) +f_{p}(x)1_{[p+1-\frac{1}{2^{p+1}},p+1)}(x)
\end{align*}
where $f_{p}(x)=2^{p+1}x+(p+1)\left(1-2^{p+1}\right)$ for $p \in \mathbb{N}$. Then $U$ satisfies Definition \ref{utilite} and we have
$$U'(x)=\sum_{p \geq 0} 2^{p+1}1_{[p+1-\frac{1}{2^{p+1}},p+1)}(x).$$
We prove that \eqref{aegen} holds true. Note that for all $x \geq 0$ we have $ x-1 \leq U(x) \leq x+1$. Let $x \geq 0$ and $\lambda \geq1$ be fixed. Then we get that
$$U(\l x) \leq \l x +1 \leq \l \left(U(x+1)+1\right)+1 \leq  \l \left(U(x+1)+2\right), $$
and \eqref{aegen} is true with $K=\overline{x}=1$ and $C=2$.
Now  for $k \geq 0$, let $x_{k}= k+1-\frac{1}{2^{k+2}}$. We have $U(x_{k})=f_{k}(x_{k})= k+\frac{1}{2}$ and
$$\frac{x_{k} U'(x_{k})}{U(x_k)}= 2^{k+1}\frac{\left(k+1- \frac{1}{2^{k+2}}\right)} {k+\frac{1}{2}} \to_{k \to \infty} +\infty.$$

\end{remark}
\begin{remark}
\label{concave}
We propose further examples where Assumption \ref{ae} holds true.
\begin{itemize}
\item[i)]
Assume that $U$ is bounded from above by some  integrable random constant $C_1 \geq 0$ and that $EU^{-} (\cdot,\frac{1}{2}) <\infty$. Then for all $x \geq 0$, $\lambda \geq 1$, $\o \in \O$ we have
\begin{align*}
U(\o, \l x) \leq C_1(\o) &\leq \l U\left(\o,x + \frac{1}{2} \right)+ \l \left(C_1(\o) -U\left(\o,x + \frac{1}{2} \right)\right) \\
& \leq  \l U\left(\o,x + \frac{1}{2} \right)+ \l \left(C_1(\o) +U^{-}\left(\o,\frac{1}{2} \right)\right),
\end{align*}
and \eqref{ae+} holds true for $x \geq 0$ with $K=1$, $\overline{\gamma}=1$ and $C(\cdot)=C_{1}(\cdot) + U^{-}(\cdot,\frac{1}{2})$. As $U(\cdot,x)=-\infty$ for $x<0$, \eqref{ae+} is true for all  $x \in \mathbb{R}$.
\item[ii)]Assume that $U$ satisfies Definition \ref{utilite} and that the restriction of $U$ to $[0,\infty)$ is concave and non-decreasing and that $E U^{-}(\cdot,1)<\infty$. We use similar arguments as in Lemma 2 in \cite{RS06}.
Indeed, let $x \geq 2$, $\l \geq 1$ be fixed we have
\begin{align*}
U(\o,\l x)  \leq U(\o,x) + U^{'}(\o,x) (\l x-x) &\leq U(\o,x) + \frac{U(\o,x) - U(\o,1)}{x-1}(\l -1)x \\
&\leq U(\o,x) + 2 (\l -1) \left(U(\o,x)-U(\o,1)\right) \\
&\leq U(\o,x) + 3 (\l - \frac{1}{3}) \left(U(\o,x)-U(\o,1)\right)\\
&\leq 3 \l \left(U(\o,x)+ U^{-}(\o,1)\right),
\end{align*}
where we have used the concavity of $U$ for the first two inequalities and the fact that $x \geq 2$ and $U$ is non-decreasing for the other ones. Thus from the proof that  \eqref{aetemp} implies \eqref{aegen}, we obtain that \eqref{aegen} holds true with $K=3$, $\overline{\gamma}=1$, $\overline{x}=2$ and $C(\cdot)=U^{-}(\cdot,1)$.\end{itemize}
\end{remark}
We can now state our main result.
\begin{theorem}
\label{main}
Assume the (NA) condition and that Assumptions \ref{uisfinite}, \ref{PhiUX} and \ref{ae} hold true. Let $x\geq 0$. Then, $u(x)<\infty$ and there exists some optimal strategy $\phi^* \in \Phi(U,x)$ such that
$$u(x) =  EU(\cdot,V^{x,\phi^*}_{T}(\cdot)).$$
Moreover $\phi^*_{t}(\cdot) \in D^{t}(\cdot)$ a.s. for all $0 \leq t \leq T$.
\end{theorem}
We will use dynamic programming in order to prove our main result. We will combine the approach of \cite{RS05}, \cite{RS06}, \cite{CR14}, \cite{CRR} and \cite{Nutz}. As in \cite{Nutz}, we will consider a one period case where the initial filtration is trivial (so that  strategies are in $\mathbb{R}^{d}$) and thus the proofs are much simpler than the ones of \cite{RS05},
 \cite{RS06}, \cite{CR14} and \cite{CRR}. The price to pay is that in the multi-period case where we use intensively measurable selection arguments (as in  \cite{Nutz}) in order to obtain Theorem \ref{main}. In our model, there is only one probability measure, so we don't have to introduce Borel spaces and analytic sets. Thus our modelisation of $(\Omega, \Fc, \mathfrak{F},P)$ is more general than the one of \cite{Nutz} restricted to one probability measure. As we are in a non concave setting we use similar ideas to theses of \cite{CR14} and \cite{CRR}.

Finally, as in \cite{RS05}, \cite{RS06}, \cite{CR14} and \cite{CRR}, we propose the following result as a simpler but still general setting where Theorem \ref{main} applies. We introduce for all $0 \leq t \leq T$
\begin{align}
\label{Wcal}
\mathcal{W}_{t}:= \left \{ X: \Omega^{t} \to \mathbb{R}\cup \{\pm \infty\},\; \mbox{$\mathcal{F}_{t}$-measurable},\; E |X|^{p} <\infty \, \mbox{ for all $p > 0$}\right\}
\end{align}
\begin{theorem}
\label{main2}
Assume the (NA) condition and that Assumption \ref{ae} hold true. Assume furthermore that $E U^{+}(\cdot,1)< +\infty$  and that  for all $0 \leq t \leq T$ $|\Delta S_{t}|, \,\frac{1}{\alpha_{t}} \in \mathcal{W}_{t}$.   Let $x\geq 0$. Then, for all $\phi \in \Phi(x)$ and all $0 \leq t \leq T$, $V_{t}^{x,\phi} \in \mathcal{W}_{t}$. Moreover, there exists some optimal strategy $\phi^* \in \Phi(U,x)$ such that
$$u(x) =  EU(\cdot,V^{x,\phi^*}_{T}(\cdot))<\infty$$
\end{theorem}

\section{One period case}
\label{seone}
Let $(\overline{\Omega}, \Hc,Q)$ be a probability space (we denote by $E$ the expectation under $Q$) and $Y(\cdot)$ a $\Hc$-measurable $\mathbb{R}^{d}$-valued random variable. $Y(\cdot)$ could represent the change of value of the price process.
Let $D \subset \mathbb{R}^d$ be
the smallest affine subspace of $\mathbb{R}^d$ containing the support of the
distribution of $Y(\cdot)$. We assume that $D$ contains 0, so that $D$ is in fact a non-empty vector subspace of $\mathbb{R}^d$. The condition corresponding to (NA) in the present setting is
\begin{assumption}
\label{AOAone}
There exists some constant $0<\alpha\leq 1$
such that for all $h \in D$
\begin{eqnarray}
\label{naQ}
Q( h Y(\cdot) \leq -\alpha |h|) \geq \alpha.
\end{eqnarray}
\end{assumption}
\begin{remark}
\label{remAOAone}
If $D=\{0\}$ then   \eqref{naQ} is trivially true.
\end{remark}
Remark \ref{proj} below is exactly Remark 8 of \cite{CR14} (see also Lemma 2.6 of \cite{Nutz}).
\begin{remark}
\label{proj}
Let $h \in  \mathbb{R}^d$ and let $h' \in  \mathbb{R}^d$  be the orthogonal
projection of $h$ on $D$. Then $h-h'\perp D$  hence
$\{Y(\cdot)\in D\}\subset\{(h-h')Y(\cdot)=0\}$. It follows that
$$
Q(h Y(\cdot)=h' Y(\cdot))=Q((h-h')Y(\cdot)=0)\geq
Q(Y(\cdot)\in D)=1
$$
by the definition of $D$. Hence $Q(h Y(\cdot)=h' Y(\cdot))=1$.
\end{remark}
\begin{assumption}
\label{samedi}
We consider a \emph{random utility} $V:~\overline{\Omega} \times \mathbb{R} \rightarrow\mathbb{R}$ satisfying
the following two conditions
	\begin{itemize}
		\item
		for every $x\in \mathbb{R}$, the function $V\!\left(\cdot,x\right):\overline{\Omega}\rightarrow\mathbb{R}$ is $\Hc$-measurable,
		\item
		 for every $\omega\in \overline{\Omega}$, the function $V\!\left(\o,\cdot\right):\mathbb{R}\rightarrow\mathbb{R}$ is non-decreasing and usc on $\mathbb{R}$,
		 \item
		 $V(\cdot,x)=-\infty$, for all $x <0$.
	\end{itemize}
\end{assumption}
Let $x \geq 0$ be fixed. We define
\begin{align}
\label{Hx}
\Hc_{x}&: =\left\{h \in \mathbb{R}^{d}, \; Q(x+ h Y(\cdot) \geq 0)=1\right\},\\
\label{Dx}
D_{x}&:= \Hc_{x} \cap D.
\end{align}
It is clear that $\Hc_{x}$ and $D_{x}$ are closed subsets of $\mathbb{R}^{d}$.
We  now define the function which is our main concern in the one period case
\begin{align}\label{parc}
		v(x) =(-\infty) 1_{(-\infty,0)}(x)+ 1_{[0,+\infty)}(x)\sup_{h \in \Hc_{x} } EV\left(\cdot,x + hY(\cdot)\right).
	\end{align}
\begin{remark}
\label{HvsD}
First note that, from Remark \ref{proj},
\begin{eqnarray}\label{noam}
v(x)  =  (-\infty) 1_{(-\infty,0)}(x)+ 1_{[0,+\infty)}(x)\sup_{h \in D_x} EV(\cdot,x + hY(\cdot)).
\end{eqnarray}
\end{remark}
\begin{remark}
It will be shown in Lemma \ref{fat} that under Assumptions \ref{AOAone}, \ref{samedi}, \ref{ae1} and \ref{dimanche}, for all $h \in \mathcal{H}_x$, $E(V(\cdot,x + hY(\cdot))$ is well-defined and more precisely that $EV^{+}(\cdot,x + hY(\cdot))<+\infty$. So, under this set of assumptions, $\Phi(V,x)$, the set of $h \in \Hc_{x}$ such that $E V(\cdot,x+hY(\cdot))$ is well-defined, equals $\Hc_{x}$.
\end{remark}
We present now the assumptions which allow to assert that there exists some optimal solution for  \eqref{parc}.  First we introduce the ``asymptotic elasticity"  assumption.
\begin{assumption}\label{ae1}
There
exist some  constants $\overline{\gamma} \geq 0$, $K>0$, as well as some $\Hc$-measurable $C $ with $C(\o) \geq 0$ for all $\omega \in \overline{\O}$ and $E(C) <\infty$, such that for all $\omega \in \overline{\O}$, for all $\lambda\geq 1$, $x \in \mathbb{R}$ we have
\begin{align}\label{elastic1}
V(\omega,\lambda x)  \leq  K\lambda^{\overline{\gamma}}\left(V\left(\omega,x+ \frac{1}{2}\right)+ C(\omega) \right).
\end{align}
\end{assumption}
\begin{remark}
The same comments as in Remark \ref{remae} apply.
\label{remext}
Furthermore, note  that since $K>0$ and $C \geq 0$ we also have that for all $\omega \in \overline{\O}$, all $\lambda\geq 1$ and $x \in \mathbb{R}$
\begin{align}\label{elasticplus}
V^+(\omega,\lambda x)  \leq  K\lambda^{\overline{\gamma}}\left(V^+\left(\omega,x+\frac{1}{2}\right)+ C(\omega) \right).
\end{align}
\end{remark}
We introduce now some integrability assumption on  $V^{+}$.
\begin{assumption}
\label{dimanche}
For every $h \in \Hc_{1} $,
\begin{align}
\label{ia}
EV^+(\cdot,1+ h Y(\cdot))<\infty.
\end{align}
\end{assumption}
The following lemma corresponds to Lemma 2.1 of \cite{RS06} in the deterministic case.
\begin{lemma}
\label{rast}
Assume that Assumption \ref{AOAone} holds true.
Let $x \geq 0$ be fixed. Then $D_x \subset B(0,\frac{x}{\alpha})$ (see \eqref{Dx} for the definition of $D_{x}$), where $B(0,\frac{x}{\alpha})=
\{h \in \mathbb{R}^{d}, \ |h| \leq \frac{x}{\alpha}\}$ and
$D_x$ is a convex, compact subspace of $\mathbb{R}^d$.
\end{lemma}
Note that if $x=0$, it follows that $D_{x}=\{0\}$.

\begin{proof}
Let $h \in D_x$. Assume that $|h| > \frac{x}{\alpha}$ and let $\omega \in \{h Y(\cdot) \leq -\alpha |h|\}$. Then
$x+ hY(\o) <x-\alpha |h|<0$ and from Assumption \ref{AOAone}
$Q(x+ hY(\cdot) <0) \geq Q( h Y(\cdot) \leq -\alpha |h|) \geq \alpha>0$, a contradiction. The convexity and the closedness of $D_x$ are clear and the compactness follows from the boundness property.
\end{proof}\\

This lemma corresponds in the deterministic case to Lemma 4.8 of \cite{CRR} (see also Lemma 2.3 of \cite{RS06} and Lemma 2.8 of \cite{Nutz}).
\begin{lemma}
\label{fat}
Assume that  Assumptions \ref{AOAone}, \ref{samedi}, \ref{ae1} and \ref{dimanche} hold true. Then there exists a $\Hc$-measurable $L\geq 0$ satisfying $ E (L)<\infty$ and such that for all $x \geq 0$ and $h\in \mathcal{H}_{x}$
\begin{align}
\label{Vplus}
V^+(\cdot,x+hY(\cdot))\leq \left((2x)^{\overline{\gamma}}K+1\right)L(\cdot)  \; Q-a.s.
\end{align}
\end{lemma}

\begin{proof}
The proof is reported in Section \ref{proofofres} of the Appendix
\end{proof}\\
\begin{lemma}
\label{cont}
Assume that Assumptions \ref{AOAone}, \ref{samedi}, \ref{ae1} and \ref{dimanche} hold true. Let  $\Dc$ be the set valued function that assigns to each $x \geq 0$ the set $D_x$. Then $Graph(\Dc):=\{(x,h) \in [0,+\infty) \times \mathbb{R}^{d},\; h \in D_{x}\}$ is a closed subset of $\mathbb{R}\times \mathbb{R}^{d}$.
Let $\psi: \mathbb{R} \times \mathbb{R}^{d} \to \mathbb{R} \cup \{\pm \infty\}$ be defined by
\begin{align}
\label{psidef}
\psi(x,h):= \begin{cases}
EV(\cdot,x +h Y(\cdot)), \mbox{if $(x,h) \in Graph(\Dc)$}\\
-\infty, \mbox{otherwise}.
\end{cases}
\end{align}
Then $\psi$ is  usc on $\mathbb{R} \times \mathbb{R}^{d}$ and $\psi<+\infty$ on  $Graph(\Dc)$.\end{lemma}
\begin{proof}
Let $ (x_{n},h_{n})_{n\geq 1} \in Graph(\Dc)$ be a sequence converging to some $(x^{*},h^{*}) \in \mathbb{R}\times\mathbb{R}^{d}$. We prove first that $(x^{*},h^{*}) \in Graph(\Dc)$, $i.e$ that $Graph(\Dc)$ is a closed set. It is clear that $x^{*}\geq 0$.  Set for $n \geq 1$ $E_{n}:=\{\o \in \overline{\O}, \; x_{n} +h_{n}Y(\o) \geq 0 \}$ and $E^*:=\{\o \in \overline{\O}, \; x^{*} +h^{*}Y(\o) \geq 0\}$. It is clear that $\limsup_{n} E_{n} \subset E^{*}$ and applying the Fatou Lemma (the $limsup$ version) we get
\begin{align*}
Q\left(x^{*}+h^{*}Y(\cdot)\geq 0\right)  =   E 1_{E^*}(\cdot) \geq   E \limsup_{n} 1_{E_{n}}(\cdot) \geq \limsup_{n}E 1_{E_{n}}(\cdot)=1,
\end{align*}
and $h^{*} \in \mathcal{H}_{x^*}$. Since $D$ is closed by definition we have $h^{*} \in {D}_{x^*}$ and $(x^{*},h^{*}) \in Graph(\Dc)$.\\
We prove now that $\psi$ is usc on $Graph(\Dc)$.  The upper semicontinuity on $\mathbb{R} \times \mathbb{R}^{d}$ will follow immediately from Lemma \ref{usc}. By Assumption \ref{samedi} $x \in \mathbb{R} \to V(x,\o)$ is  usc on $\mathbb{R}$ for all $\o \in \overline{\Omega}$ and thus
$$\limsup_n V(\o, x_n+ h_n Y(\o))  \leq  V(\o, x^*+ h^* Y(\o)).$$
By Lemma \ref{fat} for all $\o \in \overline{\Omega}$
$$V(\o, x_n+ h_n Y(\cdot)) \leq V^+(\o, x_n+ h_n Y(\cdot))\leq (|2x_n|^{\overline{\gamma}}K +1)L(\o) \leq (|2x^*|^{\overline{\gamma}}K +2)L(\o)$$
for $n$ big enough. We can apply Fatou's Lemma (the $limsup$ version)  and
$\psi$ is usc on $Graph(\Dc)$.
From Lemma \ref{fat} it is also clear that $\psi<+\infty$ on $Graph(\Dc)$.
\end{proof}\\
We are now able to state our main result.
\begin{theorem}
\label{main1}
Assume that Assumptions \ref{AOAone}, \ref{samedi}, \ref{ae1} and \ref{dimanche} hold true.  Then for all $x \geq 0$, $v(x) < \infty$  and there exists some optimal strategy $\widehat{h} \in D_x$ such that
$$v(x) =  E(V(\cdot,x + \widehat{h}Y(\cdot))).$$
Moreover, $v:\; \mathbb{R} \rightarrow [-\infty,\ \infty)$  is non-decreasing  and usc on $\mathbb{R}$.
\end{theorem}
\begin{proof}
Let $x  \geq 0$ be fixed. We show first that $v(x) < \infty$. Indeed, using Lemma \ref{fat},
$$E(V(\cdot,x + hY(\cdot))) \leq  E(V^+(\cdot,x + hY(\cdot))) \leq \left((2x)^{\overline{\gamma}}K+1\right) E L(\cdot),$$
for all $h \in D_x$. Thus, recalling \eqref{noam}, $v(x) \leq \left((2x)^{\overline{\gamma}}+1\right) E L(\cdot)
< \infty$.

\noindent From Lemma \ref{cont}, $h \in \mathbb{R}^d  \to E(V(\cdot,x + hY(\cdot)))$ is usc on  $\mathbb{R}^{d}$ and thus on $D_x$ (recall that $D_{x}$ is closed and see Lemma \ref{usc}). Since by \eqref{noam}, $v(x)  =  \sup_{h \in D_x} E(\cdot,V(x + hY(\cdot)))$ and $D_x$ is compact (see Lemma \ref{rast}), applying Theorem 2.43 of \cite{Hitch} there exists some  $\widehat{h} \in D_x$ such that
\begin{align}
\label{solve}
v(x) =  E(V(\cdot,x + \widehat{h}Y(\cdot))).
\end{align}
We show that $v$ is usc on $[0,+\infty)$. As previously, the upper semicontinuity on $\mathbb{R}$ will follow immediately from Lemma \ref{usc}. Let $(x_n)_{n\geq 0}$ be a sequence of non-negative numbers converging to some $x^* \in [0,+\infty)$.
Let $\widehat{h}_n \in D_{x_{n}}$ be the associated optimal strategies to $x_{n}$ in \eqref{solve}.
Let $(n_k)_{k\geq 1}$  be a subsequence such that $\limsup_{n} v(x_{n}) = \lim_{k} v(x_{n_k})$.
By Lemma \ref{rast} $|\widehat{h}_{n_k}| \leq x_{n_k}/\beta \leq (x^*+1)/\beta$ for
$k$ big enough. So we can extract a subsequence (that we still denote by $(n_k)_{k\geq 1}$) such that there exists some  $\underline{h}^*$ with   $\widehat{h}_{n_k} \to \underline{h}^*$.  As the sequence $(x_{n_{k}},\hat{h}_{n_{k}})_{k\geq 1} \in Graph(\Dc)$ converges to $(x^{*},\underline{h}^*)$ and $Graph(\Dc)$ is closed (see  Lemma \ref{cont}), we get that $\underline{h}^* \in \mathcal{D}_{x^*}$.
Using Lemma \ref{cont}
\begin{eqnarray*}
\limsup_{n} v(x_n)= \lim_{k} v(x_{n_k})=\lim_k E V(\cdot, x_{n_k}+ \widehat{h}_{n_k} Y(\cdot)) \leq E V(\cdot, x^*+ \underline{h}^* Y(\cdot)) \leq v(x^*),
\end{eqnarray*}
where the last inequality holds true because  $\underline{h}^* \in D_{x^{*}}$ and therefore $v$ is usc on $[0,+\infty)$.
Now as, by Assumption \ref{samedi}, $V(\o,\cdot)$ is non-decreasing for all $\o \in \overline{\Omega}$, $v$ is also non-decreasing
on $[0,+\infty)$ and since $v(x)=-\infty$ on $(-\infty,0)$, $v$ is non-decreasing on $\mathbb{R}$.
\end{proof}\\
\section{Multi-period case}
\label{secmulti}
We first prove the following proposition.
\begin{proposition}
\label{propufini}
Let Assumptions \ref{uisfinite}, \ref{PhiUX}  and \ref{ae} hold true. Then $EU^{+}\left(\cdot,V_{T}^{x,\phi}(\cdot) \right) <\infty$ for all $x\geq 0$ and $\phi \in \Phi(x)$. This implies that $\Phi(U,x)=\Phi(x)$.
\end{proposition}
\begin{proof}
Fix $0 \leq x \leq 1$ and let $\phi \in \Phi(x)$. Then $V^{x,\phi}_T \leq V^{1,\phi}_T$ and $\phi \in \Phi(1)=\Phi(1,U)$ (recall Assumption \ref{PhiUX}). For any $\omega \in \Omega$, the function $y \to U(\omega,y)$ is non-decreasing on $\mathbb{R}$, so  that $EU^{+}\left(\cdot,V_{T}^{x,\phi}(\cdot) \right)  \leq EU^{+}\left(\cdot,V_{T}^{1,\phi}(\cdot) \right) <\infty$ by Assumption \ref{uisfinite}.
Now, if $x \geq 1$, let $\phi \in \Phi(x)$ be fixed. From Assumption \ref{ae} we get that for all $\omega \in \Omega$
$$U(\omega,V^{x,\phi}_T(\omega)) = U\left(\omega, 2x\left(\frac{1}{2}+ \sum_{t=1}^{T} \frac{\phi_{t}(\o^{t-1})}{2x}\Delta S_{t}(\o^{t})\right)\right) \leq (2x)^{\overline{\gamma}}K \left( U(\omega,V^{1,\frac{\phi}{2x}}_T(\o)) +C(\omega)\right).$$
By  Assumption \ref{PhiUX}, $\frac{\phi}{2x} \in  \Phi(\frac{1}{2}) \subset \Phi(1)=\Phi(1,U)$.  Thus $$EU^{+}\left(\cdot,V_{T}^{x,\phi}(\cdot) \right) \leq (2x)^{\overline{\gamma}} K \left(EU^{+}\left(\cdot,V_{T}^{1,\frac{\phi}{2x}}(\cdot) \right) +E(C)\right) <\infty$$ using Assumption \ref{uisfinite} and the fact that $C$ is integrable (see Assumption \ref{ae}).
In both cases, we conclude that $\Phi(x)=\Phi(U,x)$.
\end{proof}\\

We introduce now the dynamic programming procedure. First we set for all $t\in\left\{0,\ldots,T-1\right\}$, $\o^t \in {\O}^t$ and $x \geq 0$
\begin{align}
\label{domaine}
	\Hc_x^{t+1}(\o^t)&:=  \left\{h \in \mathbb{R}^d, \ q_{t+1}(x+ h \Delta S_{t+1}(\o^t, \cdot) \geq 0|\o^t)=1\right\},\\
\label{domaineproj}
	\Dc_x^{t+1}(\o^t)&:= \Hc_x^{t+1}(\o^t) \cap D^{t+1}(\o^t),
\end{align}
where  $D^{t+1}$ was introduced in Definition \ref{DefD}. For $x<0$ we set $\Hc_x^{t+1}(\o^t)=\emptyset$. \\
We  define for all $t\in\{0,\ldots,T\}$  the following functions $U_t$ from $\Omega^t \times \mathbb{R} \to \mathbb{R}$. Starting with $t=T$, we set for all $x \in \mathbb{R}$, all $\o^{T} \in \Omega$
\begin{align}\label{tuskes}
 &U_T(\o^T) := U(\o^{T}).
  \end{align}
Recall  that $U(\o^{T},x)=-\infty$ for all $(\o^{T},x)  \in \Omega\times (-\infty,0)$.\\
Using  for $t \geq 1$  the  full-measure set $\widetilde{\Omega}^{t} \in \mathcal{F}_{t}$  that will be defined by induction in Propositions \ref{dyn2} and \ref{dyn3}, we set for all $x \in \mathbb{R}$ and $\o^{t} \in \Omega^{t}$
\begin{small}
\begin{align}
\label{vanek}
 U_{t}(\o^{t},x)&:= (-\infty)1_{(-\infty,0)}(x)+1_{\widetilde{\Omega}^{t} \times [0,+\infty)}(\o^{t},x)\sup_{h \in \Hc_x^{t+1}(\o^t)}\int_{\O_{t+1}}U_{t+1}(\o^t,\o_{t+1},x+h\Delta S_{t+1}(\o^t,\o_{t+1}))q_{t+1}(d\o_{t+1}|\o^t).
\end{align}
\end{small}
Finally for $t=0$
\begin{small}
\begin{align}
\label{vanek0}
 &U_{0}(x):= (-\infty)1_{(-\infty,0)}(x)+ 1_{[0,+\infty)}(x) \sup_{h \in \Hc_x^{1}} \int_{\O_{1} }U_{1}(\o_{1}, x+ h \Delta S_{1}( \o_{1}))P_{1}(d\o_{1}) .
\end{align}
\end{small}
\begin{remark}
We will prove by induction that $U_{t}$ is well-defined (see \eqref{amiens}), $i.e$ the integrals in \eqref{vanek} and \eqref{vanek0} are well-defined in the generalised sense.
\end{remark}
\begin{remark}
Before going further we provide some explanations on the choice of $U_t$.
The natural definition of $U_{t}$ should have been
$$\mathcal{U}_{t}(\o^{t},x)  := (-\infty) 1_{(-\infty,0)}(x)+ 1_{[0,+\infty)}(x) \sup_{h \in \Hc_x^{t+1}(\o^t)} \int_{\O_{t+1} }\mathcal{U}_{t+1}(\o^t, \o_{t+1}, x+ h \Delta S_{t+1}(\o^t, \o_{t+1}))q_{t+1}(d\o_{t+1}|\o^t).$$
Introducing the $P_{t}$ full measure set $\widetilde{\Omega}^{t}$ in \eqref{vanek} is related to measurability issues that will be tackled in Proposition \ref{dyn4}. This is not a surprise as this is related to the use of  conditional expectations which are defined only almost everywhere.
\end{remark}
\begin{lemma}
\label{Dhmes}
Let $0 \leq t \leq T-1$ and $H$ be a fixed $\mathbb{R}$-valued and $\Fc_{t}$-measurable random variable. Consider the following random sets
$$\Hc^{t+1}_{H} : \o^{t} \in  \Omega^{t} \twoheadrightarrow \Hc^{t+1}_{H(\o^{t})}(\o^t),$$
$$\Dc^{t+1}_{H}: \o^{t} \in \Omega^{t} \twoheadrightarrow \Dc_{H(\o^{t})}^{t+1}(\o^t).$$
Then those random sets  are all closed-valued and with graph valued in  $\mathcal{F}_{t} \otimes \mathcal{B}(\mathbb{R}^{d})$.
\end{lemma}
\begin{proof}
First it is clear that  $\Hc^{t+1}_{H}$  is closed-valued. As $D^{t+1}$ is closed-valued (see  Lemma \ref{Dmeasurability}) it follows that  $\Dc^{t+1}_{H}$ is closed-valued as well.
The fact that $Graph(\mathcal{H}_{H}^{t+1}) \in  \mathcal{F}_{t} \otimes \mathcal{B}(\mathbb{R}^{d})$  follows immediately from
$$Graph(\mathcal{H}^{t+1}_{H})=\left\{(\o^{t},h)  \in \Omega^{t}\times\mathbb{R}^{d}, H(\o^{t}) \geq 0,\;  q_{t+1}\left(\left\{H(\o^{t})+ h\Delta S_{t+1}(\o^{t},.) \geq 0 \right\}=1|\o^{t}\right)\right\},$$
and Lemma \ref{LemmaA1} (recall that $H$ is $\mathcal{F}_{t}$-measurable).
We know from Lemma \ref{Dmeasurability} that $Graph(D^{t+1}) \in \mathcal{F}_{t} \otimes \mathcal{B}(\mathbb{R}^{d})$ and it follows that
 $$Graph(\Dc^{t+1}_{H}) =Graph(D^{t+1}) \cap Graph(\mathcal{H}_{H}^{t+1}) \in \mathcal{F}_{t} \otimes \mathcal{B}(\mathbb{R}^{d})$$.
 \end{proof}\\
Finally we introduce
\begin{eqnarray}
\nonumber
C_T(\o^T) &:= & C(\o^T), \mbox{ for } \o^T \in {\O}^T,  \mbox{ where $C$ is defined in Assumption } \ref{ae}\\
\label{Crec}
C_{t}(\o^{t}) & := &  \int_{\O_{t+1} } C_{t+1}(\omega^{t}, \o_{t+1})q_{t+1}(d\omega_{t+1}|\omega^{t}) \mbox{ for } t \in \{0,\dots,T-1\}, \o^t \in \O^{t}.
\end{eqnarray}
\begin{lemma}
\label{ToolJC}
 The functions
 $\o^{t}  \in \Omega^{t}  \to C_{t}(\o^{t})$ are well-defined, non-negative (for all $\o^t$), $\mathcal{F}_{t}$-measurable and satisfy
$E (C_{t})=E(C_{T}) <\infty$. Furthermore, for all $t\in\{1,\ldots,T\}$, there exists $\Omega^{t}_{C} \in \mathcal{F}_{t}$ and
with $P_{t}(\Omega^{t}_{C})=1$ and such that $C_{t}<\infty$
on  $\Omega^{t}_{C}$.  For $t=0$ we have $C_{0}<\infty$. \end{lemma}
\begin{proof}
We proceed by induction. For $t=T$ by Assumption \ref{ae} $C_T=C$ is $\mathcal{F}_{T}$-measurable, $C_T\geq 0$  and $E(C_T)<\infty$. Assume
now that $C_{t+1}$ is $\mathcal{F}_{t+1}$-measurable, $C_{t+1}\geq 0$ and  $E(C_{t+1})=E(C_T) <\infty$.
From Proposition \ref{fubiniext} $i)$ applied to $f=C_{t+1}$ we get that  $\omega^{t} \to C_{t}(\omega^{t})=\int_{\Omega_{t+1}} C_{t+1}(\o^{t},\o_{t+1})q_{t+1}(d\o_{t+1}|\o^{t})$ is  $\mathcal{F}_{t}$-measurable. As $C_{t+1}(\o^{t+1}) \geq 0$ for all $\o^{t+1}$, it is clear that  $C_{t}(\o^{t}) \geq 0$ for all $\o^t$. Applying the Fubini theorem (see Lemma \ref{fubini0}) we get that
\begin{eqnarray*}
E (C_{t}) &= &\int_{\O^t} \int_{\O_{t+1} } C_{t+1}(\omega^{t}, \o_{t+1})q_{t+1}(d\omega_{t+1}|\omega^{t}) P_t(d \o^t)\\
 & = & \int_{\O^{t+1}}  C_{t+1}(\omega^{t+1}) P_{t+1}(d \o^{t+1})=E(C_{t+1})=E(C_T) < \infty.
\end{eqnarray*}
and the induction step is complete.
For the second part of the lemma, we apply Lemma \ref{lemmaAA} to $f=C_{t+1}$ and we obtain that $\Omega^{t}_{C}:=\{\o^{t}\in \Omega^{t},\; C_{t}(\o^{t})<\infty\} \in \mathcal{F}_{t}$ and $P_{t}(\Omega^{t}_{C})=1$.
\end{proof}\\

Propositions \ref{dyn1} to \ref{dyn4} below  solve the dynamic programming procedure and hold true under the following set of conditions. Let $1\leq t\leq T$ be fixed. \begin{align}
\label{amiens}
&U_{t} \left(\o^{t},\cdot\right)  :   \mathbb{R} \rightarrow\mathbb{R} \mbox{ is well-defined, non-decreasing and  usc on $\mathbb{R}$ for all $\o^{t} \in \O^{t}$},\\
\label{reims}
&U_{t} \left(\cdot,\cdot\right)  :   \Omega^{t}\times\mathbb{R} \rightarrow \mathbb{R} \{\pm \infty\} \mbox{ is  ${\Fc}_{t}\otimes \mathbb{B}(\mathbb{R})$-measurable},\\
\label{dimancheR}
 &\int_{\O^{t}}U^{+}_{t}(\o^{t},H(\o^{t-1})+ \xi(\o^{t-1})\Delta S_{t}(\o^{t}))P_{t}(d\o^{t}) < \infty, \\
\nonumber &\mbox{for all $\xi\in \Xi_{t-1}$ and $H=x+ \sum_{s=1}^{t-1} \phi_s \Delta S_s$}\;
 \mbox{where $x \geq 0$, $\phi_1 \in \Xi_0, \ldots,\phi_{t-1} \in \Xi_{t-2}$ }\\
 \nonumber &\mbox{and $P_{t}(H(\cdot)+\xi(\cdot) \Delta S_{t}(\cdot) \geq 0)=1$}, \\
\label{bordeaux}
&U_{t}(\o^{t}, \lambda x)  \leq   \lambda^{\overline{\gamma}}K \left(U_{t}\left(\o^{t}, x+\frac{1}{2}\right) + {C}_{t}(\o^{t})\right), \mbox{ for all  $\o^{t} \in \O^{t}$, $\lambda \geq 1$, $x \in \mathbb{R}$.}\;
\end{align}
\begin{remark}
Note  that from \eqref{amiens} and \eqref{reims} we have that $-U_{t}$ is a $\overline{\mathcal{F}}_{t}$-normal integrand (see Definition 14.27 in \cite{rw} or Section 3 of Chapter 5 in \cite{Molchanov} 
and Corollary 14.34 of \cite{rw}). However to prove that this property is preserved in the dynamic programming procedure we need  to show separately   that  \eqref{amiens} and \eqref{reims} are true. Furthermore, as our sigma-algebras are not assumed to be complete, obtaining some $\mathcal{F}_{t}$-normal integrand from $-U_{t}$  would introduce yet another layer of difficulty. For these reasons we choose to prove  \eqref{amiens} and \eqref{reims} instead of some normal integrand property.  Nevertheless  we will use again the properties of normal integrands in the proof of Lemma \ref{dyn4}.
\end{remark}

The next proposition is a first step in the construction of $\widetilde{\Omega}^{t}$.
\begin{proposition}\label{dyn1}
Let $0\leq t\leq T-1$ be fixed. Assume that (NA) condition holds true and that
  \eqref{amiens}, \eqref{reims}, \eqref{dimancheR} and \eqref{bordeaux} hold true at stage $t+1$.
Then there exists $\widetilde{\Omega}_{1}^{t} \in \mathcal{F}_{t}$ such  that  $P_{t}(\widetilde{\Omega}_{1}^{t})=1$ and such that for all $\o^{t} \in \widetilde{\Omega}_{1}^{t}$ the function $(\omega_{t+1},x) \to U_{t+1}(\o^{t},\o_{t+1},x)$ satisfies the assumptions of Theorem \ref{main1} with $\overline{\Omega}=\Omega_{t+1}$, $\mathcal{H}=\mathcal{G}_{t+1}$, $Q(\cdot)=q_{t+1}(\cdot|\o^{t})$,  $Y(\cdot)=\Delta S_{t+1}(\o^t,\cdot)$, $V(\cdot,y)=U_{t+1}(\omega^{t},\cdot,y)$ where $V$ is defined on $\Omega_{t+1} \times \mathbb{R}$.
\end{proposition}
 \begin{remark}\label{lemmecont}
Note that Lemmata \ref{fat}, \ref{cont} and Theorem \ref{main1} hold true under the same set of assumptions. Therefore we can replace Theorem \ref{main1} by either Lemmata \ref{fat} or  \ref{cont} in  the above proposition.
\end{remark}
\begin{proof}
To prove the proposition we will review one by one the assumptions needed to apply Theorem \ref{main1}  in the context $\overline{\Omega}=\Omega_{t+1}$, $\mathcal{H}=\mathcal{G}_{t+1}$, $Q(\cdot)=q_{t+1}(\cdot|\o^{t})$,  $Y(\cdot)=\Delta S_{t+1}(\o^t,\cdot)$, $V(\cdot,y)=U_{t+1}(\omega^{t},\cdot,y)$ where $V$ is defined on $\Omega_{t+1} \times \mathbb{R}$.
In the sequel we shortly call this the context $t+1$.\\
From \eqref{amiens} at $t+1$  for all $\o^{t} \in  \Omega^{t}$ and   $\o_{t+1}\in \Omega_{t+1}$,  the function $x \in \mathbb{R} \to U_{t+1}(\o^{t},\o_{t+1},x)$ is  non-decreasing and usc on $\mathbb{R}$. From \eqref{reims} at $t+1$  for all fixed $\o^{t} \in \Omega^{t}$ and  $x \in \mathbb{R}$,  the function  $\o_{t+1} \in \Omega_{t+1} \to U_{t+1}(\o^{t},\o_{t+1},x)$ is $\mathcal{G}_{t+1}$-measurable and thus Assumption \ref{samedi} is satisfied in the context $t+1$ (recall that $U_{t+1}(\o^{t},\o_{t+1},x)=-\infty$ for all $x <0$ by assumption). \\
We move now to the assumptions that are verified for $\o^{t}$ chosen in some specific $P_{t}$-full measure set. First from Lemma \ref{localNA} for all $\o^{t} \in \Omega^{t}_{NA1}$ we have $0 \in D^{t+1}(\o^{t})$ (recall that in Section \ref{seone} we have assume that $D$ contains $0$). From Proposition \ref{AOAmulti},  Assumption \ref{AOAone} holds true  for all $\o^{t} \in \Omega^{t}_{NA}$ in the context $t+1$.\\
We handle now Assumption \ref{ae1}  on asymptotic elasticity in context $t+1$. Let $\o^{t} \in \Omega^{t}_{C}$ be fixed where $\Omega^{t}_{C}$ is defined in Lemma  \ref{ToolJC}.  From \eqref{bordeaux} at $t+1$ we have  that for all  $\omega_{t+1} \in \Omega_{t+1}$, $\lambda \geq 1$ and $x \in \mathbb{R}$
$$U_{t+1}(\o^{t},\o_{t+1}, \lambda x)  \leq   \lambda^{\overline{\gamma}}K \left(U_{t+1}\left(\o^{t},\o_{t+1}, x+ \frac{1}{2}\right) + {C}_{t+1}(\o^{t},\o_{t+1})\right).$$
Now from Lemma \ref{ToolJC} since $\o^{t} \in  \Omega^{t}_{C}$, we get that
$$\int_{\Omega_{t+1}} C_{t+1}(\o^{t},\o_{t+1}) q_{t+1}(\o_{t+1}|d\o^{t})=C_{t}(\o^{t}) <\infty$$
and thus  Assumption \ref{ae1} in context $t+1$ is verified for all $\o^{t} \in  \Omega^{t}_{C}$.  want to show that for $\o^{t}$ in some $P_{t}$ full measure set to be determined and for all $h \in \mathcal{H}_{1}^{t+1}(\o^{t})$ we have that $$\int_{\Omega_{t+1}} U^{+}_{t+1}(\o^{t},\o_{t+1},1 + h \Delta S_{t+1}(\o^{t},\o_{t+1})) q_{t+1}(d\o_{t+1}|\o^{t}) <\infty.$$
We introduce the following random set $I_{1}: \Omega^{t} \twoheadrightarrow \mathbb{R}^{d}$
\begin{align}
\label{Ix}
I_{1}(\o^{t}):=\left\{ h \in  \mathcal{H}_{1}^{t+1}(\o^{t}),\; \int_{\Omega_{t+1}} U^{+}_{t+1}(\o^{t},\o_{t+1},1 + h \Delta S_{t+1}(\o^{t},\o_{t+1})) q_{t+1}(d\o_{t+1}|\o^{t}) =\infty\right\}.
\end{align}
Arguing by contradiction and using measurable selection arguments we will prove that  $I_{1}(\o^{t})=\emptyset$ for $P_{t}$-almost all $\o^{t} \in \Omega^{t}$. We show first that  $Graph(I_{1}) \in \mathcal{F}_{t}\otimes \mathcal{B}(\mathbb{R}^{d})$. It is clear from  \eqref{reims} at $t+1$ that $(\o^{t},\o_{t+1},h) \to U^{+}_{t+1}(\o^{t},\o_{t+1},1 + h \Delta S_{t+1}(\o^{t},\o_{t+1}))$ is $\mathcal{F}_{t}\otimes \mathcal{G}_{t+1} \otimes \mathcal{B}(\mathbb{R}^d)$-measurable. Using Proposition \ref{fubiniext} $ii)$  we get that  $(\o^{t},h) \to \int_{\Omega_{t+1}} U^{+}_{t+1}(\o^{t},\o_{t+1},1 + h \Delta S_{t+1}(\o^{t},\o_{t+1})) q_{t+1}(d\o_{t+1}|\o^{t})$ is $\mathcal{F}_{t}\otimes \mathcal{B}(\mathbb{R}^d)$-measurable (taking potentially the value $+\infty$). From Lemma \ref{Dhmes}, we obtain  $Graph (\mathcal{H}_{1}^{t+1}) \in \mathcal{F}_{t}\otimes \mathcal{B}(\mathbb{R}^d)$ and $Graph(I_{1}) \in \mathcal{F}_{t}\otimes \mathcal{B}(\mathbb{R}^d)$ follows. \\
Applying the Projection Theorem (see for example Theorem 3.23 in \cite{CV77}) we obtain that  $\{I_{1} \neq \emptyset\}  \in  \overline{\mathcal{F}}_{t}$ and using the Aumann Theorem (see Corollary 1 in \citet{bv}) there exists some $\overline{\mathcal{F}}_{t}$-measurable $\overline{h}_{1}: \{I_{1} \neq \emptyset\} \to \mathbb{R}^{d}$ such that for all $\o^{t} \in \{I_{1} \neq \emptyset\} $, $\overline{h}_{1}(\o^{t}) \in I_{1}(\o^{t})$.  We extend  $\overline{h}_{1}$ on all $\Omega^{t}$ by setting $\overline{h}_{1}(\o^{t})=0$ on $\O^t \setminus \{I_{1}\neq \emptyset\}$. As $\{I_{1} \neq \emptyset\}  \in  \overline{\mathcal{F}}_{t}$ it is clear that  $\overline{h}_{1}$  remains $\overline{\mathcal{F}}_{t}$-measurable.
Using Lemma \ref{completemes} we get some $\mathcal{F}_{t}$-measurable $h_{1}: \Omega^{t} \to \mathbb{R}^{d}$  and  $\Omega^{t}_{I_{1}} \in \mathcal{F}_{t}$ such that $P_t(\Omega^{t}_{I_{1}})=1$ and $\Omega^{t}_{I_{1}} \subset \{\o^{t}\in \Omega^{t}, h_{1}(\omega^{t}) =  \overline{h}_{1}(\omega^{t})\}$.
Arguing as in the proof of Lemma \ref{localNA} and using the Fubini Theorem (see Lemma \ref{fubini0}) we get that
\begin{eqnarray*}
P_{t+1}\left(1+h_{1}(\cdot) \Delta S_{t+1}(\cdot) \geq 0 \right) & = & \int_{\Omega^{t}} q_{t+1}(1+h_{1}(\o^{t}) \Delta S_{t+1}(\o^{t},\cdot) \geq 0 | \o^{t}) P_{t} (d\o^{t}) \\
 & = & \int_{\Omega^{t}} q_{t+1}(1+\overline{h}_{1}(\o^{t}) \Delta S_{t+1}(\o^{t},\cdot) \geq 0 | \o^{t})  \overline{P}_{t}(d\o^{t})\\
 &= & 1.
\end{eqnarray*}
Now assume that $\overline{P}_{t}( \{I_{1} \neq \emptyset\}) >0$. Since $h_{1} \in \Xi_{t}$ and  $P_{t+1}(1+h_{1}(\cdot) \Delta S_{t+1}(\cdot) \geq 0)=1$
from $\eqref{dimancheR}$ at $t+1$ applied to $H=1$
$$\int_{\Omega^{t+1}} U^{+}_{t+1}(\o^{t+1},1+h_{1}(\omega^{t})\Delta S_{t+1}(\omega^{t+1})) P_{t+1}(d\o^{t+1})<\infty.$$
We argue as in Lemma \ref{localNA} again. Let
\begin{align*}
\varphi_{1}(\omega^{t}) &=\int_{\Omega_{t+1}} U^{+}_{t+1}(\omega^{t},\omega_{t+1},1+h_{1}(\omega^{t})\Delta S_{t+1}(\omega^{t},\omega_{t+1})) q_{t+1}(d\omega_{t+1}|\o^{t}),\\
\overline{\varphi}_{1}(\omega^{t})&=\int_{\Omega_{t+1}} U^{+}_{t+1}(\omega^{t},\omega_{t+1},1+\overline{h}_{1}(\omega^{t})\Delta S_{t+1}(\omega^{t},\omega_{t+1})) q_{t+1}(d\omega_{t+1}|\o^{t}).
\end{align*}
We have already seen that  $(\o^{t},h) \in \Omega^{t}\times\mathbb{R}^{d} \to \int_{\Omega_{t+1}} U^{+}_{t+1}(\o^{t},\o_{t+1},1 + h \Delta S_{t+1}(\o^{t},\o_{t+1})) q_{t+1}(d\o_{t+1}|\o^{t})$ is $\mathcal{F}_{t}\otimes \mathcal{B}(\mathbb{R}^d)$-measurable (taking potentially value $+\infty$). By composition it is clear that $\varphi_{1}$ is $\mathcal{F}_{t}$-measurable and that $\overline{\varphi}_{1}$ is $\overline{\mathcal{F}}_{t}$-measurable.
 Furthermore as $\{\o^{t} \in \O^t,\; \varphi_1(\o^{t}) \neq \overline{\varphi}_1(\o^{t})\} \subset \{\o^{t}\in \O^t,\; h_1(\o^{t}) \neq \overline{h}_1(\o^{t})\} $, $\varphi_1=\overline{\varphi}_1$ $P_{t}$-almost surely. This implies that  $\int_{\Omega^{t}} \overline{\varphi}_1 d \overline{P}_{t}=\int_{\Omega^{t}} {\varphi}_1 dP_{t}$ and using again the Fubini Theorem (see Lemma \ref{fubini0}) we get that
\begin{align*}
&\int_{\Omega^{t+1}} U^{+}_{t+1}(\o^{t+1},x+h_{1}(\omega^{t})\Delta S_{t+1}(\omega^{t+1}) P_{t+1}(d\o^{t+1})\\
&=\int_{\Omega^{t}} \int_{\Omega_{t+1}} U^{+}_{t+1}(\omega^{t},\omega_{t+1},1+h_{1}(\omega^{t})\Delta S_{t+1}(\omega^{t},\omega_{t+1})) q_{t+1}(d\omega_{t+1}|\o^{t})P_{t}(d\omega^{t})\\
&=\int_{\Omega^{t}}  \int_{\Omega_{t+1}} U^{+}_{t+1}(\omega^{t},\omega_{t+1},1+\overline{h}_{1}(\omega^{t})\Delta S_{t+1}(\omega^{t},\omega_{t+1})) q_{t+1}(d\omega_{t+1}|\o^{t})\overline{P}_{t}(d\omega^{t})\\
&\geq \int_{\{I_{1} \neq \emptyset\}} (+\infty) \overline{P}_{t}(d\omega^{t})=+\infty.
\end{align*}
Therefore we must have $\overline{P}_{t}( \{I_{1} \neq \emptyset\})  =0$ $i.e$ $\overline{P}_{t}(\{I_{1} =\emptyset\}) =1$. Now since $\{I_{1} =\emptyset\} \in \overline{\mathcal{F}}_{t}$ there exists $\Omega^{t}_{int} \subset \{I_{1} =\emptyset\}$ such that $\Omega^{t}_{int} \in \mathcal{F}_{t}$ and $P_t(\Omega^{t}_{int})=\overline{P}_{t}( \{I_{1}=\emptyset\})=1$.
 For all $\o^{t} \in \Omega^{t}_{int}$, Assumption  \ref{dimanche} in the context $t+1$ is true and we can now define $\widetilde{\Omega}_{1}^{t} \subset \Omega^{t}$
 \begin{align}
 \label{deftildeomegafirst1}
 \widetilde{\Omega}_{1}^{t}:= \Omega^{t}_{NA}  \cap \Omega^{t}_{int}\cap \Omega^{t}_{C}.
 \end{align}
 It is clear that $\widetilde{\Omega}_{1}^{t} \in \mathcal{F}_{t}$, $P_t(\widetilde{\Omega}_{1}^{t})=1$ and the proof is complete.
\end{proof}\\

The next proposition enables us to initialize the induction argument that will be carried on in Proposition \ref{dyn4}.
\begin{proposition}\label{dyn2}
 Assume that the (NA) condition and Assumptions \ref{uisfinite}, \ref{PhiUX} and \ref{ae} hold true. Then $U_{T}$ satisfies \eqref{amiens}, \eqref{reims}, \eqref{dimancheR} and \eqref{bordeaux} for $t=T$. We set $\widetilde{\Omega}^T=\Omega$.
 \end{proposition}
\begin{proof}
We start with \eqref{amiens} for $t=T$.  As
$U_T=U$ (see \eqref{tuskes}), using Definition \ref{utilite},  $x \in \mathbb{R} \to U_{T}(\o^{T},x)$   is well-defined, non-decreasing  and usc on $\mathbb{R}$ and \eqref{amiens} for $t=T$ is true. We prove now  \eqref{reims} for $t=T$  $i.e$ that $U_{T}=U$  is $\mathcal{F}_{T} \otimes \mathcal{B}(\mathbb{R})$-measurable.  To do that we show that  for all $\o^{T} \in \Omega^{T}$, $x \in \mathbb{R} \to U_{T}(\o^{T},x)$ is right-continuous and for all $x \in \mathbb{R}$, $\o^{T} \in \Omega^{T} \to U_{T}(x,\o^{T})$ is  $\mathcal{F}_{T}$-measurable (this is just the second point  of Definition \ref{utilite}) so that we can use  Lemma \ref{cara} and establish \eqref{reims} for $t=T$.
Let $\o^{T} \in \Omega^{T} $ be fixed. From \eqref{amiens} at $T$ that we have just proved, $x \in \mathbb{R} \to U_{T}(\o^{T},x)$ is non-decreasing and usc on $\mathbb{R}$, thus applying Lemma \ref{getrightcont} we get that $x \in \mathbb{R} \to U_{T}(\o^{T},x)$ is right-continuous on $\mathbb{R}$. \\
We prove  now that \eqref{dimancheR} is true for $t=T$.
Let $\xi \in \Xi_{T-1}$ and $H=x+ \sum_{t=1}^{T-1} \phi_t \Delta S_t$ where $x \geq 0$, $\phi_1 \in \Xi_0$, \ldots,$\phi_{T-1} \in \Xi_{T-2}$ and $P_{T}(H(\cdot)+\xi(\cdot) \Delta S_{T}(\cdot)\geq 0)=1$. Let $(\phi^{\xi}_{i})_{1\leq i\leq T}  \in \Phi$ be defined by $\phi^{\xi}_{T}=\xi $ and $\phi^{\xi}_{i}=\phi_i$ for $1 \leq i \leq T-1$ then $V_{T}^{x,\phi^{\xi}}= H + \xi \Delta S_{T}$ and thus $\phi^{\xi} \in \Phi(x).$
Using Proposition \ref{propufini}  we get that  $E U^{+} (\cdot,V_{T}^{x,\phi^{\xi}}(\cdot))=E U^{+}_{T}\left(\cdot,H(\cdot)+\xi(\cdot)\Delta S_{T}(\cdot)\right)<\infty$ (recall that $U=U_{T}$). Therefore  \eqref{dimancheR} is verified for $t=T$. Finally,  from Assumption \ref{ae}, \eqref{bordeaux} for $t=T$ is true.\end{proof}\\

The next proposition proves that if \eqref{amiens}, \eqref{reims}, \eqref{dimancheR} and \eqref{bordeaux} hold true at $t+1$ then they are also true at $U_{t}$ for some well chosen  $\widetilde{\Omega}^{t}$.
\begin{proposition}\label{dyn3}
Let $0\leq t\leq T-1$ be fixed.  Assume that the (NA) condition  holds true and that \eqref{amiens}, \eqref{reims},  \eqref{dimancheR} and \eqref{bordeaux} are true at $t+1$ (where $U_{t+1}$ is defined from a given $\widetilde{\O}^{t+1}$ see \eqref{vanek}). Then there exists some  $\widetilde{\O}^{t} \in \mathcal{F}_{t}$ with $P_{t}( \widetilde{\O}^{t})=1$  such that  \eqref{amiens}, \eqref{reims}, \eqref{dimancheR} and  \eqref{bordeaux} are true for $t$.\\
 Moreover for all $H=x+\sum_{s=1}^{t} \phi_s \Delta S_s $, with $x \geq 0$ and $\phi_1 \in \Xi_0, \ldots, \phi_{t} \in \Xi_{t-1}$, such that  $P_{t}( H\geq 0)=1$ there exists some $\widetilde{\O}^{t}_{H} \in \mathcal{F}_{t}$ such that $P(\widetilde{\O}^{t}_{H})=1$, $\widetilde{\O}^{t}_{H} \subset \widetilde{\O}^{t}$ and some $\widehat{h}^H_{t+1} \in \Xi_{t}$ such that for all $\o^{t} \in \widetilde{\O}^{t}_{H}$, $\widehat{h}^H_{t+1}(\o^{t}) \in \Dc_{H(\o^{t})}^{t+1}(\o^t)$ and \footnote{Recall that the integral on the right hand side is defined in the generalised sense.}
\begin{align}
\label{noisette}
U_{t}(\omega^t,H(\omega^t)) & = \int_{\O_{t+1} }U_{t+1}(\o^t, \o_{t+1}, H(\o^t)+ \widehat{h}^H_{t+1}(\o^t) \Delta S_{t+1}(\o^t, \o_{t+1}))q_{t+1}(d\o_{t+1}|\o^t).
\end{align}
 \end{proposition}
\begin{proof}
{First we define $\widetilde{\Omega}^{t}$ and  prove that  \eqref{amiens} and  \eqref{reims} are true for $U_{t}$. Applying Proposition  \ref{dyn1}, we get  that for all $\o^{t} \in \widetilde{\Omega}_{1}^{t}$,  the function $(\omega_{t+1},x) \to U_{t+1}(\o^{t},\o_{t+1},x)$ satisfies the assumptions of  Lemma \ref{fat} and Theorem \ref{main1}  with $\overline{\Omega}=\Omega_{t+1}$, $\mathcal{H}=\mathcal{G}_{t+1}$, $Q=q_{t+1}(\cdot|\o^{t})$,  $Y(\cdot)=\Delta S_{t+1}(\o^t,\cdot)$, $V(\cdot,y)=U_{t+1}(\omega^{t},\cdot,y)$ where $V$ is defined on $\Omega_{t+1} \times \mathbb{R}$. In particular, for $\o^{t} \in \widetilde{\Omega}^{t}_{1}$ and all $h \in \mathcal{H}^{t+1}_{x}(\o^{t})$, recalling \eqref{Vplus} we have
\begin{align}
\label{Utplus}
 \int_{\Omega_{t+1}} U^{+}_{t+1}(\o^{t},\o_{t+1},x+h \Delta S_{t+1}(\o^{t},\o_{t+1}))q_{t+1}(d\o_{t+1}|\o^{t})<\infty.
 \end{align}
Now, we introduce  $\overline{U}_{t}: \Omega^{t} \times \mathbb{R}$  defined by
\begin{small}
$$\overline{U}_{t}(\o^{t},x):=  (-\infty) 1_{(-\infty,0)}(x)+ 1_{[0,\infty)}(x)1_{\widetilde{\Omega}_{1}^{t}}(\o^{t})\sup_{h \in \Dc_x^{t+1}(\o^t) } \int_{\O_{t+1} }U_{t+1}(\o^t, \o_{t+1}, x+ h \Delta S_{t+1}(\o^t, \o_{t+1}))q_{t+1}(d\o_{t+1}|\o^t).$$
\end{small}
From \eqref{Utplus}, $\overline{U}_{t}$ is well-defined (in the generalised sense).
First, we  prove  that $\overline{U}_{t}$ is  $\overline{\mathcal{F}}_{t} \otimes \mathbb{R}$-measurable and then we will  show that this implies that $U_{t}$ is $\mathcal{F}_{t} \otimes \mathbb{R}$-measurable for a well chosen $\widetilde{\Omega}^{t}$. To show that $\overline{U}_{t}$ is $\overline{\mathcal{F}}_{t} \otimes \mathcal{B}(\mathbb{R})$-measurable, we  use Lemma  \ref{cara} (and Remark \ref{overlineF}) after having proved that it is an extended Carath\'eodory function (see Definition \ref{extcara}).
Applying Theorem \ref{main1}, we get that for all $\o^{t} \in \widetilde{\Omega}_{1}^{t}$, the function $x\in \mathbb{R} \to \overline{U}_{t}(\o^{t},x)$ is non-decreasing and  usc  on $\mathbb{R}$.  Actually, this is true for all $\o^{t} \in \Omega^{t}$ since  outside  $\widetilde{\Omega}_{1}^{t}$, $x\in \mathbb{R} \to U_{t}(\o^{t},x)$ is constant equal to zero on $[0,\infty)$ and to $-\infty$ on $(-\infty,0)$.
Let now $\o^{t} \in \Omega^{t}$ be fixed. As $x \in \mathbb{R} \to \overline{U}_{t}(\o^{t},x)$ is  non-decreasing and  usc on $\mathbb{R}$ we can apply Lemma \ref{getrightcont}  and we get that  $x \in \mathbb{R} \to \overline{U}_{t}(\o^{t},x)$ is right-continuous on $\mathbb{R}$. For $x \geq 0$ fixed, applying Lemma \ref{dyn4} with $H=x$ (here $\Omega^{t}_{H}=\widetilde{\Omega}_{1}^{t}$)  we obtain that $\o^{t} \in \Omega^{t} \to \sup_{h \in \mathbb{R}^{d}} u_{x}(\o^{t},h)$ is $\overline{\mathcal{F}}_{t}$-measurable.
Finally, from the definitions of $\overline{U}_{t}$ and $u_{x}$, we get that $$\overline{{U}}_{t}(\o^{t},x)= (-\infty) 1_{(-\infty,0)}+1_{[0,\infty)}(x)1_{\widetilde{\Omega}_{1}^{t}}(\o^{t}) \sup_{h \in \mathbb{R}^{d}} u_{x}(\o^{t},h),$$
and this implies that $\o^{t} \in \Omega^t \to \overline{{U}}_{t}(\o^{t},x)$ is $\overline{\mathcal{F}}_{t}$-measurable for all $x \in \mathbb{R}$ and thus that $\overline{{U}}_{t}$ is an extended Carath\'eodory function as claimed\\
Finally, we prove the $\mathcal{F}_{t} \otimes \mathcal{B}(\mathbb{R})$-measurability of $U_{t}$. To do that we apply Lemma \ref{completemes2} and we obtain some $\Omega^{t}_{mes} \in \mathcal{F}_{t}$ such that $P_{t}(\Omega^{t}_{mes})=1$ and  some $\mathcal{F}_{t} \otimes \mathbb{R}$-measurable $\widetilde{U}_{t}: \Omega^{t} \times \mathbb{R} \to \mathbb{R} \cup \{\pm\infty\}$ such that for all $x \in \mathbb{R}$, $\{ \o^{t} \in \Omega^{t},\ \overline{U}_{t}(\o^{t},x) \neq \widetilde{U}_{t}(\o^{t},x)\} \subset \Omega^{t} \backslash{\Omega^{t}_{mes}}$. We are now in a position to define $\widetilde{\Omega}^{t}$ and set
\begin{align}
 \label{deftildeomegafirst}
 \widetilde{\Omega}^{t}:=\widetilde{\Omega}_{1}^{t} \cap\Omega^{t}_{mes}.
 \end{align}
 It is clear that $\widetilde{\Omega}^{t} \in \mathcal{F}_{t}$ and that $P_t(\widetilde{\Omega}^{t})=1$
 Furthermore, recalling \eqref{vanek}, Remark \ref{HvsD} (see \eqref{noam}) and the definition of $\overline{U}_{t}$ we have that  for all $x \in \mathbb{R}$,  $\o^{t} \in \Omega^{t}$
\begin{small}
 \begin{align*}
 U_{t}(\o^{t},x)&=  (-\infty) 1_{(-\infty,0)}(x)  +1_{[0,\infty)}(x)1_{\Omega^{t}_{mes}}(\o^{t})1_{\widetilde{\Omega}_{1}^{t}}(\o^{t})\sup_{h \in \Hc_x^{t+1}(\o^t) } \int_{\O_{t+1} }U_{t+1}(\o^t, \o_{t+1}, x+ h \Delta S_{t+1}(\o^t, \o_{t+1}))q_{t+1}(d\o_{t+1}|\o^t)\\
&=(-\infty) 1_{(-\infty,0)}(x)  +1_{[0,\infty)}(x)1_{\Omega^{t}_{mes}}(\o^{t})1_{\widetilde{\Omega}_{1}^{t}}(\o^{t})\sup_{h \in \Dc_x^{t+1}(\o^t) } \int_{\O_{t+1} }U_{t+1}(\o^t, \o_{t+1}, x+ h \Delta S_{t+1}(\o^t, \o_{t+1}))q_{t+1}(d\o_{t+1}|\o^t)\\
&=1_{\Omega^{t}_{mes}}(\o^{t}) \overline{U}_{t}(\o^{t},x)+ (-\infty) 1_{\Omega^{t} \backslash \Omega^{t}_{mes}}(\o^{t})1_{(-\infty,0)}(x)\\
&=1_{\Omega^{t}_{mes}}(\o^{t}) \widetilde{U}_{t}(\o^{t},x)+ (-\infty) 1_{\Omega^{t} \backslash \Omega^{t}_{mes}}(\o^{t}) 1_{(-\infty,0)}(x),
\end{align*}
\end{small}
and the $\mathcal{F}_{t}\otimes\mathcal{B}(\mathbb{R})$-measurability of $U_{t}$ follows immediately, $i.e$ \eqref{reims} is true at $t$. It is clear as well from the third equality that \eqref{amiens} is true for $t$ since we have proven that for all $\o^{t} \in \Omega^{t}$,  $x\in \mathbb{R} \to \overline{U}_{t}(\o^{t},x)$ is well-defined, non-decreasing and  usc on $\mathbb{R}$.}\\
We turn now to the assumption on asymptotic elasticity $i.e$ \eqref{bordeaux} for $t$.
If $\o^{t} \notin  \widetilde{\Omega}^{t}$, then  \eqref{bordeaux} is true since $C_{t}(\o^{t}) \geq 0$ for all $\o^{t}$.  Let $\o^{t} \in \widetilde{\Omega}^{t}$ be fixed. Let $x \geq 0$, $\lambda \geq 1$, $h \in \mathbb{R}^{d}$ such that $q_{t+1}(\lambda x+h\Delta S_{t+1}(\o^{t},.) \geq 0 |\o^{t})=1$ be fixed. By \eqref{bordeaux} for $t+1$ for all  $\omega_{t+1} \in {\Omega}_{t+1}$,  we have  that
$$U_{t+1}\left(\omega^{t},\omega_{t+1},\lambda x + h \Delta S_{t+1}(\omega^{t},\omega_{t+1})\right) \leq \lambda^{\overline{\gamma}}K U_{t+1}\left(\omega^{t},\omega_{t+1}, x + \frac{1}{2}+ \frac{h}{\lambda} \Delta S_{t+1}(\omega^{t},\omega_{t+1})\right) +  \lambda^{\overline{\gamma}} C_{t+1}(\o^{t},\o_{t+1}).$$
By integrating both sides (recall \eqref{Utplus})  we get that
\begin{align*}
&\int_{\Omega_{t+1}} U_{t+1}\left(\o^{t},\o_{t+1},\lambda x + h \Delta S_{t+1}(\omega^{t},\omega_{t+1})\right)q_{t+1}(d\o_{t+1}|\o^{t}) \leq \\
 &\lambda^{\overline{\gamma}}K \int_{\Omega_{t+1}} U_{t+1}\left(\o^{t},\o_{t+1},x +  \frac{1}{2}+ \frac{h}{\lambda} \Delta S_{t+1}(\omega^{t},\omega_{t+1})\right)q_{t+1}(d\o_{t+1}|\o^{t})
 +\lambda^{\overline{\gamma}}K \int_{\Omega_{t+1}}  C_{t+1}(\o^{t},\o_{t+1}) q_{t+1}(d\o_{t+1}|\o^{t}).
\end{align*}
Since $C_t(\o^{t})=\int_{\Omega_{t+1}}  C_{t+1}(\o^{t},\o_{t+1}) q_{t+1}(d\o_{t+1}|\o^{t})$  (see Lemma \ref{ToolJC}) and $h \in \mathcal{H}_{\lambda x}^{t+1}(\o^{t})$ implies that $\frac{h}{\lambda} \in \mathcal{H}_{x}^{t+1}(\o^{t}) \subset \mathcal{H}_{x+\frac{1}{2}}^{t+1}(\o^{t})$, we obtain by definition of $U_t$ (see \eqref{vanek})  that
\begin{align*}
&\int_{\Omega_{t+1}} U_{t+1}\left(\o^{t},\o_{t+1},\lambda x + h \Delta S_{t+1}(\omega^{t},\omega_{t+1})\right)q_{t+1}(d\o_{t+1}|\o^{t}) \leq
 &\lambda^{\overline{\gamma}}K U_t\left(\o^{t},x+\frac{1}{2}\right)
 +\lambda^{\overline{\gamma}} KC_{t}(\o^{t}).
\end{align*}
Taking the supremum over all $h \in \mathcal{H}_{\lambda x}^{t+1}(\o^{t})$ we conclude that \eqref{bordeaux} is true for $t$ for $x \geq 0$. If $x <0$, then  \eqref{bordeaux} is true by definition of $U_{t}$.  Note that we might have $\o^{t} \in \Omega^{t} \backslash\Omega^{t}_{C}$ and $C_{t}(\o^{t})=+\infty$ since \eqref{bordeaux}  does not require that $C_{t}(\o^{t})<+\infty$.

We now prove \eqref{noisette} for $U_t$.  First, from Proposition \ref{dyn1} and Theorem \ref{main1} and since $\widetilde{\Omega}^{t} \subset \widetilde{\Omega}_{1}^{t}$, we  have for all $\o^{t} \in \widetilde{\Omega}^{t}$ and $x\geq 0$ that there exists some $\xi^*  \in \Dc_{x}^{t+1}(\o^{t})$  such that
\begin{eqnarray}
\label{galet}
U_{t}(\o^{t},x)  & =  & \int_{\O_{t+1} }U_{t+1}(\o^{t}, \o_{t+1},x+ \xi^* \Delta S_{t+1}(\o^{t},\o_{t+1}))q_{t+1}(d\o_{t+1}|\o^{t}),
\end{eqnarray}
where the integral on the right hand side is defined in the generalised sense (recall  \eqref{Utplus} and Lemma \ref{fat}). Let $H=x+\sum_{s=1}^{t-1} \phi_s \Delta S_s $, with $x \geq 0$ and $\phi_s \in \Xi_s$ for $s \in \{1,\ldots,t-1\}$, be fixed such that $P(H\geq 0)=1$. Let
$
\widetilde{\Omega}^{t}_{H}:=\widetilde{\Omega}^{t} \cap \{\omega^{t} \in \Omega^{t}, H(\omega) \geq 0\}$. Then
$\widetilde{\Omega}^{t}_{H} \in \mathcal{F}_{t}$ and $P(\widetilde{\Omega}^{t}_{H})=1$.   We introduce the following random set $\psi: \Omega^t \twoheadrightarrow \mathbb{R}^{d}$
\begin{small}
$$ \psi_{H}(\o^{t}):= \left\{ h \in \mathcal{D}^{t+1}_{H(\o^{t})}(\o^{t}),\; U_{t}(\o^{t}, H(\o^{t}))= \int_{\Omega^{t+1}} U_{t+1}\left(\o^{t},\o_{t+1}, H(\o^{t})+h \Delta S_{t+1}(\o^{t},\o_{t+1})\right) q_{t+1}(d\o_{t+1}|\o^{t}) \right\},$$
\end{small}
for $\o^{t} \in \widetilde{\Omega}^{t}_{H}$ and  $\psi_{H}(\o^{t})= \emptyset$ otherwise.
To prove \eqref{noisette} it is enough to find a $\mathcal{F}_{t}$-measurable selector for $\psi_{H}$.
From the definitions of $\psi_{H}$ and $u_{H}$ (see \eqref{tildeu}) we obtain that (recall that
$\widetilde{\Omega}^{t}_{H} \subset \widetilde{\Omega}^{t}$ and $\widetilde{\Omega}^{t}_{H}  \subset  \Omega^{t}_{H}$, see  \eqref{deftildeomegafirst} and the definition of  $\Omega^{t}_{H}$ in Lemma \ref{dyn4}).
$$Graph(\psi_{H})=\left\{ (\o^{t},h) \in \left(\widetilde{\Omega}^{t}_{H}\times \mathbb{R}^{d}\right) \cap Graph(\mathcal{D}^{t+1}_{H}),\; U_{t}(\o^{t}, H(\o^{t}))=u_{H}(\o^{t},h) \right\}.$$
From Lemma \ref{Dhmes} we have that $Graph(\mathcal{D}^{t+1}_{H})  \in \mathcal{F}_{t}\otimes \mathcal{B}(\mathbb{R}^{d})$.
We have already proved  that   $(\o^{t},y) \to U_{t}(\o^{t},y)$ is $\mathcal{F}_{t}\otimes \mathcal{B}(\mathbb{R})$-measurable and, as $H$ is  $\Fc_t$-measurable, we obtain that $\o^{t} \to U_{t}(\o^{t},H(\o^{t}))$ is $\mathcal{F}_{t}$-measurable.
Now applying Lemma \ref{dyn4} we obtain that  $u_{H}$  is $\mathcal{F}_{t}\otimes \mathcal{B}(\mathbb{R}^{d})$-measurable.
The fact that $Graph(\psi_{H}) \in \mathcal{F}_{t}\otimes \mathcal{B}(\mathbb{R}^{d})$ follows immediately.

So we can apply  the Projection Theorem (see for example Theorem 3.23 in \cite{CV77}) and we get that $\left\{\psi_{H} \neq \emptyset \right\} \in \overline{\mathcal{F}}_{t}$ and using the Aumann Theorem (see Corollary 1 in \citet{bv}) that there exists some $\overline{\mathcal{F}}_{t}$-measurable $\overline{h}_{t+1}^{H}: \left\{\psi_{H} \neq \emptyset \right\}\to \mathbb{R}^{d}$  such that for all $\o^{t} \in \left\{\psi_{H} \neq \emptyset \right\}$, $\overline{h}_{t+1}^{H}(\o^{t}) \in \psi_{H}(\o^{t})$. Then  we  extend $\overline{h}_{t+1}^{H}$ on all $\Omega^{t}$ by setting $\overline{h}_{t+1}^{H}=0$ on $ \Omega^t \setminus \left\{\psi_{H} \neq \emptyset \right\}$. Now applying Lemma \ref{completemes} we get some $\mathcal{F}_{t}$-measurable $\widehat{h}_{t+1}^{H}: \Omega^{t} \to \mathbb{R}^{d}$  and some $\overline{\Omega}^{t}_{H} \in \mathcal{F}_{t}$ such that $P(\overline{\Omega}^{t}_{H} )=1$ and $\overline{\Omega}^{t}_{H} \subset \{\overline{h}^{H}_{t+1}=\widehat{h}^{H}_{t+1}\}$.
We prove now that the set $\left\{\psi_{H} \neq \emptyset \right\}$ is of full measure.  Indeed, let $\o^{t} \in \widetilde{\Omega}^{t}_{H}$ be fixed. Using \eqref{galet}  for $x=H(\o^{t}) \geq 0$,  there exists $h^*(\o^{t}) \in  \psi_{H}(\o^{t})$. Therefore $\widetilde{\Omega}^{t}_{H}\subset \left\{\psi_{H} \neq \emptyset \right\}$ and $\overline{P}_{t}(\left\{\psi_{H} \neq \emptyset \right\})=1$.  So for all $\o^{t} \in \overline{\Omega}^{t}_{H} \cap \widetilde{\Omega}^{t}_{H}$ we have
\begin{align*}
U_{t}(\omega^t,H(\omega^t)) &=  \int_{\O_{t+1} }U_{t+1}(\o^t, \o_{t+1}, H(\o^t)+ \overline{h}^H_{t+1}(\o^t) \Delta S_{t+1}(\o^t, \o_{t+1}))q_{t+1}(d\o_{t+1}|\o^t)\\
&=\int_{\O_{t+1} }U_{t+1}(\o^t, \o_{t+1}, H(\o^t)+ \widehat{h}^H_{t+1}(\o^t) \Delta S_{t+1}(\o^t, \o_{t+1}))q_{t+1}(d\o_{t+1}|\o^t).
\end{align*}
So setting
\begin{eqnarray}
\label{omegah}
\widetilde{\Omega}^{t}_{H}= \widetilde{\Omega}^{t}_{H} \cap \overline{\Omega}^{t}_{H} \subset \widetilde{\Omega}^{t}
\end{eqnarray}
\eqref{noisette} is proved for $t$.

We are  now  left with the proof of \eqref{dimancheR} for $U_{t}$.
Let  $\xi \in \Xi_{t-1}$ and $H=x+ \sum_{s=1}^{t-1} \phi_s \Delta S_s$ where $x \geq 0$ and $\phi_1 \in \Xi_0, \ldots,\phi_{t-1} \in \Xi_{t-2}$ and such that $P_t(H (\cdot)+  \xi(\cdot)\Delta S_{t}(\cdot)\geq 0)=1$. We fix some $\o^t \in \widetilde{\Omega}^t$. Let $X(\o^t)=H(\o^{t-1})+\xi(\o^{t-1})\Delta S_{t}(\o^{t})$ then $X $ is $\Fc_t$-measurable.
We apply \eqref{noisette}
 to $X(\o^t)$ (and $\Dc_{X(\o^t)}^{t+1}(\o^t)$), and
 we get some $\o^{t} \in \Omega^{t} \to \widehat{h}_{t+1}(\o^{t})$ which is $\mathcal{F}_{t}$-measurable and $ \widetilde{\Omega}^{t}_{X} \in \mathcal{F}_{t}$ such that $P_{t}(\widetilde{\Omega}^{t}_{X})=1$ and such  that for all $\o^{t} \in \widetilde{\Omega}^{t}_{X}$,
$q_{t+1}\left(X(\o{^t})+ \widehat{h}_{t+1}(\o^{t}) \Delta S_{t+1}(\o^{t},\cdot) \geq 0 |\o^{t}\right)=1$ and
\begin{eqnarray*}
U_{t}(\o^{t},X(\o^t))  & =  & \int_{\O_{t+1} }U_{t+1}(\o^{t}, \o_{t+1},X(\o^t)+ \widehat{h}_{t+1}(\o^{t})  \Delta S_{t+1}(\o^{t},\o_{t+1}))q_{t+1}(d\o_{t+1}|\o^{t}).
\end{eqnarray*}
Using Jensen's Inequality
\begin{align*}
U^{+}_{t}(\o^{t},X(\o^t))   \leq & \int_{\O_{t+1} }U^{+}_{t+1}(\o^{t}, \o_{t+1},X(\o^t)+ \widehat{h}_{t+1}(\o^{t}) \Delta S_{t+1}(\o^{t},\o_{t+1}))q_{t+1}(d\o_{t+1}|\o^{t}).
\end{align*}
Thus as $P_t(\widetilde{\Omega}^{t}_{X})= 1$
\begin{eqnarray*}
\int_{\widetilde{\Omega}^t_X} U^{+}_{t}(\o^{t},X(\o^t)) P_t(d\o^t) & = &  \int_{{\Omega}^t} U^{+}_{t}(\o^{t},X(\o^t)) P_t(d\o^t) \\
 & \leq & \int_{{\Omega}^{t+1}} U^{+}_{t+1}(\o^{t+1},X(\o^t)+ \widehat{h}_{t+1}(\o^{t}) \Delta S_{t+1}(\o^{t+1}))P_{t+1}(d\o^{t+1})< \infty,
\end{eqnarray*}
because of \eqref{dimancheR} for $t+1$ which applies  since $X =x+ \sum_{s=1}^{t-1} \phi_s \Delta S_s + \xi \Delta S_t$ where $x \geq 0$,  $\phi_1 \in \Xi_1, \ldots,\phi_{t-1} \in \Xi_{t-2}$, $\xi \in \Xi_{t-1}$ and $\widehat{h}_{t+1}\in \Xi_t$~: \eqref{dimancheR} for $t$ is  proved.
\end{proof}\\

The following lemma was essential to obtain measurability issues in the proof of Lemma \ref{dyn3}.
\begin{lemma}
\label{dyn4}
Fix some $0 \leq t \leq T-1$ and $x \geq 0$. Let $H:=x+ \sum_{s=1}^{t-1} \phi_s \Delta S_s$,
where $\phi_1 \in \Xi_0, \ldots,\phi_{t-1} \in \Xi_{t-2}$ and $P_{t}(H \geq 0)=1$.
Assume that the (NA) condition  holds true and that \eqref{amiens}, \eqref{reims}, \eqref{dimancheR} and \eqref{bordeaux} are true at $t+1$.
Let $u_{H}: \Omega^{t} \times \mathbb{R}^{d} \to \mathbb{R} \cup \{\pm \infty\}$ be defined by
\begin{align}
\label{tildeu}
u_{H}(\o^{t},h):=
\begin{cases}
\int_{\O_{t+1} }U_{t+1}(\o^t, \o_{t+1}, H(\o^{t})+ h \Delta S_{t+1}(\o^t, \o_{t+1}))q_{t+1}(d\o_{t+1}|\o^t),\\
\;\;\; \;\;\; \mbox{if $(\o^{t},h) \in {\left(\Omega^{t}_{H}\times \mathbb{R}^{d}\right) \cap Graph(\mathcal{D}_{H}^{t+1})}$},\\
 -\infty \; \mbox{ if $(\o^{t},h) \notin Graph(\mathcal{D}_{H}^{t+1})$},\\
 0 \;\mbox{ otherwise}.
 \end{cases}
\end{align}
where $\mathcal{D}_{H}^{t+1}$ is defined in Lemma \ref{Dhmes} and  ${\Omega}^{t}_{H}:=\widetilde{\Omega}^{t}_{1} \bigcap \{\o^{t} \in \Omega^{t},\ H(\o^{t}) \geq 0\}$ (see \eqref{deftildeomegafirst1} for the definition of $\widetilde{\Omega}^{t}_{1}).$
Then $u_{H}$ is well-defined, ${\mathcal{F}_{t} \otimes \mathcal{B}(\mathbb{R}^{d})}$-measurable  and for all $\o^{t} \in \Omega^{t}$, $h \in \mathbb{R}^{d} \to  u_{H}(\o^{t},h)$ is usc.  Morevover, $\o^{t} \in \Omega^{t} \to \sup_{h \in \mathbb{R}^{d}} u_{H}(\o^{t},h)$ is $\overline{\mathcal{F}}_{t}$-measurable. \end{lemma}
\begin{remark}
\label{remtrace}
In the proof below we will show that  for $(\o^{t},h) \in \left(\Omega^{t}_{H}\times \mathbb{R}^{d}\right) \cap Graph(\mathcal{D}_{H}^{t+1})$ the integral in \eqref{tildeu} is well-defined. Note that this is not the case for all $(\o^{t},h) \in \Omega^{t}\times \mathbb{R}^{d}$. Indeed, let $(\o^{t},h)$ be fixed such that $q_{t+1}(H(\o^{t})+ h \Delta S_{t+1}(\o^t,\cdot)<0|\o^{t})>0$. Then it is clear that $\int_{\O_{t+1} }U^{-}_{t+1}(\o^t, \o_{t+1}, H(\o^{t})+ h \Delta S_{t+1}(\o^t, \o_{t+1}))q_{t+1}(d\o_{t+1}|\o^t) =\infty$ and as without further assumption we cannot prove that  \\
$ \int_{\O_{t+1} }U^{+}_{t+1}(\o^t, \o_{t+1}, H(\o^{t})+ h \Delta S_{t+1}(\o^t, \o_{t+1}))q_{t+1}(d\o_{t+1}|\o^t) <\infty$ (it is easy to find some counterexamples),  the integral in \eqref{tildeu} may fail to be well-defined. We could have circumvented this issue by using the  convention $\infty-\infty=-\infty$ but we prefer to refrain from doing so.
\end{remark}
\begin{proof}
From \eqref{reims} at $t+1$, $U_{t+1}$ is $\mathcal{F}_{t}\otimes \mathcal{G}_{t+1} \otimes \mathcal{B}(\mathbb{R}^{d})$-measurable and since $H $ and $\Delta S_{t+1} $ are respectively $\mathcal{F}_{t}$ and $\mathcal{F}_{t+1}$-measurable, we obtain that $(\o^{t},\o_{t+1},h) \in \Omega^{t}\times \Omega_{t+1} \times \mathbb{R}^{d} \to   U_{t+1}(\o^{t},\o_{t+1}, H(\o^{t})+h\Delta S_{t+1}(\omega^{t},\o_{t+1}))$ is also $\mathcal{F}_{t}\otimes \mathcal{G}_{t+1} \otimes  \mathcal{B}(\mathbb{R}^{d})$-measurable. In order to prove that for $(\o^{t},h) \in \left(\Omega^{t}_{H}\times \mathbb{R}^{d}\right) \cap Graph(\mathcal{D}_{H}^{t+1})$ the integral in \eqref{tildeu} is well-defined, we introduce
$$\widetilde{u}_{H}: (\o^{t},h) \in \left(\Omega^{t}_{H}\times \mathbb{R}^{d}\right) \cap Graph(\mathcal{D}_{H}^{t+1}) \to  \int_{\O_{t+1} }U_{t+1}(\o^t, \o_{t+1}, H(\o^{t})+ h \Delta S_{t+1}(\o^t, \o_{t+1}))q_{t+1}(d\o_{t+1}|\o^t).$$
First we show that $\widetilde{u}_{H}$ is well-defined in the generalised sense.
Indeed, let $(\o^{t},h) \in \left(\Omega^{t}_{H}\times \mathbb{R}^{d}\right) \cap Graph(\mathcal{D}_{H}^{t+1})$ be fixed. As $\o^{t}$ is fixed in $\Omega^{t}_{H}$, we can show as in  Proposition \ref{dyn3} that \eqref{Utplus} holds true (here $H(\o^{t})$ is a fixed number as $\o^{t}$ is fixed) and thus
\begin{align*}
&\int_{\Omega_{t+1}} U^{+}_{t+1}(\o^{t},\o_{t+1}, H(\o^{t})+h\Delta S_{t+1}(\omega^{t},\o_{t+1})) q_{t+1}(d\o_{t+1}|\o^{t}) <\infty,
\end{align*}
So  $\widetilde{u}_{H}$ is well-defined (but may be infinite-valued).\\
We now prove that $u_{H}$ is  $ \mathcal{F}_{t} \otimes \mathcal{B}(\mathbb{R}^{d})$-measurable. We can apply  Proposition \ref{fubiniext} $iv)$ to $\mathcal{S}=\left(\Omega^{t}_{H}\times \mathbb{R}^{d} \right) \cap Graph(\Dc_{H}^{t+1})$,  with  $f(\o^{t},h,\o_{t+1})$ equal to both $U^{\pm}_{t+1}(\o^{t},\o_{t+1}, H(\o^{t})+h\Delta S_{t+1}(\omega^{t},\o_{t+1}))$, since $\left(\Omega^{t}_{H}\times \mathbb{R}^{d} \right)\cap Graph(\Dc_{H}^{t+1}) \in \mathcal{F}_{t} \otimes \mathcal{B}(\mathbb{R}^{d})$ (see Lemma \ref{Dhmes}), and both  $(\o^{t},h,\o_{t+1}) \in \Omega^{t}\times\mathbb{R}^{d}\times \Omega_{t+1} \to   U^{\pm}_{t+1}(\o^{t},\o_{t+1}, H(\o^{t})+h\Delta S_{t+1}(\omega^{t},\o_{t+1}))$ are $\mathcal{F}_{t}\otimes \mathcal{B}(\mathbb{R}^{d})\otimes \mathcal{G}_{t+1}$-measurable. So we obtain that
$\widetilde{u}_{H}$ is $\left[\mathcal{F}_{t} \otimes \mathcal{B}(\mathbb{R}^{d})\right]_{\mathcal{S}}$-measurable, where $\left[\mathcal{F}_{t} \otimes \mathcal{B}(\mathbb{R}^{d})\right]_{\mathcal{S}}$ denotes the trace sigma algebra of $\mathcal{F}_{t} \otimes \mathcal{B}(\mathbb{R}^{d})$ on $\mathcal{S}$.  Now we extend $\widetilde{u}_{H}$ to $\Omega^{t} \times \mathbb{R}^{d}$ by setting  $\widetilde{u}_{H}(\o^{t},h)=-\infty$ if $(\o^{t},h) \notin Graph(\mathcal{D}_{H}^{t+1})$ and  $\widetilde{u}_{H}(\o^{t},h)=0$ if $(\o^{t},h) \in Graph(\mathcal{D}_{H}^{t+1})$ and $\o^{t} \notin  {\Omega}_{H}^{t}$. Since
$\left[\mathcal{F}_{t} \otimes \mathcal{B}(\mathbb{R}^{d})\right]_{\mathcal{S}} \subset \mathcal{F}_{t} \otimes \mathcal{B}(\mathbb{R}^{d})$, ${\Omega}^{t}_{H} \in \mathcal{F}_{t}$ and  $Graph(\mathcal{D}_{H}^{t+1}) \in \mathcal{F}_{t} \times \mathcal{B}(\mathbb{R}^{d})$,  this extension of $\widetilde{u}_{H}$ is again  $ \mathcal{F}_{t} \otimes \mathcal{B}(\mathbb{R}^{d})$-measurable.
As it is clear that this extension of $\widetilde{u}_{H}$ and $u_{H}$ coincide, the measurability of $u_{H}$ is proved.\\
We turn now to the usc property. Let $\o^{t} \in \Omega^{t}_{H} \subset \widetilde{\Omega}_{1}^{t}$ be fixed.   We  apply Proposition \ref{dyn1} to $U_{t+1}$ and we get,  as $\o^{t} \in \widetilde{\Omega}_{1}^{t}$,  that the function $(\omega_{t+1},x) \to U_{t+1}(\o^{t},\o_{t+1},x)$ satisfies the assumptions of Lemma \ref{cont} (see Remark \ref{lemmecont}) with $\overline{\Omega}=\Omega_{t+1}$, $\mathcal{H}=\mathcal{G}_{t+1}$, $Q=q_{t+1}(\cdot|\o^{t})$,  $Y(\cdot)=\Delta S_{t+1}(\o^t,\cdot)$, $V(\cdot,y)=U_{t+1}(\omega^{t},\cdot,y)$ where $V$ is defined on $\Omega_{t+1} \times \mathbb{R}$. Therefore the  function $\phi_{\o^{t}}(\cdot,\cdot)$ defined on $\mathbb{R}\times\mathbb{R}^{d}$ by \begin{align*}
\phi_{\o^{t}}(x,h)=
\begin{cases}
\int_{\O_{t+1} }U_{t+1}(\o^t, \o_{t+1}, x+ h \Delta S_{t+1}(\o^t, \o_{t+1}))q_{t+1}(d\o_{t+1}|\o^t)\; \mbox{if $x \geq 0$ and $h\in D_{x}^{t+1}(\o^{t})$}\\
-\infty \; \mbox{otherwise/.}
\end{cases}
\end{align*}
is usc on $\mathbb{R}\times\mathbb{R}^{d}$ (see \eqref{psidef}). In particular, for $x=H(\o^{t})\geq 0$ fixed, the function $h \in\mathbb{R}^{d}\to u_{H}(\o^{t},h)= \phi_{\o^{t}}(H(\o^{t}),h)$ is usc on $\mathbb{R}^{d}$.
Now for $\o^{t} \notin  \Omega_{H}^{t}$, as $u_{H}$ is  equal to $0$ if $h \in \mathcal{D}_{H(\o^{t})}^{t+1}(\o^{t})$ and to $-\infty$ otherwise, Lemma  \ref{usc} applies (recall that the random set $\mathcal{D}^{t+1}_{H}$ is closed-valued)  and $h \in \mathbb{R}^{d} \to u_{H}(\o^{t},h)$ is usc on all $\mathbb{R}^{d}$.\\
Finally,  we  apply  Corollary 14.34 in \cite{rw} and find that $-u_{H}$ is a $\overline{\mathcal{F}}_{t}$- normal integrand \footnote{Corollary 14.34 of \cite{rw} holds true only for complete $\sigma$-algebra. That is the reason why $-u_{H}$ is a $\overline{\mathcal{F}}_{t}$- normal integrand and not a $\mathcal{F}_{t}$- normal integrand.}. Now from Theorem 14.37 of \cite{rw},  we obtain that
$\o^{t} \in \Omega^{t} \to \sup_{h \in \mathbb{R}^{d}} u_{H}(\o^{t},h)$ is $\overline{\mathcal{F}}_{t}$-measurable and this concludes the proof.
\end{proof}\\
\begin{proof}\emph{of Theorem \ref{main}.} We proceed in three steps. First, we handle some integrability issues that are essential to the proof. Then, we build by induction a candidate for the optimal strategy and  finally we establish its optimality. \\
\textbf{Integrability Issues}\\
We fix some $\phi \in \Phi(x)=\Phi(U,x)$ (recall Proposition \ref{propufini}). Since Proposition \ref{dyn2} holds true, we can apply Proposition \ref{dyn3} for $t=T-1$, and by  backward induction, we can therefore apply Proposition \ref{dyn3} for all $t=T-2, \dots, 0$. In particular, we get that \eqref{dimancheR} holds true for all $ 0 \leq t \leq T$. So choosing $H=V_{t-1}^{x,\phi}$ and  $\xi=\phi_{t}$ we get that (recall Remark \ref{nat}, from $\phi \in \Phi(x)$ we get that $P_{t}(V_{t}^{x,\phi}(\cdot)\geq 0)=1$)
\begin{align}
\label{fubiniforphi0}
\int_{\Omega^{t}} U^{+}_{t}\left(\omega^{t},V_{t}^{x,\phi} (\o^t)\right) P_{t}(d\o^{t})<\infty.
\end{align}
This implies that $\int_{\Omega^{t}} U_{t}\left(\omega^{t},V_{t}^{x,\phi} (\o^t)\right) P_{t}(d\o^{t})$ is defined in the generalised sense and that we can apply  the  Fubini Theorem for generalised integral  (see Proposition \ref{fubinirem})\begin{align}
\label{fubiniforphi}
\int_{\Omega^{t}} U_{t}\left(\omega^{t},V_{t}^{x,\phi} (\o^t)\right) P_{t}(d\o^{t})=\int_{\Omega^{t-1}} \int_{\Omega_{t}} U_{t}\left(\omega^{t-1},\o_{t},V_{t}^{x,\phi} (\o^{t-1},\o_{t})\right)q_{t-1}(d\o_{t}|\o^{t-1}) P_{t-1}(d\o^{t-1}).
\end{align}
\textbf{Construction of $\phi^{*}$}\\
We fix some $x \geq 0$ and  build  our candidate for the optimal strategy by  induction.
We start at $t=0$ and use \eqref{noisette} in Proposition \ref{dyn3} with $H=x \geq 0$. We set $\phi^*_1:=\widehat{h}^x_{1}$ and we obtain that (recall that $\mathcal{F}_{0}=\left\{ \emptyset, \Omega^{0} \right\}$)
$$P_{1}(x+ \phi^*_1 \Delta S_1(.)\geq 0)=1.$$
$$U_{0}(x)=\int_{\Omega_{1}} U_{1}\left(\omega_{1},x+\phi^*_1\Delta S_{1}(\omega_{1}\right)P_{1}(d\omega_{1}).$$
Recall from \eqref{fubiniforphi0} that the above integral is well-defined in the generalised sense.
Assume that until some $t \geq 1$ we have found some $\phi^*_1 \in \Xi_0,\ldots, \phi^*_{t} \in \Xi_{t-1}$ and some $\overline{\Omega}^1 \in \mathcal{F}_{1},\ldots,\overline{\Omega}^{t-1} \in \mathcal{F}_{t-1}$ such that for all $i=1,\dots,t-1$,  $\overline{\Omega}^i \subset \widetilde{\Omega}^{i}$, $P_{i}(\overline{\Omega}^i )=1$,
for all $i=0,\dots,t-1$, $\phi^{*}_{i+1}(\o^i) \in D^{i+1}(\o^{i})$ and $$P_{t}\left(x+\phi^*_1 \Delta S_{1}(\omega_{1})+ \dots + \phi^{*}_{t}(\o^{t-1}) \Delta S_{t}(\o^{t-1},\o_{t}) \geq 0\right)=1,$$
and finally, for all $\o^{t} \in \overline{\Omega}^t$  \begin{align*}
U_{t-1}\left(\omega^{t-1},V_{t-1}^{x,\phi^*} (\o^{t-1})\right)= \int_{\Omega_{t}} U_{t}\left(\omega^{t-1},\omega_{t},V_{t-1}^{x,\phi^*} (\o^{t-1}) + \phi^*_{t}(\o^{t-1}) \Delta S_{t}(\o^{t-1},\cdot)\right)q_{t}(d\o_{t}|\o^{t-1}),
\end{align*}
where again the integral is well-defined in the generalised sense (see \eqref{fubiniforphi0}).  We  apply Proposition \ref{dyn3} with  $H(\cdot)=V^{x,\phi^*}_{t}(\cdot)= V_{t-1}^{x,\phi^*} (\cdot) + \phi^*_{t}(\cdot ) \Delta S_{t}(\cdot)$ (recall that $P_{t}(V_{t}^{x,\phi^*} \geq 0=1)$ and there exists $\overline{\Omega}^{t}:=\widetilde{\Omega}^{t}_{V_{t}^{x,\phi^*}} \in \mathcal{F}_{t}$ such that  $\overline{\Omega}^{t} \subset \widetilde{\Omega}^{t}$, $P_t(\overline{\Omega}^{t}  )=1$ and some some $\Fc_t$-measurable $\o^{t} \to \phi^*_{t+1}(\o^{t}):=\widehat{h}^{V^{x,\phi^*}_t}_{t+1}(\o^{t})$  such that for all $\omega^{t} \in \overline{\Omega}^{t}$, $\phi^{*}_{t+1}(\o^{t}) \in D^{t+1}(\o^{t})$
$$q_{t+1}(V_{t}^{x,\phi^*} (\o^t) + \phi^*_{t+1}(\o^t) \Delta S_{t+1}(\o^t,\cdot) \geq 0|\o^{t})=1,$$
\begin{eqnarray}
\label{travi}
U_{t}\left(\omega^{t},V_{t}^{x,\phi^*} (\o^t)\right)= \int_{\Omega_{t+1}} U_{t+1}\left(\omega^{t},\omega_{t+1},V_{t}^{x,\phi^*} (\o^t) + \phi^*_{t+1}(\o^t) \Delta S_{t+1}(\o^t,\cdot)\right)q_{t+1}(d\o_{t+1}|\o^{t}).
\end{eqnarray}
Now since $P_t(\overline{\Omega}^{t})=1$, we obtain by the Fubini Theorem that
$$P_{t+1}(V_{t+1}^{x,\phi^*} \geq 0)=  \int_{\Omega^{t}}q_{t+1}(V_{t}^{x,\phi^*}(\o^{t}) + \phi^*_{t+1}(\o^{t}) \Delta S_{t+1}(\o^{t},\cdot) \geq 0|\o^{t}) P_{t}(d\omega^{t})=1$$
and we can continue the recursion.\\
\noindent Thus, we have found $\phi^*=(\phi^{*}_{t})_{1\leq t\leq T}$ such that for all $t=0,\dots,T$, $P_{t}(V_{t}^{x,\phi^{*}} \geq 0)=1$, $i.e$ $\phi^{*} \in \Phi(x)$. We have also found  some  $\overline{\Omega}^{t} \in \mathcal{F}_{t}$, such that $\overline{\Omega}^{t} \subset \widetilde{\Omega}^{t}$, $P_t(\overline{\Omega}^{t})=1$ and for all $\omega^{t} \in \overline{\Omega}^{t} $,
\eqref{travi} holds true for all $t=0,\dots,T-1$. Moreover, from Proposition \ref{propufini}, $\phi^{*} \in \Phi(U,x)$ and we have that $E(U(V_T^{x,\phi^*}))<\infty$. \\
 \textbf{Optimality of $\phi^*$}\\
We prove that $\phi^*$ is optimal in two steps.\\
\textit {Step 1: } Using \eqref{fubiniforphi} with $\phi=\phi^{*}$  and the fact that $P_{T-1}(\overline{\Omega^{T-1}})=1$, we get that
\begin{small}
\begin{align*}E(U(V_T^{x,\phi^*})) &=\int_{\Omega^{T-1}} \int_{\Omega_{T}} U\left(\omega^{T-1}, \omega_{T}, V_{T-1}^{x,\phi^*}(\omega^{T-1}) +\phi^{*}_{T}(\omega^{T-1})  \Delta S_T(\omega^{T-1}, \omega_{T})\right)
q_{T}(d\omega_{T}|\omega^{T-1}) P_{T-1}(d\o^{T-1})\\
&=\int_{\overline{\Omega}^{T-1}} \int_{\Omega_{T}} U_{T}\left(\omega^{T-1}, \omega_{T}, V_{T-1}^{x,\phi^*}(\omega^{T-1}) +\phi^{*}_{T}(\omega^{T-1})  \Delta S_T(\omega^{T-1}, \omega_{T})\right) q_{T}(d\omega_{T}|\omega^{T-1}) P_{T-1}(d\o^{T-1}).
\end{align*}
\end{small}
 Using   \eqref{travi} for $t=T-1$ and again the fact that $P_{T-1}(\overline{\Omega}^{T-1} )=1$,  we have that
$$E(U(V_T^{x,\phi^*})) =\int_{{\Omega}^{T-1}} U_{T-1}\left(\omega^{T-1}, V_{T-1}^{x,\phi^*}(\omega^{T-1})\right) P_{T-1}(d\o^{T-1}).$$
We iterate the process for $T-1$: using the  Fubini Theorem (see \eqref{fubiniforphi}),  $P_{T-2}(\overline{\Omega}^{T-2})=1$ and \eqref{travi}, we obtain that
$$E(U(V_T^{x,\phi^*}))=\int_{{\Omega}^{T-2}} U_{T-2}\left(\omega^{T-2}, V_{T-2}^{x,\phi^*}(\omega^{T-2})\right) P_{T-2}(d\o^{T-2}).$$
By backward induction, we therefore obtain that (recall  $\Omega^{0}:=\{\o_{0}\}$)
$$E(U(V_T^{x,\phi^*}))=U_0(x).$$
As $\phi^* \in \Phi(U,x)$, we get that  $U_0(x)\leq u(x)$. So $\phi^* $ will be optimal if $U_0(x)\geq u(x).$\\
\textit {Step 2: } We fix again some $\phi \in \Phi(U,x)$ (recall Proposition \ref{propufini}).  We get that
$V_t^{x,\phi} \geq 0$ $P_t$-a.s. for all $t=1,\dots,T$ (recall Remark \ref{nat}).  As  $\phi_1 \in \Hc_x^1$ we  obtain that
$$U_0(x) \geq  \int_{\Omega_{1}}U_1(\o_1, x+ \phi_1 \Delta S_1(\o_1)) P_{1}(d\omega_{1}).$$
As $P_2(V_1^{x,\phi} + \phi_2 \Delta S_2 \geq 0)=1$, there exists some $P_1$-full measure set $\widehat{\O}^1 \in \Fc_1$ such that for all $\o_1 \in  \widehat{\O}^1$, $q_2\left(V_1^{x,\phi}(\o_1) + \phi_2(\o_1) \Delta S_2(\o_1,\cdot)) \geq 0|\o_1\right)=1$ $i.e$
$q_{2}\left(\phi_2(\o_{1}) \in \Hc_{V_{1}^{x,\phi}(\o_{1})}^{2}(\o_{1})|\o_{1}\right)=1$ (see Lemma \ref{LemmaA1}). So for $\o_{1} \in  \widehat{\O}^1$, we have that
\begin{align}
\label{u1eq}
U_1(\o_1, V_{1}^{x,\phi}(\o_1))  \geq \int_{\Omega_{2}}U_2\left(\o_1,\o_2, V_{1}^{x,\phi}(\o_1) + \phi_2(\o_{1}) \Delta S_1(\o_1, \o_2)\right)q_{2}(d\omega_{2}|\o^1).
\end{align}
From \eqref{fubiniforphi0}, $\int_{\Omega^{2}}U^{+}_2\left(\o^{2}, V_{2}^{x,\phi}(\o^{2})\right)P_{2}(d\o^{2}) <\infty$ and we can apply the Fubini Theorem (see \eqref{fubiniforphi}) and
\begin{align*}
\int_{\Omega^{2}}U_2\left(\o^{2}, V_{2}^{x,\phi}(\o^{2})\right)P_{2}(d\o^{2})  &=\int_{\O^1} \int_{\Omega_{2}}U_2\left(\o_1,\o_2, V_{1}^{x,\phi}(\o_1) + \phi_2 \Delta S_1(\o_1, \o_2)\right)q_{2}(d\omega_{2}|\o_1)P_{1}(d\omega_{1})\\
&=  \int_{\widehat{\O}^1} \int_{\Omega_{2}}U_2\left(\o_1,\o_2, V_{1}^{x,\phi}(\o_1) + \phi_2 \Delta S_1(\o_1, \o_2)\right)q_{2}(d\omega_{2}|\o_1)P_{1}(d\omega_{1}).
\end{align*}
Using again  \eqref{fubiniforphi0}, $\int_{\Omega^{1}}U^{+}_1\left(\o^{1}, V_{1}^{x,\phi}(\o^{1})\right)P_{1}(d\o^{1}) <\infty$ and integrating (in the generalised sense) both side of \eqref{u1eq} we obtain
\begin{align*}
\int_{\O^1}U_1(\o_1, V_{1}^{x,\phi}(\o_1)) P_{1}(d\omega_{1}) &=\int_{\widehat{\O}^1}U_1(\o_1, V_{1}^{x,\phi}(\o_1)) P_{1}(d\omega_{1})\\
  &\geq \int_{\widehat{\O}^1} \int_{\Omega_{2}}U_2\left(\o_1,\o_2, V_{1}^{x,\phi}(\o_1) + \phi_2 \Delta S_1(\o_1, \o_2)\right)q_{2}(d\omega_{2}|\o_1)P_{1}(d\omega_{1})\\
&=\int_{\Omega^{2}}U_2\left(\o^2, V_{2}^{x,\phi}(\o^2)\right)P_{2}(d\omega^{2}).
\end{align*}
Therefore
\begin{align*}
U_0(x) \geq \int_{\Omega^{2}}U_2\left(\o^2, V_{2}^{x,\phi}(\o^2) \right)P_{2}(d\omega^{2}).
\end{align*}
We can go forward since for $P_{2}$-almost all $\o^{2}$ we have that $q_{3}\left(\phi_{3}(\o^{2}) \in \Hc_{V_{2}^{x,\phi}(\o^{2})}^{3}(\o^{2})|\o^{2}\right)=1$, $\dots$, for $P_{T-1}$ almost all $\o^{T-1}$ we have that $q_{T}\left(\phi_{T}(\o^{T-1}) \in \Hc_{V_{T-1}^{x,\phi}(\o^{T-1})}^{T}(\o^{T-1})|\o^{T-1}\right)=1$,  we obtain using again  \eqref{fubiniforphi0} and the Fubini Theorem (see \eqref{fubiniforphi}) that
\begin{small}
\begin{align}
\label{U0}
U_0(x) & \geq  &
\int_{\Omega_{1}}\int_{\Omega_{2}} \cdots \int_{\Omega_{T}} U\left(\omega^{T}, V_{T}^{x,\phi}(\o^T)\right) q_{T}(d\omega_{T}|\omega^{T-1}) \cdots q_{2}(d\omega_{2}|\omega^{1})P_{1}(d\omega_{1}).
\end{align}
\end{small}
So we have that
$U_0(x)  \geq  E(U(\cdot, V_T^{x,\phi}(\cdot)))$ for any $\phi \in  \Phi(U,x)$ and the proof is complete since $u(x)= E(U(\cdot, V_T^{x,\phi^*}(\cdot)))<\infty.$
\end{proof}\\

\begin{proof}\emph{of Theorem \ref{main2}.}
To prove Theorem \ref{main2}, we want to apply Theorem \ref{main} and thus we need to establish that Assumptions \ref{uisfinite} and \ref{PhiUX} hold true. To do so we will  prove \eqref{eqJt} below.
First we show that for all $x \geq 0$,  $ \phi \in \Phi(x)$ and  $0 \leq t\leq T$, we have for $P_t$-almost all $\o^{t} \in \Omega^{t}$
\begin{align}
\label{Vxphieq}
|V_{t}^{x,\phi}(\o^{t})| \leq x  \prod_{s=1}^{t}\left(1+ \frac{|\Delta S_{s}(\o^{s})|}{\alpha_{s-1}(\o^{s-1})}\right).
\end{align}
To do so we first  fix  $x \geq 0$, some $\phi=(\phi_{t})_{t=1,\dots T} \in \Phi(x)$ and $1 \leq t \leq T$. For $\o^{t-1} \in \Omega^{t-1}$ fixed,  we denote by $\phi^{\perp}_{t}(\o^{t-1})$ the orthogonal projection of  $\phi_{t}(\o^{t-1})$ on $D^{t}(\o^{t})$. Recalling Remark \ref{proj} we have $$q_{t}\left(\phi^{\perp}_{t}(\o^{t-1}) \Delta S_{t}(\o^{t-1},\cdot)=\phi_{t}(\o^{t-1}) \Delta S_{t}(\o^{t-1},\cdot)|\o^{t-1}\right)=1,$$ and thus $\phi^{\perp}_{t}(\o^{t-1}) \in \mathcal{D}_{V_{t-1}^{x,\phi}(\o^{t-1})}^{t}(\o^{t-1})$ (see  \eqref{domaineproj} for the definition of $\mathcal{D}_{x}^{t}$). As the NA condition holds true, Lemma \ref{localNA} applies and $0 \in D^{t}(\o^{t+1})$. We can then apply Lemma \ref{rast} and we obtain that \begin{align}
\label{recVt}
|\phi^{\perp}_{t}(\o^{t-1})| \leq \frac{V_{t-1}^{x,\phi}(\o^{t-1})}{\alpha_{t-1}(\o^{t-1})}.
\end{align} Furthermore, as it is well-know that $\o^{t-1} \in \Omega^{t-1} \to  \phi^{\perp}_{t}(\o^{t-1})$ is $\mathcal{F}_{t-1}$-measurable  we obtain,  applying the Fubini Theorem (see Lemma \ref{fubini0}), that $P_{t}\left(\phi^{\perp}_{t}\Delta S_{t} =\phi_{t}\Delta S_{t}\right)=1$ and we denote by $\Omega^{t}_{EQ}$ the $P_{t}$-full measure set on which this equality is verified.  We need to slightly modify the set $\Omega^{t}_{EQ}$ to use it for different periods. We proceed by induction. We start at $t=1$  (recall that $\Omega^{0}:=\{\o_{0}\}$) with $\Omega^{1}_{EQ}$. For $t=2$ we reset, with an abuse of notation, $\Omega^{2}_{EQ}=\Omega^{2}_{EQ} \cap \left(\Omega^{1}_{EQ} \times \Omega_{2}\right)$ and we reiterate the process until $T$.
To  prove  \eqref{Vxphieq}  we proceed by induction. It is clear at $t=0$. Fix some $t \geq 0$ and assume that \eqref{Vxphieq} holds true at $t$. Let $\o^{t+1} \in \Omega^{t+1}_{EQ}$,  using \eqref{Vxphieq} at $t$ and \eqref{recVt} we get that
\begin{align*}
|V_{t+1}^{x,\phi}(\o^{t+1})|&= \left| V_{t}^{x, \phi}(\o^{t}) + \phi_{t+1}(\o^{t}) \Delta S_{t+1}(\o^{t+1}) \right| = \left| V_{t}^{x, \phi}(\o^{t}) + \phi_{t+1}^{\perp}(\o^{t}) \Delta S_{t+1}(\o^{t+1}) \right|\\
& \leq  \left|V_{t}^{x,\phi}(\o^{t})\right|\left(1+ \frac{|\Delta S_{t+1}(\o^{t+1})|}{\alpha_{t}(\o^{t})} \right) \leq x  \prod_{s=1}^{t+1}\left(1+ \frac{|\Delta S_{s}(\o^{s})|}{\alpha_{s-1}(\o^{s-1})}\right)
\end{align*}
and \eqref{Vxphieq} is proven for $t+1$. It follows since for all $0 \leq s\leq t$, $|\Delta S_{s}| \in \mathcal{W}_{s}$ and $\frac{1}{\alpha_{s}} \in \mathcal{W}_{s}$  that  $V_{t}^{x,\phi} \in \mathcal{W}_{t}$. We will prove that for all
$\Phi \in \Phi(x)$ and $\o^{T}$ in a full measure set
\begin{align}
\label{eqJt}
U^{+}(\o^{T}, V_{T}^{x,\phi}(\o^{T})) \leq 2^{\overline{\gamma}} K\max(x,1)^{\overline{\gamma}} \left( \prod_{s=1}^{T}  \left(1+\frac{|\Delta S_{s}(\o^{s})|}{\alpha_{s-1}(\o^{s-1})}\right)\right)^{\overline{\gamma}} \left( U^{+}(\o^{T},1) +C_{T}(\o^{T})\right).
\end{align}
Since by assumptions $EU^{+}(\cdot,1)<\infty$, $EC_{T} <\infty$ and since for all $0 \leq t\leq T$, $|\Delta S_{t}| \in \mathcal{W}_{t}$ and $\frac{1}{\alpha_{t}} \in \mathcal{W}_{t}$,  we get that
$E U^{+}(\cdot, V_{T}^{x,\phi}(\cdot)) <\infty$ for all $\Phi \in \Phi(x)$ and both Assumptions \ref{uisfinite} and \ref{PhiUX} hold true.
We prove now \eqref{eqJt}. We fix some $x \geq 0$ and $ \phi \in \Phi(x)$. Then from the monotonicity of $U^{+}$, \eqref{Vxphieq}, Assumption \ref{ae}, the fact that $\prod_{s=1}^{T}  \left(1+\frac{|\Delta S_{s}(\o^{s})|}{\alpha_{s-1}(\o^{s-1})}\right)\geq 1$, we have for all $\o^{T} \in \Omega^{T}_{EQ}\bigcap \widetilde{\Omega}_{T}$ that
\begin{align*}
U^{+}\left(\o^{T},V_{T}^{x,\phi}(\o^{T}) \right)& \leq U^{+}\left(\o^{T}, \max(x,1) \prod_{s=1}^{T}  \left(1+\frac{|\Delta S_{s}(\o^{s})|}{\alpha_{s-1}(\o^{s-1})}\right) \right)\\
& \leq K\left(2\max(x,1) \prod_{s=1}^{T}  \left(1+\frac{|\Delta S_{s}(\o^{s})|}{\alpha_{s-1}(\o^{s-1})}\right)\right)^{\overline{\gamma}} \left( U^{+}(\o^{T},1) +C_{T}(\o^{T})\right).
\end{align*}

\end{proof}
\section{Appendix}
\label{seannexe}

In this appendix we report  basic facts about measure theory, measurable selection theorems and random sets. We also provide the proof  of some technical results.\\

\subsection{Generalised integral and Fubini's Theorem}
For ease of the reader we provide some well know results on  measure theory, stochastic kernels and integrals.  The first lemma provides a version  of the Fubini Theorem for non-negative functions (see for instance to Theorem  10.7.2 in \cite{boga}).  We then present our definition of generalised integral and provide another version of the Fubini Theorem for generalised integral (see Proposition \ref{fubinirem}), which is essential throughout the paper. \\

Let $(H, \Hc)$ and $(K,\mathcal{K})$ be two measurable spaces, $p$ be a probabilty measure on $(H,\mathcal{H})$ and $q$ a stochastic kernel on $(K,\mathcal{K})$ given $(H, \Hc)$ , $i.e$ such that for any $h \in H$, $C \in \mathcal{K} \to q(C|h)$ is a probability measure on $(K,\mathcal{K})$ and for any $C \in \mathcal{K}$, $h \in H \to q(C|h)$ is $\Hc$-measurable. Furthermore,  for any $A \in \mathcal{H} \otimes \mathcal{K}$ and any $h \in H$,  the section of $A$ along $h$ is defined by
\begin{align}
\label{setsection}
\left(A\right)_{h}&:=\left\{ k \in K,\; (h,k) \in A\right\}. \end{align}
\begin{lemma}
\label{fubini0}
Let $A \in \mathcal{H} \otimes \mathcal{K}$ be fixed. For any $h \in H$ we have $\left(A\right)_{h} \in  \mathcal{K}$ and we define $P$ by
\begin{align}
\label{Phk}
P(A)&:= \int_{H}\int_{K} 1_{A}(h,k) q(dk|h)p(dh)= \int_{H} q(\left(A\right)_{h}|h) p(dh).
\end{align}
Then $P$ is a probability measure on $(H\times K,  \mathcal{H}\otimes\mathcal{H}).$\\
Furthermore, if $f: H\times K \to \mathbb{R}_{+} \cup \{+\infty\}$ is non-negative and $ \mathcal{H} \otimes \mathcal{K}$-measurable then $h \in H \to \int_{K} f(h,k)q(dk|h)$ is $\mathcal{H}$-measurable with value in $ \mathbb{R}_{+} \cup \{\infty\}$ and we have
\begin{align}
\label{Phkint}
 \int_{H\times K} f dP:=\int_{H\times K} f(h,k) P(dh,dk) = \int_{H} \int_{K} f(h,k)q(dk|h)p(dh).
 \end{align}
  \end{lemma}
\begin{proof}
Let $h \in H$ be fixed. Let ${\cal T}=\{A \in \mathcal{H}\otimes\mathcal{K} \,| \, \left(A\right)_{h}\in \mathcal{K}\}$. It is easy to see that ${\cal T}$ is a sigma algebra on $H\times K$ and is included in $\mathcal{H}\otimes\mathcal{K}$. Let $A=B \times C \in \mathcal{H}\times\mathcal{K}$ then $\left(A\right)_{h}=\emptyset$ if $ h \notin B$ and $\left(A\right)_{h}=C$ if $ h \in B$. Thus $\left(A\right)_{h} \in \mathcal{K}$ and
$ \mathcal{H}\times\mathcal{K} \subset {\cal T}$. As  $\mathcal{T}$ is a sigma-algebra, $\mathcal{H}\otimes\mathcal{K} \subset \mathcal{T}$ and ${\cal T}= \mathcal{H}\otimes\mathcal{K}$ follows. \\
We show  now  that
$$h \to \int_{K} 1_{A}(h,k)q(dk|h)=\int_{K} 1_{\left(A\right)_{h}}(k)q(d k|h)=q\left(\left(A\right)_{h}|h\right)$$ is $\mathcal{H}$-measurable for any $A \in \mathcal{H}\otimes\mathcal{K}$. \\
Let
${\cal E}=\left\{A \in \mathcal{H}\otimes\mathcal{K} \,| \,  h \in H \to q \left(\left(A\right)_{h}|h\right) \mbox{ is $\mathcal{H}$-measurable} \right\}$. It is easy to see that ${\cal E}$ is a sigma algebra on $H\times K$ and is included in $\mathcal{H}\otimes\mathcal{K}$. Let $A=B \times C \in\mathcal{H}\times\mathcal{K}$ then $q\left(\left(A\right)_{h})|h\right)$ equals to $0$ if $h \notin B$ and to $q(C|h)$ if $ h \in B$. So by definition of $q(\cdot|\cdot)$,
$\mathcal{H}\times\mathcal{K} \subset {\cal E}$. As $\mathcal{E}$ is  a sigma-algebra, $\mathcal{H}\otimes\mathcal{K} \subset \mathcal{E}$ and ${\cal E}= \mathcal{H}\otimes\mathcal{K}$ follows. Thus the last integral in \eqref{Phk} is well-defined. We verify that $P$ defines a probability measure on $(H\times K,  \mathcal{H}\otimes\mathcal{H}).$ It is clear that $P(\emptyset)=0$ and $P(H\times K)=1$. The sigma-additivity property follows from the monotone convergence theorem.  \\
 We prove now that for $f: H\times K \to \mathbb{R}_{+} \cup \{+\infty\}$ non-negative and $ \mathcal{H} \otimes \mathcal{K}$-measurable, $h \in H \to \int_{K} f(h,k)q(dk|h)$ is $\mathcal{H}$-measurable and  \eqref{Phkint} holds true. If $f=1_A$ for $A\in \mathcal{H}\otimes\mathcal{K}$ the claim is proved. By taking linear combinations, it is proved  for $\mathcal{H}\otimes\mathcal{K}$-measurable step functions. Then if $f: H\times K \to \mathbb{R} \cup \{+\infty\}$ is non-negative and $\mathcal{H}\otimes\mathcal{K}$-measurable, then there exists some increasing sequence $(f_{n})_{n \geq 1}$ such that $f_{n}: H\times K \to \mathbb{R}$ is a $\mathcal{H}\otimes\mathcal{K}$-measurable step function and  $(f_{n})_{n \geq 1}$ converge to $f$. Using the monotone convergence theorem and  \eqref{Phkint} for steps functions, we conclude that  \eqref{Phkint} holds true for $f$.
\end{proof}\\
\begin{definition}
\label{DefInt}
Let $f: H\times K \to \mathbb{R} \cup \{\pm\infty\}$ be a $ \mathcal{H}\otimes\mathcal{K}$-measurable function. If $ \int_{H\times K} f^{+} dP< \infty$ or $ \int_{H\times K} f^{-}dP< \infty$,  we define the  generalised integral of $f$ by
$$ \int_{H\times K} f dP :=  \int_{H\times K} f^{+}dP -  \int_{H\times K} f^{-}dP.$$
\end{definition}
\begin{remark}
Note  that if both $ \int_{H\times K} f^{+}dP= \infty$ and $ \int_{H\times K} f^{-} dP= \infty$, the integral above is not defined. We could have introduced some convention to handle this situation, however, as in most of the cases we treat we have $ \int_{H\times K} f^{+} dP< \infty$, we refrain from doing so. \end{remark}
\begin{proposition}
\label{fubinirem}
Let $f: H\times K \to \mathbb{R} \cup \{\pm\infty\}$ be a $ \mathcal{H}\otimes\mathcal{K}$-measurable function such that $ \int_{H\times K} f^{+}dP< \infty$.  Then, we have
\begin{align}
\label{fubint}
 \int_{H \times K}f dP= \int_{H} \int_{K} f(h,k) q(dk|h) p(dh).
 \end{align}
\end{proposition}
\begin{remark}
\label{remneg}
Note that we can assume instead  that $ \int_{H\times K} f^{-}dP< \infty$ and the result holds as well. We will use this in the proof of Lemma \ref{lienespcond} later in the Appendix.
\end{remark}
\begin{proof}
Using Definition \ref{DefInt} and applying Lemma \ref{fubini0} to $f^{+}$ and $f^{-}$ we obtain that
\begin{align*}
 \int_{H \times K}f dP&=  \int_{H \times K}f^{+} dP- \int_{H \times K}f^{-} dP\\
 &=\int_{H} \int_{K} f^{+} q(dk|h) p(dh)+\int_{H} \int_{K} f^{-} q(dk|h) p(dh).
  \end{align*}
 To establish \eqref{fubint}, assume for a moment that the followng linearity result have been proved:  let $g_{i}: H\times K \to \mathbb{R} \cup \{\pm\infty\}$ be some $ \mathcal{H}\otimes\mathcal{K}$-measurable functions such that $ \int_{H\times K} g_{i}^{+}dP< \infty$  for $i=1,2$. Then
\begin{align}
\label{intlinear}
\int_{H}  \left(g_{1} +g_{2}\right) dp= \int_{H}g_{1} dp+ \int_{H} g_{2} dp.
\end{align}
We apply \eqref{intlinear} with $g_{1}(h)=\int_{K}f^{+}(h,k)q(dh|k)$ and $g_{2}=-\int_{K}f^{-}(h,k)q(dh|k)$ since by Lemma \ref{fubini0},
\begin{eqnarray*}
 \int_{H}g^{+}_{1}dp & =  & \int_{H} \left(\int_{K}f^{+}(h,k)q(dh|k)\right) p(dh)\\
 &= &  \int_{H \times K}f^{+}(h,k)q(dh|k)p(dh)=\int_{H \times K} f^{+}dP<\infty
\end{eqnarray*}
and clearly  $ \int_{H}g^{+}_{2}dp=0<\infty$.
So we obtain that
\begin{align*}
 \int_{H} \int_{K}& f^{+}(h,k) q(dk|h) p(dh) - \int_{H} \int_{K} f^{-}(h,k) q(dk|h) p(dh)\\
 &=  \int_{H}\left( \int_{K} f^{+}(h,k)q(dk|h) - \int_{K}f^{-}(h,k)q(dk|h)\right)p(dh)\\
 &= \int_{H} \int_{K} f(h,k) q(dk|h) p(dh),
 \end{align*}
where the second equality comes from the definition of the generalised integral of $f(h,\cdot)$ with respect to $q(\cdot|h)$ and  \eqref{fubint} is proven.\\
We prove now \eqref{intlinear}. If $\int_{H} g_{i}^{-} dp <\infty$ for $i=1,2$ this is trivial.
From  $\int_{H }g_{i}^{+} dp <\infty$ we get that $g^{+}_{i} <\infty$ $p$-almost surely for $i=1,2$, so the sum $g_{1}+g_{2}$ is $p$-almost surely well-defined, taking its value in $[-\infty,\infty)$. As $(g_{1}+g_{2})^{+} \leq g_{1}^{+}+g_{2}^{+}$, using   the linearity of the integral for non-negative functions we get that
\begin{align*}
 \int_{H} \left(g_{1}+g_{2}\right)^{+}(h) p(dh)\leq \int_{H }g_{1}^{+} dp +\int_{H} g_{2}^{+} dp <\infty.
\end{align*}
Now from
$$g_{1}^{+}+g_{2}^{+}-g_{1}^{-}-g_{2}^{-}=g_{1}+g_{2} =\left(g_{1}+g_{2}\right)^{+}-\left(g_{1}+g_{2}\right)^{-},$$ using again the linearity of the integral for non-negative functions we get that
$$ \int_{H} \left(g_{1}+g_{2}\right)^{+} dp + \int_{H} g^{-}_{1} dp  +  \int_{H}g^{-}_{2}dp=\int_{H} \left(g_{1}+g_{2}\right)^{-}dp + \int_{H}  g^{+}_{1}dp + \int_{H} g^{+}_{2}dp.$$
Checking the different cases, $i.e$ $\int_{H}g_{1}^{-} dp=\infty$ and $\int_{H}g_{2}^{-} dp<\infty$ (and the opposite case) as well as $\int_{H}g_{i}^{-} dp=\infty$ for $i=1,2$ we get that \eqref{intlinear} is true.
\end{proof}

\subsection{Further measure theory issues}
\label{apap}
We present now  specific  applications or results that are used throughout the paper.
We  start with four extensions of the Fubini results presented previously. As noted in Remark \ref{remtrace}, the introduction of the trace sigma-algebra is the price to pay in order to avoid using the convention $\infty-\infty=-\infty$.
\begin{proposition}
\label{fubiniext}
Fix some $t\in\{1,\ldots,T\}$.
\begin{itemize}
\item[i)] Let $f: \O^t \to \mathbb{R}_{+}\cup\{+\infty\}$ be a non-negative $\Fc_t$-measurable function. Then $\o^{t-1} \in \O^{t-1} \to \int_{\Omega_t}f(\o^{t-1},\o_t)q_t(d \o_t|\o^{t-1})$ is $\Fc_{t-1}$-measurable with values in $\mathbb{R}_{+}\cup\{+\infty\}$.
\item[ii)]Let $f:\O^t \times \mathbb{R}^{d} \to \mathbb{R}_{+}\cup \{+ \infty\}$ be a non-negative $\Fc_t \otimes \mathcal{B}(\mathbb{R}^{d})$-measurable function. Then  $(\o^{t-1},h) \in \O^{t-1}\times\mathbb{R}^{d} \to \int_{\Omega_t}f(\o^{t-1},\o_t,h)q_t(d \o_t|\o^{t-1})$ is $\Fc_{t-1}\otimes \mathcal{B}(\mathbb{R}^{d})$-measurable with values in $\mathbb{R}_{+} \cup \{ +\infty\}$
\item[iii)]Let $f: \O^t \to \mathbb{R}_{+}\cup\{+\infty\}$ be a non-negative $\overline{\Fc}_{t-1}\otimes \mathcal{G}_{t}$-measurable function. Then $\o^{t-1} \in \O^{t-1} \to \int_{\Omega_t}f(\o^{t-1},\o_t)q_t(d \o_t|\o^{t-1})$ is $\overline{\Fc}_{t-1}$-measurable with values in $\mathbb{R}_{+}\cup\{+\infty\}$.
\item[iv)]Let  $\mathcal{S}  \in \mathcal{F}_{t-1} \otimes \mathcal{B}(\mathbb{R}^{d})$.  Introduce  $\left[\mathcal{F}_{t-1} \otimes \mathcal{B}(\mathbb{R}^{d})\right]_{\mathcal{S}}:=\left\{ A \cap \mathcal{S},\; A \in  \mathcal{F}_{t-1} \otimes \mathcal{B}(\mathbb{R}^{d})\right\}$ the trace sigma-algebra of $ \mathcal{F}_{t-1} \otimes \mathcal{B}(\mathbb{R}^{d})$ on $\mathcal{S}$.  Let $f: \O^{t-1} \times \mathbb{R}^{d} \times \O_{t}  \to \mathbb{R}_{+}\cup \{+ \infty\}$ be a non-negative $\Fc_{t-1} \otimes \mathcal{B}(\mathbb{R}^{d})  \otimes \mathcal{G}_{t} $-measurable function. Then  $(\o^{t-1},h) \in \mathcal{S} \to \int_{\Omega_t}f(\o^{t-1},h,\o_t)q_t(d \o_t|\o^{t-1})$ is $\left[\mathcal{F}_{t-1} \otimes \mathcal{B}(\mathbb{R}^{d})\right]_{\mathcal{S}}$-measurable with values in $\mathbb{R}_{+} \cup \{ +\infty\}$.
\end{itemize}
\end{proposition}
\begin{proof}
Statement $i)$ is a direct application of Lemma \ref{fubini0} for $H=\O^{t-1}$, $\mathcal{H}=\mathcal{F}_{t-1}$, $K=\Omega_{t}$, $\mathcal{K}=\mathcal{G}_{t}$ and $q(\cdot|\cdot)=q_{t}(\cdot|\cdot)$.
To prove statement $ii)$, let $\bar{q}_{t}$ be defined by
\begin{align}
\label{barq}
\bar{q}_{t}: (G,\omega^{t-1},h) \in \mathcal{G}_{t}\times \Omega^{t-1}\times \mathbb{R}^{d} \to  \bar{q}_{t}(G| \o^{t-1},h):=q_{t}(G|\omega^{t-1}).
\end{align}
We first prove that $\bar{q}_{t}$ is a stochastic kernel on $\mathcal{G}_{t}$ given $\Omega^{t-1}\times \mathbb{R}^{d}$ where measurability is with respect to $\Fc_{t-1} \otimes \mathcal{B}(\mathbb{R}^{d})$. Let $(\o^{t-1},h) \in \O^{t-1}\times \mathbb{R}^{d}$ be fixed, $B \in \mathcal{G}_{t} \to \bar{q}_{t}(B|\o^{t-1},h)={q}_{t}(B|\o^{t-1})$ is a probability measure on $(\O_{t},\mathcal{G}_{t})$ by definition of $q_{t}$. Let $B \in \mathcal{G}_{t}$ be fixed, then $(\o^{t-1},h) \in \O^{t-1}\times\mathbb{R} \to  \bar{q}_{t}(B|\o^{t-1},h)={q}_{t}(B|\o^{t-1})$ is $\mathcal{F}_{t-1}\otimes\mathcal{B}(\mathbb{R}^{d})$-measurable since for any $B' \in \mathcal{B}(\mathbb{R})$, we have, by definition of $q_{t}$,
$$\left\{(\o^{t-1},h) \in \O^{t-1}\times  \mathbb{R}^{d}, \bar{q}_{t}(B|\o^{t-1},h) \in B' \right\}= \left\{\o^{t-1} \in \O^{t-1}, q_{t}(B|\o^{t-1}) \in B' \right\} \times \mathbb{R}^{d} \in \Fc_{t-1} \otimes \mathcal{B}(\mathbb{R}^{d}).$$
Statement $ii)$ follows by an application of Lemma \ref{fubini0} for $H=\O^{t-1}\times \mathbb{R}^{d}$, $\mathcal{H}=\Fc_{t-1} \otimes \mathcal{B}(\mathbb{R}^{d})$, $K=\Omega_{t}$, $\mathcal{K}=\mathcal{G}_{t}$ and $q(\cdot|\cdot)=\bar{q}_{t}(\cdot|\cdot)$.
To prove statement $iii)$ note that since $\mathcal{F}_{t-1} \subset \overline{\mathcal{F}}_{t-1}$ it is clear that $q_{t}$ is a stochastic kernel on $(\Omega_{t},\mathcal{G}_{t})$ given $(\Omega^{t-1},\overline{\Fc}_{t-1})$ ($i.e$  measurability is with respect to $\overline{\Fc}_{t-1}$). And statement $iii)$ follows immediately from an application of Lemma \ref{fubini0} for $H=\O^{t-1}$, $\mathcal{H}=\overline{\mathcal{F}}_{t-1}$, $K=\Omega_{t}$, $\mathcal{K}=\mathcal{G}_{t}$ and $q(\cdot|\cdot)=q_{t}(\cdot|\cdot)$.
We prove now the last statement.
It is well known that $(S, \left[\mathcal{F}_{t-1} \otimes \mathcal{B}(\mathbb{R}^{d})\right]_{\mathcal{S}})$ is a measurable space.
Let $\tilde{q}_{t}$ be defined by
\begin{align}
\label{tildeq}
\tilde{q}_{t}: (G,\omega^{t-1},h) \in \mathcal{G}_{t}\times S\to  \tilde{q}_{t}(G| \o^{t-1},h):=q_{t}(G|\omega^{t-1}).
\end{align}
We prove that  $\tilde{q}_{t}$ is a stochastic kernel on $(\Omega_{t},\mathcal{G}_{t})$ given $\left(\mathcal{S},\left[\mathcal{F}_{t-1} \otimes \mathcal{B}(\mathbb{R}^{d})\right]_{\mathcal{S}}\right)$.  Indeed, let $(\o^{t-1},h) \in S$ be fixed, $B \in \mathcal{G}_{t} \to \tilde{q}_{t}(B|\o^{t-1},h)={q}_{t}(B|\o^{t-1})$ is a probability measure on $(\O_{t},\mathcal{G}_{t})$, by definition of $q_{t}$. Let $B \in \mathcal{G}_{t}$ be fixed, then $(\o^{t-1},h) \in S \to  \tilde{q}_{t}(B|\o^{t-1},h)={q}_{t}(B|\o^{t-1})$ is $\left[\mathcal{F}_{t-1} \otimes \mathcal{B}(\mathbb{R}^{d})\right]_{\mathcal{S}}$-measurable since for any $B' \in \mathcal{B}(\mathbb{R})$, we have, by definition of $q_{t}$ \begin{eqnarray*}
\left\{(\o^{t-1},h) \in S,\; \tilde{q}_{t}(B|\o^{t-1},h) \in B' \right\} & = & \left(\left\{\o^{t-1} \in \O^{t-1}, q_{t}(B|\o^{t-1}) \in B' \right\} \times \mathbb{R}^{d}\right) \bigcap S\\
& \in & \left[\mathcal{F}_{t-1} \otimes \mathcal{B}(\mathbb{R}^{d})\right]_{\mathcal{S}}.
\end{eqnarray*}
Now let $f_{S}$ be the restriction of $f$ to  $\mathcal{S}\times \O_{t}$. Using similar arguments and the fact that
 \begin{align}
 \label{trace}
 \left[\mathcal{F}_{t-1} \otimes \mathcal{B}(\mathbb{R}^{d})\otimes \mathcal{G}_{t}\right]_{\mathcal{S}\times \O_{t}}=\left[\mathcal{F}_{t-1} \otimes \mathcal{B}(\mathbb{R}^{d})\right]_{\mathcal{S}}\otimes \mathcal{G}_{t},
 \end{align}
we obtain that $f_{S}$ is $\left[\mathcal{F}_{t-1} \otimes \mathcal{B}(\mathbb{R}^{d})\right]_{\mathcal{S}}\otimes \mathcal{G}_{t}$-measurable.
Finally, statement $iv)$ follows from another  application of Lemma \ref{fubini0} for $H=S$, $\mathcal{H}=\left[\mathcal{F}_{t-1} \otimes \mathcal{B}(\mathbb{R}^{d})\right]_{\mathcal{S}}$, $K=\Omega_{t}$, $\mathcal{K}=\mathcal{G}_{t}$ and $q(\cdot|\cdot)=\tilde{q}_{t}(\cdot|\cdot)$.
\end{proof}\\
\begin{lemma}
\label{lemmaAA}
Let $f : \Omega^{t+1} \to \mathbb{R}_{+} \cup \{\infty\}$ be $\mathcal{F}_{t+1}$-measurable, non-negative and such that \\ $\int_{\Omega^{t+1}} f(\o^{t+1}) P_{t+1}(d\o^{t+1}) <\infty$.
Then $\o^{t} \in \Omega^{t} \to \int_{\Omega_{t+1}} f(\o^{t},\o_{t+1})q_{t+1}(d\o_{t+1}|\o^{t})$ is $\mathcal{F}_{t}$-measurable.
Furthermore, let $$N^{t}:=\{\o^{t} \in \O^{t},\; \int_{\Omega_{t+1}} f(\o^{t},\o_{t+1})q_{t+1}(d\o_{t+1}|\o^{t})=\infty\}.$$ Then $N_{t} \in \mathcal{F}_{t}$ and $P_{t}(N^{t})=0$
\end{lemma}
\begin{proof}
The first assertion of the lemma is a direct application of $i)$ of Proposition \ref{fubiniext}.
So it is clear that $N^{t} \in \mathcal{F}_{t}$.
Furthermore, applying the Fubini Theorem (see Lemma \ref{fubini0}) we get that
$$\int_{\Omega^{t}} \int_{\Omega_{t+1}} f(\o^{t},\o_{t+1})q_{t+1}(d\o_{t+1}|\o^{t}) P_{t}(d\o^{t})=\int_{\Omega^{t+1}} f(\o^{t+1}) P_{t+1}(d\o^{t+1}) <\infty.$$
Assume that $P_{t}(N^{t})>0$. Then
$$\int_{\Omega^{t+1}} f(\o^{t+1}) P_{t+1}(d\o^{t+1}) \geq \int_{N^{t}} \int_{\Omega_{t+1}} f(\o^{t},\o_{t+1})q_{t+1}(d\o_{t+1}|\o^{t}) P_{t}(d\o^{t})=\infty.$$
We get a contradiction~: $P_{t}(N^{t})=0$.
\end{proof}\\

The next lemma, loosely speaking,  allows to obtain ``nice" sections ($i.e$ set of full measure for a certain probability measure). We use it in the proofs of Theorem \ref{main2} and Lemma \ref{LemmaA1}.
\begin{lemma}
\label{bouillante2}
Fix some $t\in\{1,\ldots,T\}$.
Let $\widetilde{\Omega}^t \in \Fc_t$ such that $P_t(\widetilde{\Omega}^t)=1$ and  $\widetilde{\Omega}^{t-1} \in \Fc_{t-1}$ such that $P_{t-1}(\widetilde{\Omega}^{t-1})=1$ and set
$$\overline{\Omega}^{t-1}:=\left\{ \o^{t-1} \in \widetilde{\Omega}^{t-1},\; q_t\left(\left(\widetilde{\Omega}^t\right)_{\o^{t-1}}|\o^{t-1}\right)=1\right\}$$
see Lemma \ref{fubini0} for the definition of $\left(\widetilde{\Omega}^t\right)_{\o^{t-1}}$.
Then $\overline{\Omega}^{t-1} \in \mathcal{F}_{t-1}$ and $P_{t}(\overline{\Omega}^{t-1})=1$.
\end{lemma}
\begin{proof}
From Lemma \ref{fubini0} we know $\o^{t-1} \to q_t\left(\left(\widetilde{\Omega}^t\right)_{\o^{t-1}}|\o^{t-1}\right)$ is $\mathcal{F}_{t-1}$-measurable and the fact that $\overline{\Omega}^{t-1} \in \mathcal{F}_{t-1}$ follows immediately.\\
Furthermore, using the Fubini Theorem (see Lemma \ref{fubini0})  we have  that
\begin{align*}
1=P_{t}(\widetilde{\Omega}^{t})&=\int_{\Omega^{t-1}}\int_{\Omega_{t}} 1_{\widetilde{\Omega}^t}(\o^{t-1},\o_{t})q_t(d\o_{t}|\o^{t-1})P_{t-1}(d\o^{t-1})\\
&=\int_{\Omega^{t-1}}\int_{\Omega_{t}} 1_{\left(\widetilde{\Omega}^t\right)_{\o^{t-1}}}(\o_{t})q_t(d\o_{t}|\o^{t-1})P_{t-1}(d\o^{t-1})\\
&= \int_{\widetilde{\Omega}^{t-1}}\int_{\Omega_{t}} 1_{\left(\widetilde{\Omega}^t\right)_{\o^{t-1}}}(\o_{t})q_t(d\o_{t}|\o^{t-1})P_{t-1}(d\o^{t-1})\\
&=\int_{\widetilde{\Omega}^{t-1}}q_t\left(\left(\widetilde{\Omega}^t\right)_{\o^{t-1}}|\o^{t-1}\right)P_{t-1}(d\o^{t-1})\\
&=\int_{\overline{\Omega}^{t-1}} 1 \times P_{t-1}(d\o^{t-1})+ \int_{\widetilde{\Omega}^{t-1}\backslash{\overline{\Omega}^{t-1}}}q_t\left(\left(\widetilde{\Omega}^t\right)_{\o^{t-1}}|\o^{t-1}\right)P_{t-1}(d\o^{t-1}),
\end{align*}
where we have used for the third line the fact that $P(\widetilde{\Omega}^{t-1})=1$.\\
But if $P(\widetilde{\Omega}^{t-1}\backslash{\overline{\Omega}^{t-1}})>0$ then we have that by definition of $\overline{\Omega}^{t-1}$ that $$\int_{\widetilde{\Omega}^{t-1}\backslash{\overline{\Omega}^{t-1}}}q_t\left(\left(\widetilde{\Omega}^t\right)_{\o^{t-1}}|\o^{t-1}\right)P_{t-1}(d\o^{t-1})< P_{t-1}(\widetilde{\Omega}^{t-1}\backslash{\overline{\Omega}^{t-1}}),$$ and thus \begin{align*}
1<P_{t-1}(\overline{\Omega}^{t-1})+P_{t-1}(\widetilde{\Omega}^{t-1}\backslash{\overline{\Omega}^{t-1}})=1,
\end{align*}
which is absurd and thus $P_{t-1}(\widetilde{\Omega}^{t-1}\backslash{\overline{\Omega}^{t-1}})=0$. We conclude using again that $P_{t-1}(\widetilde{\Omega}^{t-1})=1$.
\end{proof}\\
The following lemma is used throughout the paper. In particular, the last statement is used in the proof of the main theorem
\begin{lemma}
\label{LemmaA1}
Let $0\leq t \leq T-1$,  $B \in \mathcal{B}(\mathbb{R})$, $H: \Omega^{t} \to \mathbb{R}$  and  $h_{t}: \Omega^{t} \to \mathbb{R}^{d}$ be  $\mathcal{F}_{t}$-measurable be fixed. Then the functions
\begin{align}
\label{eq1}
& (\omega^{t},h) \in \O^{t} \times \mathbb{R}^{d} \to q_{t+1}(H(\o^{t})+h\Delta S_{t+1}(\o^t,\cdot) \in B|\o^t),\\
\label{eq2}
& \o^{t} \in \O^{t} \to q_{t+1}(H(\o^{t})+h_{t}(\o^{t})\Delta S_{t+1}(\o^t,\cdot) \in B|\o^t),
 \end{align}
are respectively $\Fc_{t} \otimes \mathcal{B}(\mathbb{R}^{d})$-measurable and $\mathcal{F}_{t}$-measurable.
Furthermore, assume that \\ $P_{t+1}\left(H(\cdot)+h_{t}(\cdot)\Delta S_{t+1}(\cdot) \in B\right)=1$, then there exists some $P_{t}$-full measure set $\overline{\Omega}^{t}$ such that for all $\o^{t} \in \overline{\Omega}^{t}$, $q_{t+1}(H(\o^{t})+h_{t}(\o^{t})\Delta S_{t+1}(\o^t,\cdot) \in B|\o^t)=1$.
 \end{lemma}
\begin{proof}
As  $h \in \mathbb{R}^{d} \to h \Delta S_{t+1}(\o^{t},\o_{t+1})$ is continuous for all $(\o^{t},\o_{t+1}) \in \O^{t}\times\O_{t+1}$ and $(\o^{t},\o_{t+1}) \in  \O^{t}\times\O_{t+1} \to h \Delta S_{t+1}(\o^{t},\o_{t+1})$ is $\mathcal{F}_{t+1}=\mathcal{F}_{t} \otimes \mathcal{G}_{t+1}$-measurable for all $h \in \mathbb{R}^{d}$ (recall that $S_{t} $ and $S_{t+1} $ are respectively $\Fc_t$ and $\Fc_{t+1}$ measurable by assumption),  $(\omega^{t},\omega_{t+1},h)  \in \Omega^{t}\times\Omega_{t+1}\times\mathbb{R}^{d} \to h\Delta S_{t+1}(\o^t,\o_{t+1})$ is $\mathcal{F}_{t}\otimes \mathcal{G}_{t+1} \otimes \mathcal{B}(\mathbb{R}^{d})$-measurable as a Carath\'eodory function. As $H$ is $\mathcal{F}_{t}$-measurable we obtain that $\psi: (\o^{t},\o_{t+1},h) \in  \Omega^{t}\times\Omega_{t+1}\times\mathbb{R}^{d} \to H(\o^{t}) + h\Delta S_{t+1}(\o^t,\o_{t+1})$ is also $\mathcal{F}_{t}\otimes \mathcal{G}_{t+1} \otimes \mathcal{B}(\mathbb{R}^{d})$-measurable. Therefore, for any $B \in \mathcal{B}(\mathbb{R})$, $f_{B}: (\o^{t},\o_{t+1},h) \in  \Omega^{t}\times\Omega_{t+1}\times\mathbb{R}^{d} \to 1_{\psi(\cdot,\cdot,\cdot) \in B}(\o^{t},\o_{t+1},h)$ is $\mathcal{F}_{t}\otimes \mathcal{G}_{t+1} \otimes \mathcal{B}(\mathbb{R}^{d})$. We conclude using statement $i)$ of Proposition \ref{fubiniext} applied to $f_{B}$ and \eqref{eq1} is proved.
We prove \eqref{eq2} using similar arguments. Since $h_{t}$ is $\mathcal{F}_{t}$-measurable, it is clear that $\psi_{h_{t}}: (\o^{t},\o_{t+1}) \in  \Omega^{t}\times\Omega_{t+1} \to H(\o^{t}) + h_{t}(\o^{t})\Delta S_{t+1}(\o^t,\o_{t+1})$ is  $\mathcal{F}_{t}\otimes \mathcal{G}_{t+1}$-measurable.  Therefore, for any $B \in \mathcal{B}(\mathbb{R})$, $f_{B,h_{t}}: (\o^{t},\o_{t+1}) \in  \Omega^{t}\times\Omega_{t+1}\to 1_{\psi_{h_{t}}(\cdot,\cdot) \in B}(\o^{t},\o_{t+1})$ is $\mathcal{F}_{t}\otimes \mathcal{G}_{t+1}$-measurable. We conclude applying $i)$   of Proposition  \ref{fubiniext}  to  $f_{B,h_{t}}$.\\
For the last statement, we set
$$\widetilde{\Omega}^{t+1}:=\left\{\o^{t+1}=(\o^{t},\o_{t+1}) \in \O^{t}\times \O_{t+1},\; H(\o^{t})+h_{t}(\o^{t})\Delta S_{t+1}(\o^{t},\o_{t+1}) \in B\right\}.$$ It is clear that $\widetilde{\Omega}^{t+1} \in \mathcal{F}_{t+1}$ and that $P_{t+1}(\widetilde{\Omega}^{t+1})=1$. We can then apply Lemma \ref{bouillante2} and we obtain some $P_{t}$-full measure set $\overline{\Omega}^{t}$ such that for all $\o^{t} \in \overline{\Omega}^{t}$, $q_{t+1}(H(\o^{t})+h_{t}(\o^{t})\Delta S_{t+1}(\o^t,\cdot) \in B|\o^t)=1$.
 \end{proof}\\

Lemma  \ref{completemes} is often used in conjunction with the Aumann Theorem (see Corollary 1 in \citet{bv}) to obtain a $\mathcal{F}_{t}$-measurable selector.

\begin{lemma}																			
\label{completemes}
Let  $f:\Omega^t \to \mathbb{R} $ be $\overline{\Fc}_{t}$-measurable. Then there exists $g: \Omega^t \to \mathbb{R} $ that is ${\Fc}_{t}$-measurable and such that $f=g$ $P_{t}$-almost surely, $i.e$ there exists $\Omega^{t}_{fg} \in \mathcal{F}_{t}$ with $P_{t}\left(\Omega^{t}_{fg}\right)=1$ and $\Omega^{t}_{fg} \subset \{f=g\}$.
\end{lemma}
\begin{proof}
Let $f=1_B$ with $B \in \overline{\Fc}_{t}$ then $B=A\cup N$, with $A \in \mathcal{F}_{t}$ and $N \in {\cal N}_{P_{t}}$. Let $g=1_A$. Then
$g$ is ${\Fc}_{t}$-measurable. Clearly, $\{f \neq g\}=N \in {\cal N}_{P_{t}}$, thus $f=g$ $P_{t}$ a.s.
By taking linear combinations, the lemma is proven for step functions using the same argument for each indicator function.
Then it is always possible to approximate some $\overline \Fc_t$-measurable function $f$ by a sequence of step function $(f_n)_{n \geq 1}$.  From the preceding step for all $n \geq 1$, we get some  $\mathcal{F}_{t}$-measurable step functions $g_n$ such that $f_n=g_n$ $P_{t}$-almost surely. Let  $g = \limsup g_n$, $g$ is $\mathcal{F}_{t}$-measurable and we conclude since  $\{f \neq g\} \subset \cup_{n \geq 1} \{f_n \neq g_n\}$ which is again in ${\cal N}_{P_{t}}$.\end{proof}\\

Next we provide some simple but  useful results on usc functions.
\begin{lemma}
\label{usc}
Let $C$ be a closed subset  of $\mathbb{R}^{m}$ for some $m \geq 1$. Let $g: \mathbb{R}^{m} \to \mathbb{R} \cup \{\pm \infty\}$ be such that $g=-\infty$ on $\mathbb{R}^{m} \backslash{C}$. Then $g$ is usc on $\mathbb{R}^{m}$ if and only if $g$ is  usc on $C$.
\end{lemma}
\begin{proof}
We prove that if  $g$ is usc on $C$ then it is usc on $\mathbb{R}^{m}$ as the reverse implication is trivial. Let $\alpha \in \mathbb{R}$ be fixed. We prove that $S_{\alpha}:=\{x \in \mathbb{R}^{m} ,\ g(x) \geq \alpha\}$ is closed in $\mathbb{R}^{m}$. Let $(x_{n})_{n \geq 1} \subset S_{\alpha}$ converge to $x \in \mathbb{R}^{m}$. Then $x_{n} \in C$ for all $n \geq 1$ and as $C$ is a closed set, $x \in C$. As $g$ is  usc on $C$, ($i.e$ the set $\{x \in C,\ g(x) \geq \alpha\}$ is closed for the induced topology of $\mathbb{R}^{m}$ on $C$) we get that $g(x) \geq \alpha$, $i.e$ $x \in S_{\alpha}$ and $g$ is  usc on $\mathbb{R}^{m}$.
\end{proof}

\begin{lemma}
\label{getrightcont}
Let $S \subset \mathbb{R}$ be a closed subset of $\mathbb{R}$. Let $f: \mathbb{R} \to  \mathbb{R} \cup \{\pm \infty\}$ be such that $f$ is usc and non-decreasing on $S$. Then $f$ is right-continuous on $S$.
\end{lemma}
\begin{proof}
Let $(x_{n})_{n \geq 1} \subset S$ be a sequence converging to some $x^{*}$ from above. Then $x^{*} \in S$ since  $S$ is closed.   As $x \in S \to f(x)$ is non-decreasing, for all $n \geq 1$ we have that $f(x_{n}) \geq f(x^{*})$ and thus $\liminf_{n} f(x_{n}) \geq f(x^{*})$. Now as $f$ is usc on $S$, we get that $\limsup_{n} f(x_{n}) \leq f(x^{*})$. The right continuity of $f$ on $S$ follows immediately.
\end{proof}\\
We now establish a useful  extension  of Lemma \ref{completemes}.
\begin{lemma}																			
\label{completemes2}
Let $f:\Omega^t \times \mathbb{R} \to \mathbb{R} \cup \{\pm \infty\}$  be an $\overline{\Fc}_{t}\otimes \mathcal{B}(\mathbb{R})$-measurable function such that for all $\o^{t} \in \Omega^{t}$, $x \in \mathbb{R} \to f(\o^{t},x)$ is usc and non-decreasing. Then, there exists some ${\Fc}_{t}\otimes \mathcal{B}(\mathbb{R})$-measurable function $g$ from $\Omega^t \times \mathbb{R}$ to $\mathbb{R} \cup \{\pm \infty\}$ and some $\Omega^{t}_{mes} \in \mathcal{F}_{t}$ such that $P_{t}(\Omega^{t}_{mes})=1$ and $f(\o^{t},x)=g(\o^{t},x)$ for all $(\o^{t},x) \in \Omega^{t}_{mes}\times \mathbb{R}$.
\end{lemma}
\begin{remark}
In particular, for all $\o^{t} \in \Omega^{t}_{mes}$, $x \in \mathbb{R} \to g(\o^{t},x)$ is usc and non-decreasing.
\end{remark}
\begin{proof}
Let $n \geq 1$ and $k \in \mathbb{Z}$ be fixed. We apply Lemma \ref{completemes} to $f(\cdot)=f(\cdot,\frac{k}{2^{n}})$ that is $\overline{\mathcal{F}}_{t}$-measurable by assumption and we get some $\mathcal{F}_{t}$-measurable $g_{n,k}: \Omega^{t} \to \mathbb{R} \cup \{\pm \infty\}$  and  some $\Omega^{t}_{n,k} \in \mathcal{F}_{t}$ such that $P_{t}(\Omega^{t}_{n,k})=1$ and  $\Omega^{t}_{n,k}  \subset \left\{\o^{t} \in \Omega^{t},\; f(\o^{t},\frac{k}{2^{n}}) = g_{n,k}(\o^{t})\right\}$. We set
\begin{align}
\label{omegames}
\Omega^{t}_{mes}:= \bigcap_{n \geq 1, k \in \mathbb{Z}} \Omega^{t}_{n,k}.
\end{align}
It is clear  that $\Omega^{t}_{mes} \in \mathcal{F}_{t}$ and that $P_{t}(\Omega^{t}_{mes})=1$.\\
Now, we define for all $n \geq 1$, $g_{n}: \Omega^{t} \times \mathbb{R} \to \mathbb{R} \cup\{\pm \infty\}$ by
\begin{align*}
g_{n}(\o^{t},x):= \sum_{k \in \mathbb{Z}} 1_{\big(\frac{k-1}{2^{n}},\frac{k}{2^{n}}\big]}(x)g_{n,k}(\o^{t}).
\end{align*}
It is clear  that $g_{n}$ is ${\Fc}_{t}\otimes \mathcal{B}(\mathbb{R})$-measurable for all $n \geq 1$. Finally, we define $g: \Omega^{t} \times \mathbb{R} \to \mathbb{R} \cup\{\pm \infty\}$ by
\begin{align}
\label{defg}
g(\o^{t},x):= \lim_{n} g_{n}(\o^{t},x).
\end{align}
Then $g$ is again ${\Fc}_{t}\otimes \mathcal{B}(\mathbb{R})$-measurable and it remains to prove  that $f(\o^{t},x)=g(\o^{t},x)$  for all $(\o^{t},x) \in \Omega^{t}_{mes}\times \mathbb{R}$. Let $(\o^{t},x) \in \Omega^{t}_{mes}\times \mathbb{R}$ be fixed. For all $n \geq 1$, there exists $k_{n} \in \mathbb{Z}$ such that $\frac{k_n-1}{2^{n}} < x \leq \frac{k_n}{2^{n}}$ and such that  $g_{n}(\o^{t},x)=g_{n,k_n}(\o^{t})=f(\o^{t},\frac{k_n}{2^{n}})$. Applying Lemma \ref{getrightcont} to $f(\cdot)=f(\o^{t},\cdot)$ (and $S=\mathbb{R}$), we get that  $x \in \mathbb{R} \to f(\o^{t},x)$ is right-continuous on $\mathbb{R}$. As $\left(\frac{k_n}{2^{n}}\right)_{n \geq 1}$ converges to $x$ from above,  it follows that $g(\o^{t},x)= \lim_{n}  f(\o^{t},\frac{k_n}{2^{n}})=f(\o^{t},x)$ and this concludes the proof.
\end{proof}\\

Finally, we introduce the following definition.
\begin{definition}
\label{extcara}
Let  $S$ be a closed interval of $\mathbb{R}$. A function $f: \Omega^{t} \times S \to \mathbb{R}$ is an  extended Carath\'eodory function if
\begin{itemize}
\item[i)] for all $\o^{t} \in \Omega^{t}$, $x \in S \to f(\o^{t},x)$ is right-continuous,
\item[ii)] for all $x \in S$, $\o^{t} \in \Omega^{t} \to f(\o^{t},x)$ is  $\mathcal{F}_{t}$-measurable.
\end{itemize}
\end{definition}
And we prove the following lemma that is an extension of a well-know result on Carath\'eodory functions (see for example 4.10 in \cite{Hitch})
\begin{lemma}																			
\label{cara}
Let $S \subset \mathbb{R}$ be a closed interval of $\mathbb{R}$ and $f: \Omega^{t} \times S \to \mathbb{R}$ be  an extended Carath\'eodory function.  Then $f$ is $\mathcal{F}_{t} \otimes \mathcal{B}(\mathbb{R})$-measurable.
\end{lemma}
\begin{proof}
We define for all $n \geq 1$, $f_{n}:\Omega^{t} \times \mathbb{R} \to \mathbb{R}$ by
$$f_{n}(\o^{t},x):= \sum_{k \in \mathbb{Z}} 1_{\big(\frac{k-1}{2^{n}},\frac{k}{2^{n}}\big]}(x) 1_{S}(\frac{k}{2^{n}}) f(\o^{t},\frac{k}{2^{n}}) .$$
It is clear that $f_{n}$ is $\mathcal{F}_{t} \otimes \mathcal{B}(\mathbb{R})$-measurable. From the right continuity of $f$, we can show  as in the proof of Lemma \ref{completemes2} that  $f(\o^{t},x)=\lim_{n} f_{n}(\o^{t},x)$ for all $(\o^{t},x) \in \Omega^{t} \times S$ and the proof is complete (recall that  $ \O\times S \in \mathcal{F}_{t}\otimes \mathcal{B}(\mathbb{R})$ as $S$ is a closed subset of $\mathbb{R}$).\end{proof}
\begin{remark}
\label{overlineF}
Note that we have the same result if we replace $\mathcal{F}_{t}$ with $\overline{\mathcal{F}_{t}}$.
\end{remark}

\subsection{Proof of technical results}
\label{proofofres}
Finally, we provide the missing results and proofs of the paper.
We start with the following results from Section \ref{se2}.\\

\begin{proof}[of Lemma \ref{lienespcond}]
We refer to Section 6.1 of \cite{CR14} for the definition and various properties of generalized conditional expectations. In particular since $ E(h^{+})=\int_{\O^{t}} h^{+}dP_{t}<\infty$, $E(h| \mathcal{F}_{s})$ is well-defined (in the generalised sense) for all $0 \leq s \leq t$ (see Lemma 6.2 of \cite{CR14} ).  Similarly, from Proposition \ref{fubinirem} we have  that $\f: \O^{s} \to \mathbb{R} \cup \{\pm \infty\}$ is well-defined (in the generalised sense) and $\Fc_s$-measurable.\\
As $\f(X_1,\ldots,X_s)$ is $\Fc_s$-measurable, it remains  to prove that $E(gh)=E(g \f(X_1,\ldots,X_s))$ for all $g: \Omega^{s} \to \mathbb{R}_{+}$ non-negative, $\mathcal{F}_{s}$-measurable and such that $E(gh)$ is well-defined in the generalised sense, $i.e$ such that $E\left(gh\right)^{+}<\infty$ or $E\left(gh\right)^{-}<\infty$. Recalling the notations of the beginning of Section \ref{se2} and  using the Fubini Theorem for the third and fourth equality (see Proposition \ref{fubinirem} and Remark \ref{remneg}), we get that
\begin{small}
\begin{eqnarray*}
E(gh) & = & E(g(X_1,\ldots,X_s)h(X_1,\ldots,X_t))=\int_{\O^T} g(\o_1,\ldots,\o_s)h(\o_1,\ldots,\o_t)P(d \o^T)\\
 &= & \int_{\O^t} g(\o_1,\ldots,\o_s)h(\o_1,\ldots,\o_t)q_{t}(\o_{t}|\o^{t-1}) \ldots q_{s+1}(\o_{s+1}|\o^s)P_s(d \o^s)\\
  &= & \int_{\O^s} g(\o_1,\ldots,\o_s)\left(\int_{\O_{s+1}\times \ldots \times \O_t} h(\o_1,\ldots,\o_s,\o_{s+1},\ldots,\o_t)q_{t}(\o_{t}|\o^{t-1}) \ldots q_{s+1}(\o_{s+1}|\o^s)\right)P_s(d \o^s)\\
   &= & \int_{\O^s} g(\o_1,\ldots,\o_s)\f(\o_1,\ldots,\o_s)P_s(d \o^s)\\
      & = & E(g(X_1,\ldots,X_s)\f(X_1,\ldots,X_t)),
\end{eqnarray*}
\end{small}
which concludes the proof.
\end{proof}\\

 We give now the proof of results of Section  \ref{secna}.\\
\begin{proof}[of Lemma \ref{Dmeasurability}]
We  first prove that $\widetilde{D}^{t+1}$ is a non-empty, closed-valued and  $\mathcal{F}_{t}$-measurable random set. It is clear from its  definition (see \eqref{defd1}) that for all $\o^{t} \in \O^{t}$,  $\widetilde{D}^{t+1}(\o^{t})$ is a non-empty and closed subset of $\mathbb{R}^{d}$. We now show that $\widetilde{D}^{t+1}$ is  measurable. Let $O$ be a fixed open set in $\mathbb{R}^{d}$ and introduce
\begin{eqnarray*}
\mu_{O}: \omega^{t} \in \O^{t} \to \mu_{O}(\o^{t}) &:= & q_{t+1}\left(\Delta S_{t+1}(\o^{t},.) \in O|\o^{t}\right)\\
 &= & \int_{\O_{t+1}}1_{\Delta S_{t+1}(\cdot,\cdot) \in O}(\o^{t},\o_{t+1})q_{t+1}(d\o_{t+1}|\o^{t}).
\end{eqnarray*}
We prove  that  $\mu_{O}$ is $\mathcal{F}_{t}$-measurable.
 As $(\o^{t},\o_{t+1}) \in \O^{t}\times \O_{t+1}  \to \Delta S_{t+1}(\o^{t},\o_{t+1})$ is $\mathcal{F}_{t}\otimes \mathcal{G}_{t+1}$-measurable and  $O \in \mathcal{B}(\mathbb{R}^{d})$,  $(\o^{t},\o_{t+1}) \to 1_{\Delta S_{t+1}(\cdot,\cdot) \in O}(\o^{t},\o_{t+1})$ is $\mathcal{F}_{t}\otimes \mathcal{G}_{t+1}$-measurable and the result follows from Proposition \ref{LemmaA1}.\\
By definition of  $\widetilde{D}^{t+1}(\o^{t})$ we get that
$$\{\o^{t} \in \Omega^t,\; \widetilde{D}^{t+1}(\o^{t}) \cap O \neq \emptyset \} =\{\o^{t} \in \Omega^t, \; \mu_{O}(\omega^{t})>0\} \in \mathcal{F}_{t}.$$
Next we prove that $D^{t+1}$ is a non-empty, closed-valued and  $\mathcal{F}_{t}$-measurable random set. Using \eqref{defd2}, $D^{t+1}$ is  a non-empty and  closed-valued random set. It remains to  prove that $D^{t+1}$ is $\mathcal{F}_{t}$-measurable. As $\widetilde{D}^{t+1}$ is  $\mathcal{F}_{t}$-measurable, applying the Castaing representation   (see Theorem 2.3
in Chapter 1 of \cite{Molchanov} or Theorem 14.5 of \cite{rw}), 
we obtain a countable family of  $\mathcal{F}_{t}$-measurable  functions $(f_{n})_{n \geq 1}:\Omega^t \to \mathbb{R}^{d}$ such that for all $\o^{t} \in \Omega^{t}$, $\widetilde{D}^{t+1}(\o^t)= \overline{\left\{f_{n}(\o^{t}),\; n\geq 1\right\}}$ (where the closure is taken in $\mathbb{R}^{d}$ with respect to the usual topology). Let $\o^{t} \in \O^{t}$ be fixed. It can be easily shown that
\begin{align}
D^{t+1}(\o^{t})= \mbox{Aff} (\widetilde{D}^{t+1}(\o^t)) =\overline{\left\{f_{1}(\o^{t})+ \sum_{i=2}^{p} \lambda_{i} (f_{i}(\o^{t})-f_{1}(\o^{t})),\; (\lambda_{2},\dots, \lambda_{p}) \in \mathbb{Q}^{p-1},\; p\geq 2\right\}}.
\end{align}
So, using again the Castaing representation (see Theorem 14.5 of \cite{rw}), we obtain that $D^{t+1}(\o^{t})$ is $\mathcal{F}_{t}$-measurable.
From Theorem 14.8 of \cite{rw}, $Graph(D^{t+1}) \in \mathcal{F}_{t} \otimes \mathcal{B}(\mathbb{R}^{d})$ (recall that $D^{t+1}$ is closed-valued).
\end{proof}\\

\begin{proof} [of Lemma \ref{localVectorSpace}]
Introduce $C^{t+1}(\o^{t}):=\overline{\mbox{Conv}} (\widetilde{D}^{t+1}(\o^{t}))$ the closed convex hull generated by $\widetilde{D}^{t+1}(\o^{t})$.  As $C^{t+1}(\o^{t}) \subset D^{t+1}(\o^{t})$ we will prove that $0 \in C^{t+1}(\o^{t})$.  Since $C^{t+1}(\o^{t}) \subset D^{t+1}(\o^{t})$ by assumption,  for all $h \in C^{t+1}(\o^{t}) \backslash\{0\}$
\begin{align}
\label{h0}
q_{t+1}(h\Delta S_{t+1}(\o^t,\cdot)\geq 0|\o^t)<1 .
\end{align}
Thus if we find some $h_0 \in C^{t+1}(\o^{t})$ such that $q_{t+1}(h_0\Delta S_{t+1}(\o^t,\cdot)\geq 0|\o^t)=1$ then $h_0=0$. We distinguish two cases.
First assume that for all $h \in \mathbb{R}^{d}$, $h \neq 0$, $q_{t+1}(h\Delta S_{t+1}(\o^{t},.) \geq 0|\o^{t})<1$. Then the polar cone of $C^{t+1}(\o^{t})$, $i.e$ the set
$$\left(C^{t+1}(\o^{t})\right)^{\circ}:=\{y \in \mathbb{R}^{d},\; yx\leq 0, \,  \forall \, x \in C^{t+1}(\o^{t})\}$$
is reduced to $\{0\}$. Indeed if this is not the case there exists $y_{0} \in \mathbb{R}^{d}$ such that  $-y_0x\geq 0$ for all $x \in C^{t+1}(\o^{t})$. As $A:=\{\o_{t+1} \in \Omega_{t+1},\; \Delta S_{t+1}(\o^{t},\o_{t+1}) \in \widetilde{D}^{t+1}(\o^{t})\} \subset\{\o_{t+1} \in \Omega_{t+1},\; -y_{0}\Delta S_{t+1}(\o^{t},\o_{t+1})\geq 0\}$ and $q_{t+1}(A|\o^{t})=1$ we obtain that $q_{t+1}(-y_0 \Delta S_{t+1}(\o^{t},\cdot) \geq 0|\o^{t})=1$ a contradiction.
As $\left(\left(C^{t+1}(\o^{t})\right)^{\circ}\right)^{\circ}= \mbox{cone}\left(C^{t+1}(\o^{t})\right)$ where $ \mbox{cone}\left(C^{t+1}(\o^{t})\right)$ denote the cone generated by $C^{t+1}(\o^{t})$ we get that $ \mbox{cone}\left(C^{t+1}(\o^{t})\right)=\mathbb{R}^{d}$. Let $u \neq 0  \in  \mbox{cone}\left(C^{t+1}(\o^{t})\right)$  then $-u \in  \mbox{cone}\left(C^{t+1}(\o^{t})\right)$  and there exist $\lambda_{1}>0$, $\lambda_{2}>0$ and $v_{1},\; v_{2}  \in C^{t+1}(\o^{t})$ such that $u=\lambda_{1} v_{1}$ and $-u=\lambda_{2} v_{2}$. Thus $0=\frac{\l_1}{\l_1+ \l_2} v_{1} + \frac{\l_2}{\l_1+ \l_2} v_{2} \in C^{t+1}(\o^{t})$ by convexity of  $C^{t+1}(\o^{t})$. \\
Now we assume that there exists some $h _{0}\in \mathbb{R}^{d}$, $h_{0} \neq 0$ such that  $q_{t+1}(h_0\Delta S_{t+1}(\o^{t},.) \geq 0|\o^{t})=1$. Note that since $h_{0} \in \mathbb{R}^{d}$ we cannot use \eqref{h0}. Introduce the orthogonal projection on $C^{t+1}(\o^{t})$ (recall that $C^{t+1}(\o^{t})$ is a closed convex subset of $\mathbb{R}^{d}$) $$p: h \in \mathbb{R}^{d} \to p(h) \in C^{t+1}(\o^{t}).$$ Then $p$ is continuous and we have $\left(h-p(h)\right)\left(x-p(h)\right) \leq 0$ for all $x \in C^{t+1}(\o^{t})$. Fix $\o_{t+1} \in \{\o_{t+1} \in \O_{t+1},\; \Delta S_{t+1}(\o^{t},\o_{t+1}) \in \widetilde{D}^{t+1}(\o^{t})\} \cap \{\o_{t+1} \in \O_{t+1},\; h_{0}\Delta S_{t+1}(\o^{t},\o_{t+1})\geq 0\} $ and $\lambda \geq 0$. Let $h=\lambda h_{0}$ and $x= \Delta S_{t+1}(\o^{t},\o_{t+1}) \in C^{t+1}(\o^{t})$ in the previous equation, we obtain (recall that  $\widetilde{D}^{t+1}(\o^{t}) \subset C^{t+1}(\o^{t})$) \begin{align*}
0 \leq \lambda h_{0} \Delta S_{t+1}(\o^{t},\o_{t+1}) &= \left(\lambda h_0-p(\lambda h_0)\right) \Delta S_{t+1}(\o^{t},\o_{t+1}) + p(\lambda h_0)  \Delta S_{t+1}(\o^{t},\o_{t+1}) \\
&\leq \left(\lambda h_0-p(\lambda h_0)\right)p(\lambda h_0) +p(\lambda h_0)  \Delta S_{t+1}(\o^{t},\o_{t+1}).
\end{align*}
As this is true for all $\lambda \geq 0$ we may take the limit when $\lambda$ goes to zero and use the continuity of $p$
$$ p(0)  \Delta S_{t+1}(\o^{t},\o_{t+1}) \geq  |p(0)|^{2} \geq 0$$
As $q_{t+1}\left( \left\{\o_{t+1} \in \O_{t+1},\; \Delta S_{t+1}(\o^{t},\o_{t+1}) \in \widetilde{D}^{t+1}(\o^{t})\right\}|\o^t\right)=1$ by definition of $\widetilde{D}^{t+1}(\o^{t})$ and as \\ $q_{t+1}(h_0\Delta S_{t+1}(\o^{t},.) \geq 0|\o^{t})=1$ as well we have obtained that $$q_{t+1}( p(0)  \Delta S_{t+1}(\o^{t},\cdot) \geq  0|\o^{t})=1.$$
The fact that  $p(0) \in  C^{t+1}(\o^{t})$ together with \eqref{h0} implies that $p(0)=0$ and  $0\in C^{t+1}(\o^{t})$ follows.\\
\end{proof}\\

The following lemma has been used in the proof of Lemma \ref{localNA}. It corresponds to Lemma 2.5 of \cite{Nutz}
\begin{lemma}
\label{CfNutz}
Let $\omega^{t} \in \Omega^{t}$ be fixed. Recall that
$L^{t+1}(\omega^{t}):=\left(D^{t+1}(\omega^{t})\right)^{\bot}$ is the orthogonal space of $D^{t+1}(\omega^{t})$ (see \eqref{lt}). Then for $h \in \mathbb{R}^{d}$ we have that
$$ q_{t+1}(h \Delta S_{t+1}(\omega^{t},\cdot)=0|\omega^{t})=1 \iff h \in L^{t+1}(\omega^{t}).$$
\end{lemma}
\begin{proof}
Assume that $h \in L^{t+1}(\omega^{t})$. Then $\{\omega \in \O_t,\;  \Delta S_{t+1}(\omega^{t},\omega) \in D^{t+1}(\omega^t)\} \subset \{\omega \in \O_t, h \Delta S_{t+1}(\omega^{t},\omega)=0\}$. As by definition of $D^{t+1}(\o^{t})$, $q_{t+1}(\Delta S_{t+1}(\omega^{t},.) \in D^{t+1}(\omega^{t})|\o^{t})=1$, we conclude that $q_{t+1}(h \Delta S_{t+1}(\omega^{t},.)=0|\omega^{t})=1.$
Conversely, we assume that  $ h \notin L^{t+1}(\omega^{t})$ and we show that $q_{t+1}(h \Delta S_{t+1}(\omega^{t},.)=0|\omega^{t})<1$.
We first show that there exists $v \in \widetilde{D}^{t+1}(\omega^{t})$ such that $hv \neq 0$. If not, for all $v \in \widetilde{D}^{t+1}(\omega^{t})$,  $hv = 0$
and for any  $w \in{D}^{t+1}(\omega^{t})$ with $w=\sum_{i=1}^m \l_i v_i$ where $\l_i \in \mathbb{R}$,
$\sum_{i=1}^m \l_i=1$ and $v_i  \in \widetilde{D}^{t+1}(\omega^{t})$, we get that $hw = 0$, a contradiction.
Furthermore there exists an open ball centered in $v$ with radius $\varepsilon>0$, $B(v,\varepsilon)$,  such that $hv' \neq 0$ for all $v' \in B(v,\varepsilon)$. Assume that  $q_{t+1}( \Delta S_{t+1}(\omega^{t},.) \in B(v,\varepsilon)|\o^{t})=0$ or equivalently that
$q_{t+1}( \Delta S_{t+1}(\omega^{t},.) \in \mathbb{R}^d \setminus B(v,\varepsilon)|\o^{t})=1$. By definition of the support,
$ \widetilde{D}^{t+1}(\o^{t}) \subset \mathbb{R}^d \setminus B(v,\varepsilon)$: this contradicts $v \in \widetilde{D}^{t+1}(\omega^{t})$. Therefore
$q_{t+1}( \Delta S_{t+1}(\omega^{t},.) \in B(v,\varepsilon)|\o^{t})>0$. Let $\o \in  \{ \Delta S_{t+1}(\omega^{t},.) \in B(v,\varepsilon)\}$, then
$h\Delta S_{t+1}(\omega^{t},\o) \neq 0$ $i.e$ $q_{t+1}(h \Delta S_{t+1}(\omega^{t},.)=0|\omega^{t}))<1$.
\end{proof}\\

We prove now the following result of Section 5.\\
\begin{proof}[of Proposition \ref{fat}]
We start with the proof of  \eqref{Vplus} when $h \in D_{x}$.
Since $D$ is a vectorial subspace of $\mathbb{R}^d$ and $0 \in \mathcal{H}_{x}$, the affine hull of $D_{x}$ is also a vector space that we denote  by $\mbox{Aff} (D_x) $.
If $x \leq 1$ we have by Assumption \ref{samedi} that for all $\omega \in \overline{\Omega}$, $h \in D_{x}$,
\begin{eqnarray}
\label{inf1}
V^{+}(\o,x+hY(\o)) \leq V^{+}\left (\o,1+ h Y(\o)\right).
\end{eqnarray}
If $x>1$ using Assumption \ref{ae1} (see \eqref{elasticplus} in Remark \ref{remext}) we get that
for all $\omega \in \overline{\Omega}$, $h \in D_{x}$
\begin{eqnarray}
\label{sup1}
V^{+}(\o,x+hY(\o))=V^{+}\left(2x\left(\frac{1}{2}+\frac{h}{2x}Y(\o)\right)\right) \leq  (2x)^{\overline{\gamma}}K\left(V^{+}\left(\o,1+ \frac{h}{2x} Y(\o)\right)+ C(\o)\right).
\end{eqnarray}
First we treat the case of $Dim (\mbox{Aff} (D_x) )=0$, $i.e$ $D_{x}=\{0\}$.
For all $\omega \in \overline{\Omega}$, $h \in D_{x}=\{0\}$, using \eqref{inf1} and \eqref{sup1}, we obtain that
\begin{align}
\label{Ineq1}
V^{+}(\o,x+hY(\o)) \leq V^{+}(\o,1) + (2x)^{\overline{\gamma}}K\left(V^{+}\left(\o,1\right)+ C(\o)\right) \leq ((2x)^{\overline{\gamma}}K+1)(V^{+}(\o,1)+C(\o)).
\end{align}
We assume now that $Dim(\mbox{Aff} (D_x) )>0$. If $x=0$ then $Y=0$ $Q$-a.s. If  this is not the case then we should have $D_{0}=\{0\}$ a contradiction. Indeed if there exists some $h \in D_{0}$ with $h \neq 0$,  then $Q\left(\frac{h}{|h|}Y(\cdot)<0\right)>0$ by Assumption \ref{AOAone} which contradicts $h \in D_{0}$. So for $x=0$, $Y=0$ $Q$-a.s and by Assumption \ref{samedi} we get that for all $\omega \in \overline{\Omega}$, $h \in D_{0}$,
$$V^{+}(\o,0+hY(\o)) \leq V^{+}(\o,1).$$
From now we assume  that $x>0$. Then as for $g \in \mathbb{R}^{d}$, $g \in D_{x}$ if and only if $\frac{g}{x} \in D_{1}$, we have that $\mbox{Aff} (D_x) =\mbox{Aff} (D_1)$. We set $d':=Dim(\mbox{Aff} (D_1))$.  Let $(e_{1},\dots,e_{d'})$ be an orthonormal basis of $\mbox{Aff} (D_1)$ (which is a sub-vector space of $\mathbb{R}^{d}$) and $\varphi$ $:$ $(\lambda_{1},\dots ,\lambda_{d'})\in \mathbb{R}^{d'} \to \Sigma_{i=1}^{d'} \lambda_{i} e_{i} \in \mbox{Aff} (D_1) $. Then $\varphi$ is an isomorphism (recall that $(e_{1},\dots,e_{d'})$ is a basis of $\mbox{Aff} (D_1) $). As $\varphi$ is linear and the spaces considered are of finite dimension, it is also an homeomorphism between $\mathbb{R}^{d'}$ and $\mbox{Aff} (D_1)$. Since  $D_{1}$ is compact by Lemma \ref{rast}, $\varphi^{-1}(D_{1})$ is a compact subspace of $\mathbb{R}^{d'}$ . So there exists some $c \geq 0$ such that for all $h=\Sigma_{i=1}^{d'} \lambda_{i} e_{i} \in D_1$, $|\lambda_{i}| \leq c$ for all $i=1,\ldots,d'$.
We complete the family of vector $(e_{1},\dots,e_{d'})$ in order to obtain an  orthonormal  basis of $\mathbb{R}^{d}$, denoted by $(e_{1},\dots,e_{d'}, e_{d'+1},\dots e_{d})$. For all $\o \in \O$, let $(y_{i}(\o))_{i=1,\dots,d}$ be the coordinate of $Y(\o)$ in this basis.\\
Now let $h\in D_x$ be fixed.  Then $\frac{h}{2x} \in  D_{\frac{1}{2}} \subset D_1$ and $\frac{h}{2x}=\Sigma_{i=1}^{d'} \lambda_{i} e_{i}$ for some $(\lambda_{1},\dots \lambda_{d'})\in \mathbb{R}^{d'}$ with $|\lambda_{i}| \leq c$ for all $i=1,\ldots,d'$. Note that as $\frac{h}{2x} \in D_1$, $\lambda_{i}=0$ for $i \geq d'+1$. Then as $(e_{1},\dots,e_{d})$ is an orthonormal basis of $\mathbb{R}^d$, we obtain for all $\omega \in \overline{\Omega}$
\begin{align*}
1 +\frac{h}{2x}  Y(\o) &= 1 + \Sigma_{i=1}^{d'} \lambda_{i}y_{i}(\o)\\
& \leq 1 + \Sigma_{i=1}^{d'} |\lambda_{i}| |y_{i}(\o)|\\
&\leq  1+ c \Sigma_{i=1}^{d'} |y_{i}(\o)|.
\end{align*}
Thus from Assumption \ref{samedi} for all $\omega \in \overline{\Omega}$ we get that
\begin{align*}
V^{+}\left(\o,1+\frac{h}{2x} Y(\o)\right)&\leq V^{+}\left(\o,1+ c \Sigma_{i=1}^{d'} |y_{i}(\o)|\right).
\end{align*}
We set
$$L(\cdot):=  V^{+}\left(\o,1+ c \Sigma_{i=1}^{d'} |y_{i}(\o)|\right)1_{d'>0}+ V^{+}(\cdot,1) + C(\cdot).$$
As $d'=Dim (\mbox{Aff} (D_1))$ it is clear  that  $L$ does not  depend on $x$. It is also clear that $L$ is $\Hc$-measurable.\\
Then using  \eqref{inf1}, \eqref{sup1} and \eqref{Ineq1} we obtain that for all $\omega \in \overline{\Omega}$
$$V^{+}(\o,x+hY(\o))  \leq ((2x)^{\overline{\gamma}}K+1) L(\o).$$
Note that the first term in $L$ is used in the above inequality if $x \neq 0$ and $Dim(\mbox{Aff} (D_x) )>0$. The second and the third one are there for both  the case of $Dim(\mbox{Aff} (D_x) )=0$ and the case of $x=0$ and $Dim(\mbox{Aff} (D_x) )>0$.
As by Assumptions \ref{ae1} and \ref{dimanche}, $E(V^{+}(\cdot,1)+C(\cdot)) < \infty$,  it remains to prove that $d'>0$ implies
$E \left(V^{+}\left(\cdot,1+ c \Sigma_{i=1}^{d'} |y_{i}(\cdot)|\right)\right)<\infty$.\\
Introduce $W$, the finite set of $\mathbb{R}^d$  whose coordinates on $(e_{1},\dots,e_{d'})$  are $1$ or $-1$  and $0$ on $(e_{d'+1},\dots e_{d})$. Then  $W \subset \mbox{Aff} (D_1)$ and the vectors of $W$ will be denoted by $\theta^{j}$ for $j \in \{1,\dots, 2^{d'}\}$.
Let ${\theta}^{\omega}$ be the vector whose coordinates on $(e_{1},\dots,e_{d'})$ are $(sign(y_{i}(\omega)))_{i=1\dots d'}$ and $0$ on  $(e_{d'+1},\dots e_{d})$. Then $\theta^{\omega} \in W$ and we get that
\begin{align*}
 & V^{+}\left(\o,1+c \Sigma_{i=1}^{d'}  |y_{i}(\o)|\right) =
V^{+}(\o,1+ c\theta^{\o} Y(\o)) \leq  \sum_{j=1}^{2^{d'}} V^{+}(\o,1+ c \theta^{j} Y(\o)).
\end{align*}
So to prove that $E L <\infty$ it is sufficient to prove that if $d'>0$  for all $1\leq j \leq 2^{d'}$, $E V^{+}(\cdot,1+ c \theta^{j}Y(\cdot))<\infty.$
Recall that $\theta^{j} \in \mbox{Aff} (D_1)$. Let $ri(D_{1})=\{y \in D_{1}, \, \exists \alpha>0 \, s.t \;  \mbox{Aff} (D_1) \cap B(y,\alpha) \subset D_1\}$ \footnote{Here $B(y,\alpha)$ is the ball  of $\mathbb{R}^{d}$ centered at $y$ and with radius $\alpha$.} denote the relative interior of $D_{1}$. As $D_{1}$ is convex and non-empty (recall $d'>0$), $ri(D_{1})$ is also non-empty and convex and we fix some $e^{*} \in ri(D_{1})$. We prove that $\frac{e^{*}}{2} \in ri(D_{1})$. Let $\alpha>0$ be  such that $\mbox{Aff} (D_1) \cap B(e^*,\alpha) \subset D_1$ and  $g \in \mbox{Aff} (D_1) \cap B(\frac{e^*}{2},\frac{\alpha}{2})$. Then $2g  \in \mbox{Aff} (D_1) \cap B({e^*},{\alpha})$ (recall that $ \mbox{Aff} (D_1)$ is actually a vector space) and thus $2g \in D_1$. As $D_{1}$ is convex and $0 \in D_{1}$, we get that $g \in  D_1$ and $\mbox{Aff} (D_1) \cap B(\frac{e^*}{2},\frac{\alpha}{2}) \subset D_1$ which proves that $\frac{e^{*}}{2} \in ri(D_{1})$.  Now let $\varepsilon_{j} $ be such that $\varepsilon_{j}(\frac{c}{2} \theta^{j}-\frac{e^{*}}{2})\in B(0,\frac{\alpha}{2})$. It is easy to see that one can chose $\varepsilon_{j} \in (0,1)$. Then as $\bar{e}^{j}:=\frac{e^{*}}{2}+ \frac{\varepsilon_{j}}{2}(c \theta^{j}-e^{*}) \in \mbox{Aff} (D_1)\cap B(\frac{e^*}{2},\frac{\alpha}{2})$ (recall that $\theta^{j} \in W \subset \mbox{Aff} (D_1)$), we deduce that $\bar{e}^{j} \in D_{1}$. Using \eqref{elasticplus} we obtain that for $Q$-almost all $\o$
\begin{align*}
V^{+}(\o,1+c \theta^{j}Y(\o))&=V^{+}(\o,1+e^{*}Y(\o)+(c \theta^{j}-e^{*})Y(\o))\\
& \leq \left(\frac{2}{\varepsilon_{j}}\right)^{\overline{\gamma}}K\left[V^{+}\left(\o,\frac{\varepsilon_{j}}{2}(1+e^{*}Y(\o))+\frac{\varepsilon_{j}}{2}(c \theta^{j}-e^{*})Y(\o)+ \frac{1}{2}\right)+C(\o)\right]\\
& \leq \left(\frac{2}{\varepsilon_{j}}\right)^{\overline{\gamma}}K\left[V^{+}\left(\o,\frac{1}{2}+\frac{e^{*}}{2}Y(\o)+\frac{\varepsilon_{j}}{2}(c \theta^{j}-e^{*})Y(\o)+ \frac{1}{2}\right)+C(\o)\right]\\
& \leq \left(\frac{2}{\varepsilon_{j}}\right)^{\overline{\gamma}}K\left[V^{+}(\o,1+\bar{e}^{j}Y(\o))+C(\o))\right],
\end{align*}
where the second inequality follows from the fact that $1+e^{*}Y(\cdot) \geq 0$ $Q$-a.s  (recall that $e^{*} \in ri(D_{1})$) and the monotonicity property of $V$ in Assumption \ref{utilite}. Note that the above inequalities are true even if $1+c \theta^{j}Y(\o)<0$ since \eqref{elasticplus} (see remark \ref{remext}) and the monotonicity property of $V$ hold true for all $x \in \mathbb{R}$.\\
From Assumption \ref{dimanche} we get that $E V^{+}(\cdot,1+\bar{e}^{j}Y(\cdot)) <\infty$ (recall that $\bar{e}_{j} \in D_{1}$) and
Assumption \ref{ae1} implies $E C<\infty$, therefore $E V^{+}(\cdot,1+ c \theta^{j}Y(\cdot)) <\infty$ and \eqref{Vplus} is proven for $h \in D_{x}$. Now let $h \in \mathcal{H}_{x}$ and $h'$ its orthogonal projection on $D$, then $h Y(\cdot)=h'Y(\cdot)$ $Q$-a.s (see Remark \ref{proj}). It is clear that $h' \in D_{x}$ thus $V^{+}(\cdot,x+hY(\cdot))=V^{+}(\cdot,x+h'Y(\cdot))$ $Q$-a.s and \eqref{Vplus} is true also for $h \in \mathcal{H}_{x}$.
\end{proof}\\

To conclude,  the following lemma was used in the proof of Theorem \ref{main}.
\begin{lemma}
\label{AOAAPP}
Assume that (NA) holds true. Let $\phi \in \Phi$ such that $V_T^{x,\phi} \geq 0$ $P$-a.s, then
$V_t^{x,\phi} \geq 0$ $P_t$-a.s.
\end{lemma}
\begin{proof}
Assume that there is some $t$ such that $P_t(V_{t}^{{x,\phi}} \geq 0)<1$ or equivalently
$P_t(V_{t}^{{x,\phi}} < 0)>0$ and   let  $n=\sup\{t | P_t(V_{t}^{{x,\phi}}<0)>0\}$.
Then $P_n(V_{n}^{x,\phi}<0)>0$ and for all $s \geq n+1$, $P_s(V_{s}^{x,\phi} \geq 0)=1$.
Let $\Psi_s(\o)=0$ if  $s\leq n$ and  $\Psi_s(\o)=1_A \phi_s(\o)$ if $s\geq n+1$  with
$A=\{V_{n}^{\Phi}<0\}$. Then
\begin{eqnarray*}
{V}^{0,\Psi}_s  =  \sum_{k=1}^s \Psi_s \Delta  S_s =\sum_{k=n+1}^s \Psi_s \Delta  S_s
 =  1_A\left({V}^{x,\phi}_s-{V}^{x,\phi}_n\right)
\end{eqnarray*}
If $s\geq n+1$ $P_s(V_{s}^{x,\phi}\geq 0)=1$ and on $A$, $-{V}^{\Phi}_n>0$ thus $P_T(V_{T}^{0,\Psi}  \geq 0)=1$ and $V_{T}^{0,\Psi}  > 0$ on $A$.
As by the (usual) Fubini Theorem $P_T(A)=P_n(V_{n}^{x,\phi}<0)>0$, we get an arbitrage opportunity. Thus for all $t \leq T$, $P_t(V_{t}^{{x,\phi}} \geq 0)=1$.
\end{proof}\\

\section*{Acknowledgments.}
L. Carassus thanks LPMA (UMR 7599) for support. 
M. R\'asonyi was supported by the ``Lend\"ulet" grant LP2015-6 of the Hungarian Academy of Sciences.


\begin{thebibliography}{19}
\providecommand{\natexlab}[1]{#1}
\providecommand{\url}[1]{\texttt{#1}}
\expandafter\ifx\csname urlstyle\endcsname\relax
  \providecommand{\doi}[1]{doi: #1}\else
  \providecommand{\doi}{doi: \begingroup \urlstyle{rm}\Url}\fi

\bibitem[Aliprantis and Border(2006)]{Hitch}
C.~D. Aliprantis and K.~C. Border.
\newblock \emph{Infinite Dimensional Analysis~: A Hitchhiker's Guide}.
\newblock Grundlehren der Mathematischen Wissenschaften [Fundamental Principles
  of Mathematical Sciences]. Springer-Verlag, Berlin, 3rd edition, 2006.

\bibitem[Bensoussan et~al.(2015)Bensoussan, Cadenillas, and Koo]{bensoussan}
Alain Bensoussan, Abel Cadenillas, and Hyeng~Keun Koo.
\newblock Entrepreneurial decisions on effort and project with a nonconcave
  objective function.
\newblock \emph{Math. Oper. Res.}, 40\penalty0 (4):\penalty0 902--914, 2015.
\newblock \doi{10.1287/moor.2014.0702}.
\newblock URL \url{http://dx.doi.org/10.1287/moor.2014.0702}.

\bibitem[Bogachev(2007)]{boga}
V.I. Bogachev.
\newblock \emph{Measure Theory, vol 2}.
\newblock Springer-Verlag, Berlin, 2007.

\bibitem[Carassus and R\'asonyi(2015)]{CR14}
L.~Carassus and M.~R\'asonyi.
\newblock Maximization of non-concave utility functions in discrete-time
  financial market.
\newblock \emph{Mathematics of Operations Research, Published Online ISSN
  1526-5471}, 2015.

\bibitem[Carassus et~al.(2015)Carassus, M., and Rodrigues]{CRR}
L.~Carassus, R\'asonyi M., and A.~M. Rodrigues.
\newblock Non-concave utility maximisation on the positive real axis in
  discrete time.
\newblock \emph{Mathematics and Financial Economics}, 9\penalty0 (4):\penalty0
  325--349, 2015.

\bibitem[Carlier and Dana(2011)]{cd}
G.~Carlier and R.-A. Dana.
\newblock Optimal demand for contingent claims when agents have law invariant
  utilities.
\newblock \emph{Math. Finance}, 21:\penalty0 169--201, 2011.

\bibitem[Castaing and Valadier(1977)]{CV77}
C.~Castaing and M.~Valadier.
\newblock \emph{Convex Analysis and Measurable Multifunctions}, volume 580.
\newblock Springer, Berlin, 1977.

\bibitem[Dalang et~al.(1990)Dalang, Morton, and Willinger]{dmw}
R.C. Dalang, A.~Morton, and W.~Willinger.
\newblock Equivalent martingale measures and no-arbitrage in stochastic
  securities market models.
\newblock \emph{Stochastics Stochastics Rep.}, 29:\penalty0 185--201, 1990.

\bibitem[Dellacherie and Meyer(1979)]{dm}
C.~Dellacherie and P.-A. Meyer.
\newblock \emph{Probability and potential}.
\newblock North-Holland, Amsterdam, 1979.

\bibitem[F{\"o}llmer and Schied(2002)]{fs}
H.~F{\"o}llmer and A.~Schied.
\newblock \emph{Stochastic Finance: An Introduction in Discrete Time}.
\newblock Walter de Gruyter \& Co., Berlin, 2002.

\bibitem[He and Zhou(2011)]{hz}
X.~He and X.Y. Zhou.
\newblock Portfolio choice under cumulative prospect theory: An analytical
  treatment.
\newblock \emph{Management Science}, 57:\penalty0 315--331, 2011.

\bibitem[Jacod and Shiryaev(1998)]{JS98}
J.~Jacod and A.~N. Shiryaev.
\newblock Local martingales and the fundamental asset pricing theorems in the
  discrete-time case.
\newblock \emph{Finance Stoch.}, 2:\penalty0 259--273, 1998.

\bibitem[Jin and Zhou(2008)]{jz}
H.~Jin and X.Y. Zhou.
\newblock Behavioural portfolio selection in continuous time.
\newblock \emph{Math. Finance}, 18:\penalty0 385--426, 2008.

\bibitem[Kramkov and Schachermayer(1999)]{KS99}
D.~O. Kramkov and W.~Schachermayer.
\newblock The asymptotic elasticity of utility functions and optimal investment
  in incomplete markets.
\newblock \emph{Ann. Appl. Probab.}, 9:\penalty0 904--950, 1999.

\bibitem[Molchanov(2005)]{Molchanov} I.~Molchanov.
\newblock \emph{Theory of random sets}.
\newblock Springer-Verlag, London, 2005.

\bibitem[Nutz(2014)]{Nutz}
M.~Nutz.
\newblock Utility maximisation under model uncertainty in discrete time.
\newblock \emph{Math. Finance. Published Online DOI: 10.1111/mafi.12068}, 2014.

\bibitem[R\'asonyi and Stettner(2005)]{RS05}
M.~R\'asonyi and L.~Stettner.
\newblock On the utility maximization problem in discrete-time financial market
  models.
\newblock \emph{Ann. Appl. Probab.}, 15:\penalty0 1367--1395, 2005.

\bibitem[R\'asonyi and Stettner(2006)]{RS06}
M.~R\'asonyi and L.~Stettner.
\newblock On the existence of optimal portfolios for the utility maximization
  problem in discrete time financial models.
\newblock \emph{In: Kabanov, Y.; Lipster, R.; Stoyanov,J. (Eds), From
  Stochastic Calculus to Mathematical Finance, Springer.}, pages 589--608,
  2006.

\bibitem[Rockafellar and Wets(1998)]{rw}
R.~T. Rockafellar and R.~J.-B. Wets.
\newblock \emph{Variational analysis}, volume 317 of \emph{Grundlehren der
  Mathematischen Wissenschaften [Fundamental Principles of Mathematical
  Sciences]}.
\newblock Springer-Verlag, Berlin, 1998.
\newblock ISBN 3-540-62772-3.

\bibitem[Sainte-Beuve(1974)]{bv}
M.-F. Sainte-Beuve.
\newblock On the extension of von {N}eumann-{A}umann's theorem.
\newblock \emph{J.\ Functional Analysis}, 17\penalty0 (1):\penalty0 112--129,
  1974.

\end{thebibliography}

\end{document}